\documentclass[reqno]{amsart}
\usepackage{tkz-euclide}
\usepackage{hyperref}
\usepackage{amsmath}  
\usepackage{amssymb} 
\usepackage{mathrsfs} 
\usepackage{graphicx}  
\usepackage{tikz} 
\usetikzlibrary{calc} 
\usepackage{xcolor}  
\usetikzlibrary{arrows.meta} 
\usepackage{amsthm} 
\usepackage{float}
\usepackage{enumitem}
\usepackage[english]{babel}
\usepackage[font=footnotesize, labelfont=bf]{ caption}
\usepackage{ esint }
\usepackage{stackengine,scalerel,graphicx}
\stackMath
\newcommand\overarc[1]{\ThisStyle{%
  \setbox0=\hbox{$\SavedStyle#1$}%
  \stackon[.5pt]{\SavedStyle#1}{%
  \rotatebox{90}{$\SavedStyle\scaleto{)}{.95\wd0}$}}}}

\definecolor{trueblue}{rgb}{0.0, 0.45, 0.81}

\newcommand{\dx}{\, {\rm d}x}
\newcommand{\e}{\varepsilon}

\newcommand{\eg}{{\it e.g.}, }
\newcommand{\ie}{{\it i.e.}, }

\theoremstyle{plain}
\begingroup
\newtheorem{theorem}{Theorem}[section]
\newtheorem{lemma}[theorem]{Lemma}
\newtheorem{proposition}[theorem]{Proposition}

\endgroup

\newenvironment{step}[1]{\underline{Step #1}.}{}

\theoremstyle{definition}
\begingroup

\endgroup
\theoremstyle{remark}
\newtheorem{remark}[theorem]{Remark}

\renewcommand{\tilde}{\widetilde}
\newcommand{\sm}{\setminus}
\newcommand{\cf}{{\it cf.}}
\DeclareMathOperator{\argmin}{argmin}

\DeclareMathOperator{\sign}{sign}
\renewcommand{\d}{ \mathrm{d}}
\newcommand{\Adm}[2]{\mathrm{Adm}_{#1,#2}}
\DeclareMathOperator{\curl}{curl}
\DeclareMathOperator{\conv}{conv}
\DeclareMathOperator{\dist}{dist}
\DeclareMathOperator{\supp}{supp}

\numberwithin{equation}{section}
\newcommand{\N}{\mathbb{N}}
\newcommand{\Z}{\mathbb{Z}}
\newcommand{\R}{\mathbb{R}}

\renewcommand{\S}{\mathbb{S}}
\renewcommand{\L}{\mathcal{L}}
\renewcommand{\H}{\mathcal{H}}
\newcommand{\K}{\mathcal{K}}
\newcommand{\D}{\mathrm{D}}
\newcommand{\dH}{\, {\rm d}\H^1}
\newcommand{\de}{\partial}
\newcommand{\T}{\mathcal{T}}
\newcommand{\SF}{\mathcal{SF}}
\newcommand{\tb}{\hat{e}}
\newcommand{\x}{{\times}}
\newcommand{\defas}{:=}
\newcommand{\wto}{\rightharpoonup}
\newcommand{\wsto}{\overset{*}{\wto}}

\newcommand{\mres}{\mathbin{\vrule height 1.6ex depth 0pt width 0.13ex\vrule height 0.13ex depth 0pt width 1.3ex}}

\newcommand{\loc}{\mathrm{loc}}

\newcommand{\A}[2]{A_{#1,#2}}

\usepackage[normalem]{ulem}

\begin{document}

\title[The antiferromagnetic XY model on the triangular lattice]{The antiferromagnetic XY model on the triangular lattice: topological singularities}

\author{Annika Bach}
\address[Annika Bach]{TU M\"unchen, Germany}
\email{annika.bach@ma.tum.de}

\author{Marco Cicalese}
\address[Marco Cicalese]{TU M\"unchen, Germany}
\email{cicalese@ma.tum.de}

\author{Leonard Kreutz}
\address[Leonard Kreutz]{WWU M\"unster, Germany}
\email{lkreutz@uni-muenster.de}

\author{Gianluca Orlando}
\address[Gianluca Orlando]{TU M\"unchen, Germany}
\email{orlando@ma.tum.de}

\keywords{$\Gamma$-convergence, Frustrated lattice systems, Topological singularities.}


\begin{abstract}
We study the discrete-to-continuum variational limit of the antiferromagnetic XY model on the two-dimensional triangular lattice in the vortex regime. Within this regime, the spin system cannot overcome the energetic barrier of chirality transitions, hence one of the two chirality phases is prevalent. We find the order parameter that describes the vortex structure of the spin field in the majority chirality phase and we compute explicitly the $\Gamma$-limit of the scaled energy, showing that it concentrates on finitely many vortex-like singularities of the spin field.
\end{abstract}

\subjclass[2010]{49J45, 49M25, 82B20, 82D40, 35Q56.} 
\maketitle

\setcounter{tocdepth}{1}
\tableofcontents

\maketitle

\section{Introduction}

Antiferromagnetic spin systems are magnetic lattice systems in which the exchange interaction between two spins favors anti-alignment. Such systems are said to be geometrically frustrated if, due to the geometry of the lattice, no spin configuration can simultaneously minimize all pairwise interactions. As a consequence of that, ground states of frustrated spin systems may exhibit nontrivial patterns and give rise to unconventional magnetic order, whose understanding has occupied the Statistical Physics and Condensed Matter communities in the last decades~\cite{Die, LJNL, MS}.

In this paper we are interested in the antiferromagnetic XY spin system on the triangular lattice (AFXY), a system that has attracted the attention of a large scientific community because of its relevance in understanding phase transition properties of frustrated spin models as those governing the physics of Josephson junctions, helimagnets and discotic liquid crystals (see for instance \cite{Ka98} and references therein). Our present contribution is undertaken within the framework of ``discrete-to-continuum variational analysis'' by means of $\Gamma$-convergence (\cf~\cite{DM,Bra}). It aims at the first mathematically rigorous derivation of the coarse grained energy of the AFXY system as the lattice spacing vanishes and the energy scaling allows the formation of finitely many spin vortices. This is a further step towards a complete understanding of the AFXY model, whose variational analysis has been initiated in~\cite{BacCicKreOrl} at a different scaling, which leads to interfacial-type energies, as we recall below. It is worth mentioning that interfacial energies often result from different frustration mechanisms in the variational analysis of spin systems, \eg those induced by the competition of ferromagnetic (favoring alignment) and antiferromagnetic interactions~\cite{AliBraCic, CicSol, BraCic, SciVal, CicForOrl, DanRun}.  

\par

The AFXY is a 2-dimensional nearest-neighbors antiferromagnetic planar spin model on the triangular lattice, \cf~\cite[Chapter~1]{Die}. We let $\e > 0$ be a small parameter and we consider the triangular lattice $\L_\e$ with spacing $\e$ (see below for the precise definition). To every spin field $u \colon \L_\e \to \S^1$ we associate the energy 
\begin{equation} \label{intro:AFXY energy}
    \sum_{\substack{\e \sigma, \e \sigma' \in \L_\e \\ |\sigma - \sigma'| = 1}}  u(\e \sigma) \cdot u(\e \sigma')   \, ,
\end{equation}
where $\cdot$ denotes the scalar product. This model is antiferromagnetic since the interaction energy between two neighboring spins is minimized by two opposite vectors. The geometry of the triangular lattice, though, frustrates the system. In fact, already for a single triangular plaquette of the lattice no spin configuration minimizes the energy of all the three interacting pairs, since such a configuration should be made of three pairwise opposite vectors. 
In order to find the ground states of the system, one can rearrange the indices of the sum in~\eqref{intro:AFXY energy} to have 
\begin{equation} \label{intro:recasting the energy}
    \sum_{T} \big( u(\e i) \cdot  u(\e j)  +   u(\e j) \cdot u(\e k)   +    u(\e k) \cdot  u(\e i)   \big) = \frac{1}{2} \sum_{T} \big( |u(\e i) + u(\e j) + u(\e k) |^2 - 3 \big)\,, 
\end{equation}
where the sum is now running over all triangular plaquettes $T$ with vertices $\e i, \e j, \e k \in \L_\e$. 
The formula above shows that in each triangle $T$ the energy is minimized (and is equal to $-\frac{3}{2}$) only when $u(\e i) + u(\e j) + u(\e k) = 0$, \ie when the vectors of a triple $(u(\e i),u(\e j),u(\e k))$ point at the vertices of an equilateral triangle. The set of all the ground states is then obtained from this configuration thanks to the symmetries of the system, namely the $\S^1$ and the $\Z_2$ symmetry. By the $\S^1$-symmetry, every rotation of a minimizing triple is minimizing, too. By the $\Z_2$-symmetry, triples obtained from a minimizing triple $(u(\e i),u(\e j),u(\e k))$ via a permutation of negative sign as $(u(\e i),u(\e k),u(\e j))$ are also minimizing.
 The symmetry analysis above shows the existence of two families of ground states that can be distinguished through the {\em chirality}, a scalar quantity (invariant under rotations) which quantifies the handedness of a certain spin structure. To define the chirality of a spin field $u$ in a triangle $T$, we need to fix an ordering of its vertices $\e i$, $\e j$, $\e k$. We write the triangular lattice $\L$ as $\mathcal{L}:= \{z_1 \tb_1 + z_2 \tb_2 : z_1,z_2 \in \mathbb{Z}\}$ with $\tb_1=(1,0)$, and $\tb_2=\frac{1}{2}(1,\sqrt{3})$. We introduce also $\tb_3\defas\tfrac{1}{2}(-1,\sqrt{3})$ as a further unit vector connecting points of $\L$ and define three pairwise disjoint sublattices of $\L$, denoted by $\L^1$, $\L^2$, and~$\L^3$, by
 \begin{equation*} 
 \L^1\defas\{z_1(\tb_1+\tb_2)+z_2(\tb_2+\tb_3) : z_1,z_2\in\Z\}\, ,\quad\L^2\defas\L^1+\tb_1\, ,\quad\L^3\defas\L^1+\tb_2\, .
 \end{equation*}
 We assume that $\e i \in \e \L^1$, $\e j \in \e \L^2$, $\e k \in \e \L^3$ and we set (see~\eqref{def:chirality} for the precise definition)
\begin{equation*}
    \chi(u,T) = \frac{2}{3\sqrt{3}}\big(u(\e i)\x u(\e j)+u(\e j)\x u(\e k)+u(\e k)\x u(\e i) \big)  \in [-1,1]  \, ,
\end{equation*}
where $\times$ is the cross product. We let $\chi(u) \in L^\infty(\R^2)$ denote the function equal to $\chi(u,T)$ on the interior of each triangular plaquette $T$.
 By \cite[Remark~2.2]{BacCicKreOrl}, the ground states are characterized as those spin configurations $u$ that satisfy either $\chi(u) \equiv 1$ or $\chi(u) \equiv - 1$. In order to describe more precisely our framework, let us fix $\Omega \subset \R^2$ open, bounded, and connected (if not we work on each connected component) and let us consider the energy~\eqref{intro:recasting the energy} restricted to $\Omega$, \ie computed only on those plaquettes of $\L_\e$ contained in $\Omega$. We refer the energy to its minimum by removing the energy of the ground states ($-\frac{3}{2}$ for each plaquette) and then divide by the number of lattice points in $\Omega$, which is of order $1/\e^{2}$, to obtain the energy per particle. Up to a multiplicative constant the latter reads as
\begin{equation*}
    E_\e(u,\Omega) = \sum_{T \subset \Omega} \e^2 |u(\e i) + u(\e j) + u(\e k) |^2.
\end{equation*}
In \cite{BacCicKreOrl} we analyze the energetic regime at which the two families of ground states coexist and the energy of the system concentrates at the interface between the two chiral phases $\{\chi=1\}$ and $\{\chi=-1\}$. This happens assuming that, as $\e \to 0$, sequences of spin fields $u_\e \colon \L_\e \to \S^1$ can deviate from ground states under the constraint $E_\e(u_\e,\Omega) \leq C \e$. In this case, the chiralities $\chi(u_\e)$ converge strongly in $L^1(\Omega)$ to some $\chi \in BV(\Omega;\{-1,1\})$. As a result, the continuum limit of $\frac{1}{\e}E_\e$ is a function of a partition of $\Omega$ into sets of finite perimeter (the phases) where the chirality is either $+1$ or $-1$. More precisely, it $\Gamma$-converges to an anisotropic perimeter of the phase boundary.
At this energy scaling, the asymptotic behavior of the AFXY model shares similarities with systems having finitely many phases, such as Ising systems~\cite{CafDLL05, AliBraCic, BraKre} or Potts systems~\cite{CicOrlRuf-I}. 

\par 

In this paper we are interested in a much lower energetic regime, which does not allow chirality phase transitions. We turn our attention to sequences of spin fields $u_\e$ that satisfy the bound 
\begin{equation} \label{eqintro:bound vortices}
    E_\e(u_\e,\Omega) \leq C \e^{2}|\log\e| \, .
\end{equation}
Since $\e^{2}|\log\e|\ll \e$, within this energy bound the spin system cannot overcome the energetic barrier of the chirality transition (of order $\e$), hence the chiralities $\chi(u_\e)$ converge strongly in $L^1$ to either $\chi\equiv1$ or $\chi\equiv-1$, see Lemma~\ref{lemma:chi converges to 1}. 
However, within a fixed chirality phase, there is enough energy for the spin field to create finitely many vortices whose complex structure is displayed in Figure~\ref{fig:vortex}. Within the framework of ``discrete-to-continuum variational analysis'', the emergence of vortices in spin systems has been first observed in the ferromagnetic XY model~\cite{AliCic}, a system which is driven by an energy with neighboring interactions $-  u(\e \sigma) \cdot u(\e \sigma')$. The latter model has been thoroughly investigated both on the square lattice~\cite{Pon, AliCic, AliCicPon, AliDLGarPon, BraCicSol, CicOrlRuf1, CicOrlRuf2} and on the triangular lattice~\cite{CanSeg,DL}. Independently of the geometry of the lattice, it has been proved that spin fields that deviate from the ground states by an amount of energy of order $\e^{2}|\log\e|$ may form finitely many vortex-like singularities (topological charges as those arising in the Ginzburg-Landau model~\cite{BBH, SanSer-book, San, Jer,JerSon2, AlbBalOrl, SanSer, AliPon}).

In order to describe the vortex structure in the AFXY spin system, we assume that the limit chirality is $\chi\equiv1$. 
 Then, to every spin field $u\colon \L_\e \to \S^1$ we associate the auxiliary field $v\colon \L_\e \to \S^1$ defined by
    \begin{equation} \label{defintro:from u to v}
        v(\e i) := u(\e i) \, , \quad v(\e j) := R[-\tfrac{2\pi}{3}]( u(\e j) ) \, , \quad v(\e k) := R[\tfrac{2\pi}{3}]( u(\e k) ) \, ,
    \end{equation}
    for $\e i \in \L^1_\e$, $\e j \in \L^2_\e$, $\e k \in \L^3_\e$, where $R[\alpha](\cdot)$ denotes the counterclockwise rotation by $\alpha$. The operation above transforms a ground state with chirality $1$ into a set of three parallel vectors, see Figure~\ref{fig:from u to v}. 
    
    \begin{figure}[H]
        \includegraphics{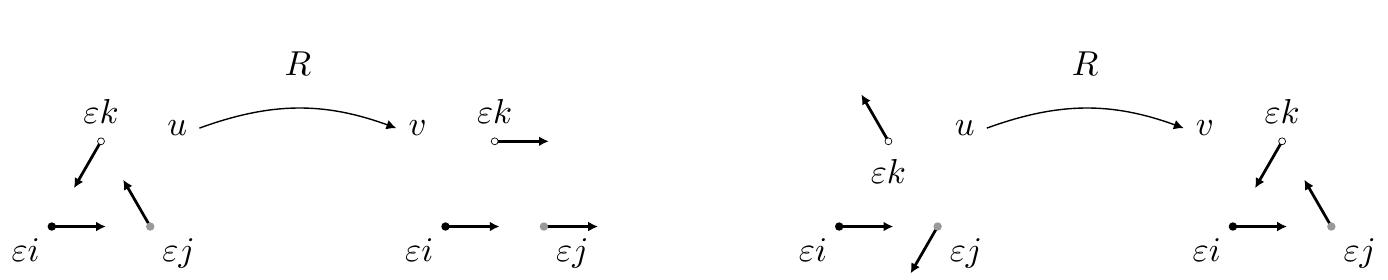}
        \caption{On the left: a ground state $u$ with chirality 1 is transformed into an auxiliary spin field $v$ given by parallel vectors. On the right: a ground state with chirality $-1$ is transformed in an auxiliary spin field $v$ with a nonzero XY-energy.}
        \label{fig:from u to v}
    \end{figure}

    The auxiliary variable $v$ introduced above plays a fundamental role in identifying the vortex structure in the AFXY spin system. In the first instance, this is suggested by the asymptotic behavior of $E_\e$ at the bulk scaling, \ie when assuming that sequences of spin fields $u_\e$ deviate from ground states satisfying the stricter energy constraint $E_\e(u_\e,\Omega) \leq C \e^2$. Under these assumptions we have, {\em a fortiori}, that the chiralities $\chi(u_\e)$ converge strongly in $L^1$ to either $\chi\equiv1$ or $\chi\equiv-1$. We work in the former case and we associate the auxiliary spin field $v_\e$  to every $u_\e$  as in~\eqref{defintro:from u to v}. In Theorem~\ref{thm:bulk} we prove that, under the previous assumptions, the piecewise affine interpolations of~$v_\e$ converge strongly in $L^2$ to a limit map $v \in H^1(\Omega;\S^1)$ and 
    \begin{equation} \label{eqintro:bulk}
        \frac{1}{\e^2} E_\e(u_\e,\Omega) \to \sqrt{3} \int_\Omega \! |\nabla v|^2 \, \d x \, , \quad \frac{1}{\e^2} XY_\e(v_\e,\Omega) \to \sqrt{3} \int_\Omega \! |\nabla v|^2 \, \d x \, ,
    \end{equation}
    in the sense of $\Gamma$-convergence, where $XY_\e(v,T) = \frac{1}{2} \e^2 (|v(\e i) - v(\e j)|^2 + |v(\e j) - v(\e k)|^2 + |v(\e k) - v(\e i)|^2)$ is the XY-energy of~$v$ in a triangle $T$. The proof relies on the relation (\cf~Lemma~\ref{lemma:bounds with XY} for the precise statement)
    \begin{equation*} 
        E_\e(u,T) \sim XY_\e(v,T) \quad \text{if} \quad \chi(u,T) \sim 1 \, ,
    \end{equation*}
    and the fact that, in the bulk scaling, the regions of $\Omega$ where the chirality is far from 1 concentrate around finitely many points, negligible for the limit energy. The asymptotic formula~\eqref{eqintro:bulk} and the known results for the XY-energy (see the discussion above) suggest that the limit of~$\frac{1}{\e^2|\log \e|} E_\e$ might detect a vortex structure in $u$ by means of the discrete vorticity measure $\mu_v$ of the auxiliary variable~$v$ (see~\eqref{def:vorticity} for the precise definition), as in Figure~\ref{fig:vortex}. However, the rigorous proof of the latter statement cannot result from a mere comparison between $E_\e$ and~$XY_\e$ under assumption~\eqref{eqintro:bound vortices}, as the next argument shows.
    
    \begin{figure}[H]
        \includegraphics{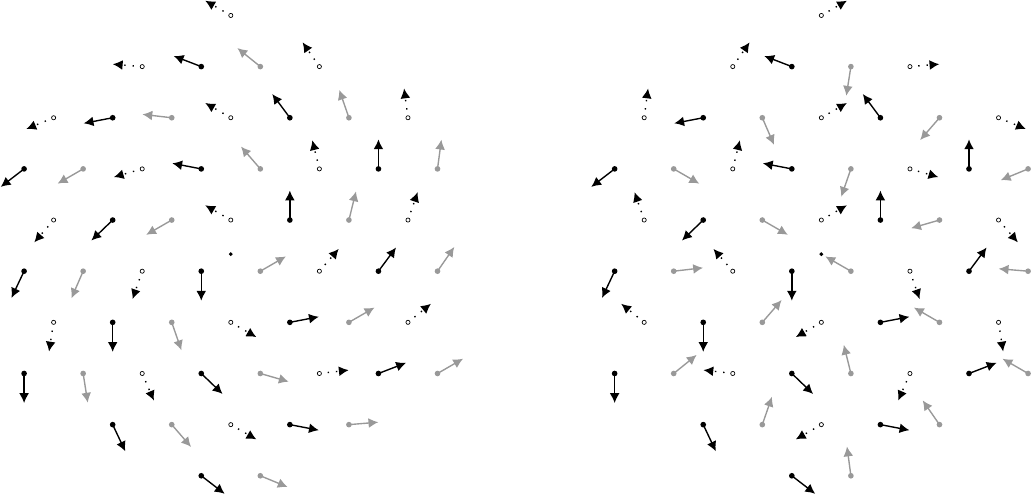}
        \caption{On the left: a vortex for the auxiliary spin field $v$. On the right: the corresponding spin field $u$.}
        
        \label{fig:vortex}
        \end{figure}
        
In the literature concerning the variational analysis of the XY model~\cite{AliCic,AliCicPon,DLGarPon, DL,CanSeg,BadCicDLPon}, the formation of finitely many vortex-like singularities in the limit as~$\e \to 0$ of a sequence of spin fields~$v_\e$ is proven if $XY_\e(v_\e,\Omega) \leq C \e^2 |\log \e|$. However, such a bound does not follow from our working assumption $E_\e(u_\e,\Omega) \leq C \e^2 |\log \e|$ in~\eqref{eqintro:bound vortices} for the corresponding spin field $u_\e$. This is due to the fact that an inequality $XY_\e(v,T) \leq C E_\e(u,T)$ does not hold true. Indeed, if $u$ is a ground state with $\chi(u,T) = -1$, then $E_\e(u,T) = 0$, but $XY_\e(v,T) > 0$, see Figure~\ref{fig:from u to v}. A fine estimate on the measure of the set $\{\chi(u_\e) \sim -1 \}$ allows us in Lemma~\ref{lemma:XY is log2} to obtain the sharp bound
\begin{equation}\label{eq:intro:boundXY}
    XY_\e(v_\e,\Omega) \leq C \e^2 |\log \e|^2.
\end{equation}
This weaker bound for the XY-model is, in general, not sufficient for detecting finitely many vortex-like singularities in the limit as~$\e \to 0$ and usually requires a different type of analysis~\cite{JerSon,SanSer,SanSer-book,AliCicPon,GarLeoPon}, related to the possible diffusion of the scaled measures $\frac{\mu_{v_\e}}{|\log\e|}$.
Nevertheless, due to the special structure of $\mu_{v_\e}$ in our setting, this phenomenon is ruled out and we are still able to prove that in the limit as $\e\to0$ the vorticity measures concentrate on finitely many points (the convergence is made rigorous in the flat topology, see~\eqref{def:flat} for the precise definition). This is contained in the main  theorem of the paper stated below.

\begin{theorem} \label{thm:main}
    Assume that $\Omega$ is an open, bounded, and connected set. The following results hold true:
    \begin{enumerate}[label=\roman*)]
        \item {\em (Compactness)} Let $u_\e\colon\L_\e\to\S^1$ be such that $E_\e(u_\e,\Omega) \leq C \e^2 |\log \e|$.  Then, up to a subsequence, either $\chi(u_\e) \to 1$ or $\chi(u_\e) \to -1$ in $L^1(\Omega)$. Assume that $\chi(u_\e) \to 1$, let $v_\e\colon\L_\e\to\S^1$ be the auxiliary spin field defined as in~\eqref{defintro:from u to v}.  Then there exists $\mu = \sum_{h = 1}^N d_h \delta_{x_h}$ with $d_h \in \Z$ and $x_h \in \Omega$ such that, up to a subsequence, $\|\mu_{v_\e} - \mu\|_{\mathrm{flat},\Omega^\prime} \to 0$ for all $\Omega^\prime\subset\subset \Omega$.
        \item  {\em ($\liminf$ inequality)} Let $u_\e\colon\L_\e\to\S^1$ be such that $\chi(u_\e) \to 1$ in $L^1(\Omega)$, let $v_\e\colon\L_\e\to\S^1$ be the auxiliary spin field   defined as in~\eqref{defintro:from u to v}. Let $\mu = \sum_{h = 1}^N d_h \delta_{x_h}$ with $d_h \in \Z$, $x_h \in \Omega$ and assume that $\|\mu_{v_\e} - \mu\|_{\mathrm{flat},\Omega^\prime} \to 0$ for all $\Omega^\prime\subset\subset \Omega$. Then 
        \begin{equation*}
            2\sqrt{3} \pi |\mu|(\Omega) \leq \liminf_{\e \to 0} \frac{1}{\e^2|\log \e|} E_\e(u_\e, \Omega) \, .
        \end{equation*}
        \item {\em ($\limsup$ inequality)} Let $\mu = \sum_{h = 1}^N d_h \delta_{x_h}$ with $d_h \in \Z$ and $x_h \in \Omega$. Then there exist $u_\e\colon\L_\e\to\S^1$ such that $\|\mu_{v_\e} - \mu\|_{\mathrm{flat},\Omega} \to 0$ and 
        \begin{equation*}
            \limsup_{\e \to 0} \frac{1}{\e^2 |\log \e|} E_\e(u_\e, \Omega)  \leq 2\sqrt{3}\pi |\mu|(\Omega)  \, ,
        \end{equation*}
        where $v_\e\colon\L_\e\to\S^1$ is the auxiliary spin field defined as in~\eqref{defintro:from u to v}.
    \end{enumerate}
\end{theorem}

We now illustrate the main ideas of the proof. To obtain {\em i)}, from~\eqref{eq:intro:boundXY} we first deduce that $|\mu_{v_\e}|(\Omega')\leq C|\log\e|^2$. However, we observe that only $|\log\e|$ many vortices can occur in the region where $\chi(u_\e)\sim 1$, which is consistent with the concentration of the energy on finitely many points. Instead, $|\log\e|^2$ many vortices only appear as $\e$-close dipoles in the region where $\chi(u_\e)\sim -1$ (see Figure~\ref{fig:dipoles}). Those can be shown to be asymptotically irrelevant using a variant of the ball construction~\cite{Jer, San}. 

\begin{figure}[H]
    \includegraphics{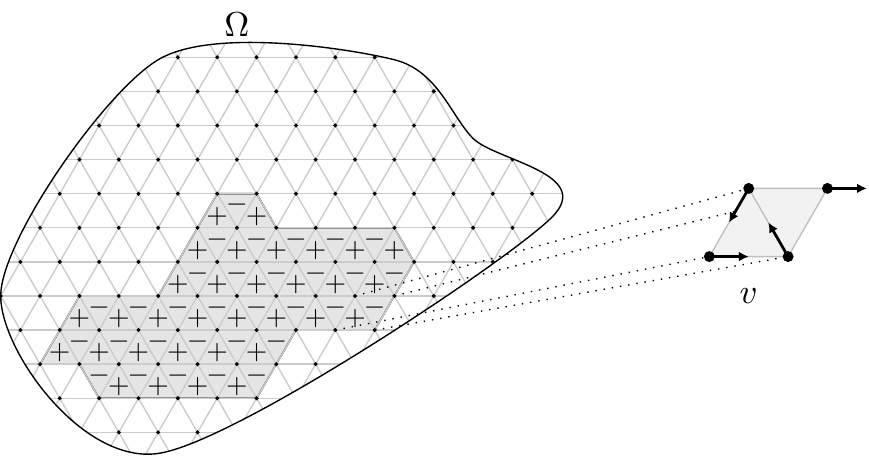}
        \caption{In the grey region the spin field $u_\e$ has chirality $-1$. There, the auxiliary variable $v_\e$ given by~\eqref{defintro:from u to v} may have $|\log \e|^2$ short dipoles, \eg as those depicted on the right. }
        \label{fig:dipoles}
\end{figure}

The ball construction is also the main tool for proving the asymptotic lower bound {\em ii)}, which is the most demanding part of the proof. Here, the choice of a precise expansion rate in the ball construction (see also~\cite{DLGarPon}) allows us to identify finitely many annuli, in which the energy concentrates. The radii of these annuli converge to zero at a much slower scale than $\e$, thus making possible to estimate the energy by exploiting the bulk scaling limit~\eqref{eqintro:bulk} (see also~\cite{AliBraCicDLPia} for a similar argument in the context of homogenization). A crucial step in the proof consists in the modification of the spin field in a diverging number of balls, where short dipoles annihilate. It is worth pointing out that in the discrete setting this is an additional source of difficulties, which is solved by proving the extension Lemma~\ref{lemma:extension}.

\vspace{.5cm}

\section{Notation and preliminary results}

\subsection*{Basic notation}

We let $\d_{\S^1}(a,b)$ denote the geodesic distance on~$\S^1$ between $a, b \in \S^1$ given by $ \d_{\S^1}(a,b) := 2 \arcsin(\frac{1}{2} |a -b|)$. Note that 
\begin{equation*} 
    |a - b| \leq \d_{\S^1}(a,b) \leq \frac{\pi}{2} |a-b| \, .
\end{equation*}

For every $0<r<R$ and $x_0 \in \mathbb{R}^2$, we define the annulus $\A{r}{R}(x_0):= B_R(x_0) \setminus \overline B_r(x_0)$. In the case $x_0=0$ we write $\A{r}{R}$. 

We let $\mathcal{A}(\mathbb{R}^2)$ denote the collection of open sets of $\R^2$. The Lebesgue measure of a measurable set $A$ will be denoted by $|A|$, while $\H^1$ stands for the 1-dimensional Hausdorff measure.

\subsection*{Triangular lattice} 
  
Here we set the notation for the triangular lattice $\L$. It is given~by
\begin{align*}
\mathcal{L}:= \{z_1 \tb_1 + z_2 \tb_2 : z_1,z_2 \in \mathbb{Z}\}\, ,
\end{align*}
with $\tb_1=(1,0)$, and $\tb_2=\frac{1}{2}(1,\sqrt{3})$. For later use, we find it convenient here to introduce $\tb_3\defas\tfrac{1}{2}(-1,\sqrt{3})$ and to define three pairwise disjoint sublattices of $\L$, denoted by $\L^1$, $\L^2$, and $\L^3$, by
\begin{equation*}  
\L^1\defas\{z_1(\tb_1+\tb_2)+z_2(\tb_2+\tb_3) : z_1,z_2\in\Z\}\, ,\quad\L^2\defas\L^1+\tb_1\, ,\quad\L^3\defas\L^1+\tb_2\, .
\end{equation*}
Eventually, we define the family of triangles subordinated to the lattice $\L$ by setting
\begin{equation*}
\T(\R^2)\defas\big\{T=\conv\{i,j,k\} : i,j,k\in\L,\ |i-j|=|j-k|=|k-i|=1\big\}\, ,
\end{equation*}
where $\conv\{i,j,k\}$ denotes the closed convex hull of $i,j,k$. 

For $\e>0$, we consider rescaled versions of $\L$ and $\T(\R^2)$  given by $\L_\e\defas\e\L$ and $\T_\e(\R^2)\defas\e\T(\R^2)$. With this notation every $T\in\T_\e(\R^2)$ has vertices $\e i,\e j,\e k\in\L_\e$. The same notation applies to the sublattices, namely $\L_\e^\alpha\defas\e\L^\alpha$ for $\alpha\in\{1,2,3\}$. Given a set $A\subset\R^2$ we let $\T_\e(A)\defas\{T\in\T_\e(\R^2)\colon T\subset A\}$ denote the subfamily of triangles contained in~$A$. Eventually, we introduce the set of admissible configurations as the set of all \textit{spin fields}
\begin{equation*} 
\SF_\e\defas\{u:\L_\e\to\S^1\}\, .
\end{equation*}
In the case $\e=1$ we set $\SF\defas \SF_1$.

\subsection*{The antiferromagnetic \texorpdfstring{$XY$}{} model}
For every $u \in \SF_\e$ and $T = \conv\{\e i, \e j, \e k\} \in\T_\e(\R^2)$ we set
\begin{equation*}
E_\e(u,T)\defas\e^2|u(\e i)+u(\e j)+u(\e k)|^2,
\end{equation*}
and we extend the energy to any set $A\subset\R^2$ by setting
\begin{equation*} 
E_\e(u,A)\defas\sum_{T\in\T_\e(A)}E_\e(u,T)\,.
\end{equation*}

\subsection*{Chirality} Given $u \in \mathcal{SF}_\e$  and $T=\mathrm{conv}\{\e i,\e j,\e k\} \in \mathcal{T}_\e(\R^2)$ with $\e i \in \mathcal{L}_\e^1, \e j \in \mathcal{L}_\e^2, \e k \in \mathcal{L}_\e^3$ we set 
\begin{equation}\label{def:chirality}
\chi(u,T) := \frac{2}{3\sqrt{3}}(u(\e i) \times u(\e j) + u(\e j)\times u(\e k) + u(\e k) \times u(\e i))\,.
\end{equation} 
Moreover, we define $\chi(u) \colon \Omega \to \R$ almost everywhere by setting $\chi(u)(x) := \chi(u,T)$ if $x \in \mathrm{int}\, T$. Given $u \in \mathcal{SF}_\e$ it is convenient to rewrite $\chi(u,T)$ in terms of the angular lift of $u$. More precisely, let $\theta(\e i),\theta(\e j),\theta(\e k) $ be such that $u(x) = \exp(\iota \theta(x))$, $x \in \{\e i,\e j,\e k\}$. Then
    \begin{equation*} 
        \chi(u, T) =\frac{2}{3\sqrt{3} } \big( \sin(\theta(\e j)-\theta(\e i)) + \sin(\theta(\e k)-\theta(\e j)) + \sin(\theta(\e i)-\theta(\e k)) \big) \, .
    \end{equation*}

   \begin{remark}\label{rmk:chiralityenergy} Let $T=\conv\{\e i,\e j,\e k\} \in \mathcal{T}_\e(\R^2)$ with $\e i \in \mathcal{L}_\e^1, \e j \in \mathcal{L}_\e^2, \e k \in \mathcal{L}_\e^3$. Given $u \in \mathcal{SF}_\e$ one can show that if $E_\e(u,T)= 0$, then $\chi(u,T)\in \{-1,1\}$. Therefore, a continuity argument shows that for every $\eta \in (0,1)$ there exists $C_\eta>0$ such that for every $u \in \mathcal{SF}_\e$  the following implication holds:
\begin{equation*}
\chi(u,T) \in (-1+\eta,1-\eta) \implies E_\e(u,T) \geq \e^2C_\eta\,.
\end{equation*}   
   \end{remark}

Given a triangle $T \in \T_\e(\R^2)$ we introduce the class $\mathcal{N}_\e(T)$ of its neighboring triangles, namely those triangles in $\T_\e(\R^2)$ that share a side with $T$. More precisely, we define
\begin{equation*}
    \mathcal{N}_\e(T):=\{T' \in \T_\e(\R^2) : \H^1(T \cap T')= \e \} \,.
\end{equation*}

\begin{lemma} \label{lemma:energy of two triangles} Let $u \in \mathcal{SF}_\e$ and let $\eta \in (0,1]$. Let $T \in \mathcal{T}_\e(\mathbb{R}^2)$ and $T' \in \mathcal{N}_\e(T)$ and assume that  $\chi(u,T) \leq 1-\eta$ and $\chi(u,T')\geq 1-\eta$. Then there exists a constant $C_\eta >0$ such that  $E_\e(u,T\cup T') \geq \e^2 C_\eta$.
\end{lemma}
\begin{proof} Without loss of generality let $T= \mathrm{conv}\{\e i, \e j, \e k\}$ and $T' =  \mathrm{conv}\{\e i', \e j, \e k\}$ with $\e i, \e i'\in \mathcal{L}_\e^1$, $\e j \in \mathcal{L}_\e^2$, and $\e k \in \mathcal{L}_\e^3$. Let $u^\eta \in \mathcal{SF}_\e$ be such that 
\begin{equation*}
E_\e(u^\eta,T\cup T') = \min \big\{ E_\e(u,T\cup T') \colon u \in \mathcal{SF}_\e \, , \ \chi(u,T) \leq 1-\eta \, , \ \chi(u,T') \geq 1-\eta \big\}\,.
\end{equation*}
By a scaling argument $u^\eta$ is independent of $\e$. Moreover, since 
\begin{equation*}
\begin{split}
& E_\e(u^\eta,T\cup T') = \e^2\big(|u^\eta(\e i) + u^\eta(\e j) + u^\eta(\e k)|^2+ |u^\eta(\e i') + u^\eta(\e j) + u^\eta(\e k)|^2 \big) \\& \quad  \geq \frac{\e^2}{2} |u^\eta(\e i) + u^\eta(\e j) + u^\eta(\e k)-(u^\eta(\e i') + u^\eta(\e j) + u^\eta(\e k))|^2 = \frac{\e^2}{2} |u^\eta(\e i)-u^\eta(\e i')|^2,
\end{split}
\end{equation*}
we have that
\begin{equation*}
E_\e(u^\eta,T\cup T') \geq \frac{\e^2}{2} \max \Big\{|u^\eta(\e i)-u^\eta(\e i')|^2, |u^\eta(\e i) + u^\eta(\e j) + u^\eta(\e k)|^2\Big\}=:\e^2 C_\eta\,.
\end{equation*}
Now either $u^\eta(\e i) \neq u^\eta(\e i')$ in which case it is clear that $C_\eta >0$. On the other hand, if $u^\eta(\e i) = u^\eta(\e i')$, then $\chi(u^\eta,T)= \chi(u^\eta,T') =1-\eta$. Then, due to Remark~\ref{rmk:chiralityenergy}, there exists $C'_\eta >0$ such that $\e^2C_\eta\geq E_\e(u^\eta,T) \geq \e^2C'_\eta$. Thus, $C_\eta>0$, which concludes the proof. 
\end{proof}

\begin{remark} \label{rmk:chirality angles}
    Let us consider the function
    \begin{equation*}
        \chi(\theta_1, \theta_2) :=\frac{2}{3\sqrt{3} } \big( \sin(\theta_1) + \sin(\theta_2 - \theta_1) - \sin(\theta_2) \big) \, .
    \end{equation*}
    For every $\eta'>0$ there exists $\eta \in (0,1)$ such that, if $\theta_1, \theta_2 \in [-\pi,\pi]$ satisfy $\chi(\theta_1,\theta_2) > 1-\eta$, then $|\theta_1 - \frac{2\pi}{3}| < \eta'$ and $|\theta_2 + \frac{2\pi}{3}| < \eta'$. This follows from a continuity argument since the global maximum~1 is achieved only at $(\theta_1, \theta_2) = (\tfrac{2\pi}{3}, - \tfrac{2\pi}{3})$ in the square $[-\pi,\pi]^2$ (\cf~\cite[Lemma 2.1]{BacCicKreOrl}). Analogously, if $\theta_1, \theta_2 \in [-\pi,\pi]$ satisfy $\chi(\theta_1,\theta_2) < -1+\eta$, then $|\theta_1 + \frac{2\pi}{3}| < \eta'$ and $|\theta_2 - \frac{2\pi}{3}| < \eta'$. 
\end{remark}

In the next lemma we count the number of triangles where the chirality $\chi(u_\e)$ is far from 1, assuming that $\chi(u_\e) \to 1$. 

\begin{lemma} \label{lemma:counting} Let $\Omega \subset \R^2$ be open, bounded, and connected and let and let $U \subset \subset \Omega$  with Lipschitz boundary. Let $u_\e \in \SF_\e$ be such that $\chi(u_\e) \to 1$ in $L^1(\Omega)$. Given $\eta \in (0,1)$ there exists $C_\eta>0$, depending on $U$,  such that for $\e $ small enough
    \begin{equation*}
       \# \{ T \in \T_\e(U) : \chi(u_\e, T) \leq 1 - \eta\}  \leq   C_\eta\Big(\frac{1}{\e^2}E_\e(u_\e,\Omega)\Big)^2 .
     \end{equation*}
\end{lemma}
\begin{proof} Without restriction we assume that also $U$ is connected.
Let us consider the set 
\begin{equation*}
   N_\e^\eta := \bigcup \{ T \in \T_\e(\Omega) : \chi(u_\e, T) \leq 1 - \eta \} \, .
\end{equation*}
Note that, since $\chi(u_\e) \to 1$ in $L^1(\Omega)$,
\begin{equation*}
   |N_\e^\eta \cap U| \leq |\{ |1-\chi(u_\e)| \geq \eta \} \cap U| \to 0 \, .
\end{equation*}
Thus, by the relative isoperimetric inequality, there exists $C>0$ depending on $U$ such that 
\begin{equation} \label{eq:isoperimetric}
   |N_\e^\eta \cap U|  = \min\{|U \sm N_\e^\eta|, |N_\e^\eta \cap U|\} \leq C  ( \H^1( \de N_\e^\eta \cap U) )^2
\end{equation}
for $\e$ small enough. We define the collection of triangles 
\begin{equation*}
   \T_\e^{1-\eta} := \{ T \in \T_\e(\Omega) : \chi(u_\e, T) \leq 1 - \eta \text{ and  }  \chi(u_\e, T') > 1 -\eta \text{ for some } T' \in \mathcal{N}_\e(T) \cap \T_\e(\Omega)\} 
\end{equation*}
and we remark that $\de N_\e^\eta  \cap U \subset \de \big( \bigcup_{T \in \T_\e^{1-\eta}} T \big)$. From the previous inclusion it follows that 
\begin{equation} \label{eq:29101950}
    \H^1( \de N_\e^\eta \cap U ) \leq 3 \e \# \T_\e^{1-\eta} .
\end{equation}
By Lemma~\ref{lemma:energy of two triangles}, we obtain that 
\begin{equation*}
   C_\eta \e^2 \# \T_\e^{1-\eta}  \leq \sum_{T \in \T_\e^{1-\eta}}  \sum_{T' \in \mathcal{N}_\e(T) \cap \T_\e(\Omega)} E_\e(u_\e, T \cup T') \leq 3 E_\e(u_\e, \Omega)\, .
\end{equation*}
Then~\eqref{eq:isoperimetric} and~\eqref{eq:29101950} imply that
\begin{equation*}
    \frac{\sqrt{3}}{4}\e^2 \# \{ T \in  \T_\e(U) : \chi(u_\e, T) \leq 1 - \eta \} \leq |N_\e^\eta \cap U| \leq  C ( \H^1( \de N_\e^\eta  \cap U) )^2  \leq  \frac{C_\eta}{\e^2}  E_\e(u_\e, \Omega)^2. 
\end{equation*}
This concludes the proof.
\end{proof}

\subsection*{The ferromagnetic \texorpdfstring{$XY$}{} model}

In this subsection we fix the notation for the ferromagnetic XY model and we recall some properties that relate it to the Ginzburg-Landau functional.   For every $v \in \SF_\e$ and $T = \conv\{\e i, \e j, \e k\} \in\T_\e(\R^2)$ we set 
\begin{equation*}
    XY_\e(v,T) \defas \frac{1}{2}\e^2 \big( |v(\e i) - v(\e j)|^2 + |v(\e j) - v(\e k)|^2 + |v(\e k) - v(\e i)|^2\big)
\end{equation*}
and for any set $A\subset\R^2$
\begin{equation*}
    XY_\e(v,A)\defas\sum_{T\in\T_\e(A)}XY_\e(v,T) \,.
\end{equation*}

\begin{remark} \label{rmk:XY is integral}
    Given $v \in \SF_\e$, we let $\hat v \colon \R^2 \to \R^2$ denote its piecewise affine interpolation determined by the following conditions: for every $T = \conv\{\e i, \e j, \e k\} \in \T_\e(\R^2)$ the map $\hat v$ is affine in $T$ and $\hat v(\e i) = v(\e i)$, $\hat v(\e j) = v(\e j)$, $\hat v(\e k) = v(\e k)$. Then 
    \begin{equation*}
        \begin{split}
            XY_\e(v,T) & = \frac{1}{2}\e^4 \big( |\nabla \hat v \,  \tb_1|^2 + |\nabla \hat v  \, \tb_2|^2 + |\nabla \hat v \, \tb_3|^2\big) \\
            & = \frac{1}{2}\e^4 \big( |\de_1 \hat v|^2 + |\tfrac{1}{2}\de_1 \hat v + \tfrac{\sqrt{3}}{2}\de_2 \hat v|^2 + |-\tfrac{1}{2}\de_1 \hat v + \tfrac{\sqrt{3}}{2}\de_2 \hat v|^2\big) \\
            & = \frac{3}{4}\e^4 |\nabla \hat v|^2 =  \sqrt{3} \e^2 \int_T |\nabla \hat v|^2 \, \d x \, .    
        \end{split}
    \end{equation*}
\end{remark}
We recall the following key lemma proven in~\cite[Lemma~2]{AliCic} in the case of the XY-energy on the square lattice. The same proof can be repeated for the XY-energy on the triangular lattice. 
\begin{lemma} \label{lemma:potential part}
    Let $v \in \SF_\e$ and let $\hat v$ be its piecewise affine interpolation. Let $\Omega' \subset \subset \Omega$. Then 
    \begin{equation*}
         \int_{\Omega'} (1-|\hat v|^2)^2 \, \d x \leq C XY_\e(v,\Omega) \, .
    \end{equation*} 
\end{lemma}

\subsection*{Auxiliary spin field}
We introduce here an auxiliary variable suited to describe the vortex structure in the AFXY model. 
Given a vector $u = \exp(\iota \phi )\in \S^1$ with $\phi \in \R$ and an angle~$\theta \in \R$, we set 
\begin{equation} \label{def:rotation}
    R[\theta](u) := \exp\big(\iota (\phi + \theta)\big) \, .
\end{equation}

    Let $u \in \SF_\e$ and let $v \in \SF_\e$ be defined by 
    \begin{equation} \label{def:from u to v}
        v(\e i) := u(\e i) \, , \quad v(\e j) := R[-\tfrac{2\pi}{3}]( u(\e j) ) \, , \quad v(\e k) := R[\tfrac{2\pi}{3}]( u(\e k) ) \, ,
    \end{equation}
    for $\e i \in \L^1_\e$, $\e j \in \L^2_\e$, $\e k \in \L^3_\e$. Note that the operation above transforms a ground state with chirality 1 into a set of three parallel vectors. 
    
    We can relate the $E_\e$ energy of $u$ to the $XY_\e$ energy of $v$ in a triangle $T = \conv\{\e i, \e j, \e k\}$. Letting $\theta(\e j) \in \R$ be an angle between $u(\e i)$ and $u(\e j)$ and letting $\theta(\e k) \in \R$ be an angle between~$u(\e i)$ and~$u(\e k)$  we get that 
    \begin{equation}  \label{eq:E in terms of v}
        \begin{split}
            E_\e(u,T) & = \e^2|u(\e i)+u(\e j)+u(\e k)|^2 \\
            & = \e^2 \big( 3 + 2 \cos\big(\theta(\e j)\big) + 2 \cos\big(\theta(\e k)-\theta(\e j)\big) + 2 \cos\big(\theta(\e k)\big) \big) \\
            & = \e^2 \big( 3 - \cos\big(\theta(\e j) - \tfrac{2 \pi}{3} \big) - \cos\big(\theta(\e k) + \tfrac{2 \pi}{3} -\theta(\e j)  + \tfrac{2 \pi}{3}\big)  - \cos\big(\theta(\e k) + \tfrac{2 \pi}{3}\big) \big) \\
            & \quad  - \sqrt{3} \e^2   \big(  \sin(\theta(\e j) - \tfrac{2 \pi}{3} )  +  \sin(\theta(\e k) + \tfrac{2 \pi}{3} -\theta(\e j) + \tfrac{2 \pi}{3}) -  \sin(\theta(\e k) + \tfrac{2 \pi}{3})   \big) \\
            & = \e^2 \big( 3 - v(\e i) \cdot v(\e j) - v(\e j) \cdot v(\e k)  - v(\e k)\cdot v(\e i) \big) \\
            & \quad  + \frac{3}{4 }  \e^2\big( 2 \cos(\theta(\e j)) + 2 \cos(\theta(\e k)  -\theta(\e j)) +  2 \cos(\theta(\e k)) \big) \\
            &\quad  + \frac{\sqrt{3}}{2 }   \e^2   \big(  \sin(\theta(\e j) )    +    \sin(\theta(\e k)  -\theta(\e j))   + \sin(-\theta(\e k) )   \big)   \\
            & = XY_\e(v,T) + \frac{3}{4 } E_\e(u,T) - \frac{9}{4 } \e^2 (1 -  \chi(u,T))  \, .
        \end{split} 
    \end{equation}
    The previous inequality yields
    \begin{equation} \label{eq:E in terms of v final}
        E_\e(u,T) = 4 XY_\e(v,T) - 9 \e^2 (1 - \chi(u,T))
    \end{equation}
    and, in particular,  
    \begin{equation} \label{eq:E less than XY}
        E_\e(u,T) \leq 4 XY_\e(v,T) \, .
    \end{equation}
	\begin{remark}[Lower bound]
    In general, a lower bound $E_\e(u,T) \geq c XY_\e(v,T)$ does not hold true. For instance, if $u$ is a ground state with negative chirality in $T$, then $E_\e(u,T) = 0$ but $XY_\e(v,T) > 0$. Nonetheless, we show now that this kind of lower bound holds true if $u$ has chirality close to $1$ in $T$. 
	\end{remark}
    
     \begin{lemma} \label{lemma:bounds with XY}
        Let $u \in \SF_\e$ and let $v \in \SF_\e$ be the auxiliary spin field defined by~\eqref{def:from u to v}. Then for every $\lambda \in (0,1)$ there exists $\eta \in (0,1)$  such that   
        \begin{equation} \label{eq:bounds with XY}
             (1-\lambda)  XY_\e(v,T) \leq E_\e(u,T)  \leq (1+\lambda) XY_\e(v,T)
        \end{equation}
       for every $T = \conv\{ \e i, \e j, \e k\} \in \T_\e(\R^2)$, $\e i \in \L^1_\e, \e j \in \L^2_\e, \e k \in \L^3_\e$, such that either $\d_{\S^1}(v(\e i), v(\e j)) < \eta $ and $\d_{\S^1}(v(\e i), v(\e k)) < \eta$ or $\chi(u,T) > 1 - \eta$. 
    \end{lemma}
    \begin{proof}
     Let us fix $\lambda \in (0,1)$ and let $\delta \in (0,\frac{1}{7})$ be such that 
    \begin{equation} \label{eq:lambda}
        1-\lambda < \frac{1-7\delta}{1-\delta} \quad \text{and} \quad \frac{1+7\delta}{1+\delta} < 1 + \lambda \, .
    \end{equation}
     On the one hand, we consider the function 
    \begin{equation*}
        \chi(\theta_1, \theta_2) :=\frac{2}{3\sqrt{3} } \big( \sin(\theta_1) + \sin(\theta_2 - \theta_1) - \sin(\theta_2) \big)
    \end{equation*}
    so that, adopting the notation from the computations in~\eqref{eq:E in terms of v}, $\chi(u,T) =  \chi(\theta(\e j), \theta(\e k))$. By Taylor expanding $\chi$ around the point $(\frac{2\pi}{3}, -\frac{2\pi}{3})$ we obtain that
    \begin{equation*} 
        1 - \chi(\theta_1, \theta_2) = 
         \frac{1}{3}A \begin{pmatrix} \theta_1 - \frac{2\pi}{3}  \\ \theta_2 + \frac{2\pi}{3}  \end{pmatrix}  \cdot \begin{pmatrix} \theta_1 - \frac{2\pi}{3}  \\ \theta_2 + \frac{2\pi}{3}  \end{pmatrix} + \rho_1(\theta_1, \theta_2) \, , \quad A:= \begin{pmatrix}  2 & -1  \\  -1 & 2 \end{pmatrix} = 3 \D^2 \chi\big(\tfrac{2\pi}{3}, -\tfrac{2\pi}{3} \big) \, ,
    \end{equation*}
    where $|\rho_1(\theta_1, \theta_2)| \leq C  \big( |\theta_1 - \tfrac{2\pi}{3}|^3 +  |\theta_2 + \tfrac{2\pi}{3}|^3 \big)$, the constant $C$ depending only on $\|\D^3 \chi \|_{L^\infty}$. There exists $\delta_1>0$ (depending only on $\delta$ and $\|\D^3 \chi \|_{L^\infty}$) such that, if $|\theta_1 - \tfrac{2\pi}{3}| < \delta_1$ and $|\theta_2 + \tfrac{2\pi}{3}| < \delta_1$, then
    \begin{equation*}
        \begin{split}
               |\rho_1(\theta_1,\theta_2)| & \leq C \big( |\theta_1 - \tfrac{2\pi}{3}|^3 +  |\theta_2 + \tfrac{2\pi}{3}|^3 \big)  \leq C \delta_1 \big( |\theta_1 - \tfrac{2\pi}{3}|^2 +  |\theta_2 + \tfrac{2\pi}{3}|^2 \big)\\
             &  \leq  \frac{\delta}{3} \big( |\theta_1 - \tfrac{2\pi}{3}|^2 +  |\theta_2 + \tfrac{2\pi}{3}|^2 \big)   \leq \frac{\delta }{3} A \begin{pmatrix} \theta_1 - \frac{2\pi}{3}  \\ \theta_2 + \frac{2\pi}{3}  \end{pmatrix}  \cdot \begin{pmatrix} \theta_1 - \frac{2\pi}{3}  \\ \theta_2 + \frac{2\pi}{3}  \end{pmatrix} \, . 
        \end{split}
    \end{equation*}
    This implies that 
    \begin{equation} \label{eq:Taylor chi}
        \frac{1-\delta}{3}  A \begin{pmatrix} \theta_1 - \frac{2\pi}{3}  \\ \theta_2 + \frac{2\pi}{3}  \end{pmatrix}  \cdot \begin{pmatrix} \theta_1 - \frac{2\pi}{3}  \\ \theta_2 + \frac{2\pi}{3}  \end{pmatrix} \leq 1-\chi(\theta_1,\theta_2)  \leq \frac{1+\delta}{3}  A \begin{pmatrix} \theta_1 - \frac{2\pi}{3}  \\ \theta_2 + \frac{2\pi}{3}  \end{pmatrix}  \cdot \begin{pmatrix} \theta_1 - \frac{2\pi}{3}  \\ \theta_2 + \frac{2\pi}{3}  \end{pmatrix}  \, .
    \end{equation} 
    On the other hand, we consider the function 
    \begin{equation*}
       f(\theta_1, \theta_2) :=  3 - \cos(\theta_1 - \tfrac{2\pi}{3}) - \cos(\theta_2 - \theta_1 + \tfrac{4\pi}{3}) - \cos(\theta_2 + \tfrac{2\pi}{3})  \, ,
    \end{equation*}
    so that $XY_\e(v,T) = \e^2 f(\theta(\e j), \theta(\e k))$. By Taylor expanding $f$ around $(\frac{2\pi}{3}, -\frac{2\pi}{3})$ we obtain that
    \begin{equation*} 
        f(\theta_1, \theta_2) = A \begin{pmatrix} \theta_1 - \frac{2\pi}{3}  \\ \theta_2 + \frac{2\pi}{3}  \end{pmatrix}  \cdot \begin{pmatrix} \theta_1 - \frac{2\pi}{3}  \\ \theta_2 + \frac{2\pi}{3}  \end{pmatrix} + \rho_2(\theta_1, \theta_2 ) \, ,
    \end{equation*}
    where $|\rho_2(\theta_1, \theta_2)| \leq C \big( |\theta_1 - \tfrac{2\pi}{3}|^3 +  |\theta_2 + \tfrac{2\pi}{3}|^3 \big)$, the constant $C$ depending only on $ \|\D^3 f \|_{L^\infty}$. There exists $\delta_2>0$ (depending only on $\delta$ and $\|\D^3 f \|_{L^\infty}$) such that, if $|\theta_1 - \tfrac{2\pi}{3}| < \delta_2$ and $|\theta_2 + \tfrac{2\pi}{3}| < \delta_2$, then
    \begin{equation*}
        \begin{split}
            |\rho_2(\theta_1,\theta_2)|  & \leq C \big( |\theta_1 - \tfrac{2\pi}{3}|^3 +  |\theta_2 + \tfrac{2\pi}{3}|^3 \big)  \leq C \delta_2 \big( |\theta_1 - \tfrac{2\pi}{3}|^2 +  |\theta_2 + \tfrac{2\pi}{3}|^2 \big)\\
             &  \leq  \delta \big( |\theta_1 - \tfrac{2\pi}{3}|^2 +  |\theta_2 + \tfrac{2\pi}{3}|^2 \big)   \leq \delta A \begin{pmatrix} \theta_1 - \frac{2\pi}{3}  \\ \theta_2 + \frac{2\pi}{3}  \end{pmatrix}  \cdot \begin{pmatrix} \theta_1 - \frac{2\pi}{3}  \\ \theta_2 + \frac{2\pi}{3}  \end{pmatrix} \, . 
        \end{split}
    \end{equation*}
    This implies that 
    \begin{equation} \label{eq:Taylor f}
        (1-\delta) A \begin{pmatrix} \theta_1 - \frac{2\pi}{3}  \\ \theta_2 + \frac{2\pi}{3}  \end{pmatrix}  \cdot \begin{pmatrix} \theta_1 - \frac{2\pi}{3}  \\ \theta_2 + \frac{2\pi}{3}  \end{pmatrix}  \leq f(\theta_1, \theta_2)  \leq (1+\delta) A \begin{pmatrix} \theta_1 - \frac{2\pi}{3}  \\ \theta_2 + \frac{2\pi}{3}  \end{pmatrix}  \cdot \begin{pmatrix} \theta_1 - \frac{2\pi}{3}  \\ \theta_2 + \frac{2\pi}{3}  \end{pmatrix} \, .
    \end{equation} 
    Setting $\eta' := \min\{\delta_1, \delta_2\}$, by~\eqref{eq:Taylor chi} and~\eqref{eq:Taylor f} we conclude that 
    \begin{equation} \label{eq:angles close to ground state}
        3\frac{1-\delta}{1+\delta} f(\theta_1, \theta_2) \leq 9( 1-\chi(\theta_1,\theta_2) )  \leq 3\frac{1+\delta}{1-\delta} f(\theta_1, \theta_2) \,,  \quad \text{for } |\theta_1 - \tfrac{2\pi}{3}| < \eta', \ |\theta_2 + \tfrac{2\pi}{3}| < \eta' \, . 
    \end{equation}
    This and~\eqref{eq:E in terms of v final} imply that~\eqref{eq:bounds with XY} holds true if $\d_{\S^1}(v(\e i), v(\e j)) < \eta'$ and $\d_{\S^1}(v(\e i), v(\e k)) < \eta'$. 
    
    By Remark~\ref{rmk:chirality angles}, there exists $\eta'' \in (0,1)$(depending on $\eta'$) such that, if $\theta_1,\theta_2 \in [-\pi,\pi]$ satisfy $\chi(\theta_1,\theta_2) > 1 -  \eta'' $, then $|\theta_1- \tfrac{2\pi}{3}| < \eta', \ |\theta_2 + \tfrac{2\pi}{3}| < \eta'$.  By~\eqref{eq:angles close to ground state} we conclude that 
    \begin{equation*}
        3\frac{1-\delta}{1+\delta} f(\theta_1, \theta_2) \leq  9( 1-\chi(\theta_1,\theta_2) )  \leq 3\frac{1+\delta}{1-\delta} f(\theta_1, \theta_2) \,,  \quad  \text{for } \chi(\theta_1, \theta_2) > 1 - \eta'' \, . 
    \end{equation*}
    This, together with~\eqref{eq:E in terms of v final} and~\eqref{eq:lambda}, implies that~\eqref{eq:bounds with XY} holds true if  $\chi(u,T) > 1 - \eta''$. Setting  $\eta := \min\{\eta',\eta''\}$  concludes the proof.
\end{proof}

\section{Topological singularities}

In this section we recall the definition of discrete vorticity and its relation with the Jacobian of maps in the continuum. In particular, in Remark~\ref{rmk:arcwise interpolation} we introduce an interpolation of discrete spin fields that makes this relation clear. 

\subsection*{Discrete vorticity} Let $v \in \SF_\e$ and let $T \in \T_\e(\R^2)$. We define the {\em discrete vorticity} $\mu_v(T)$ of~$v$ in~$T$ as follows. Let $\varphi \colon \L_\e \to \R$ be any function such that $v(x) = \exp(\iota \varphi(x))$. We define the projection on~$2 \pi \Z$ by
\begin{equation*}
    P_{2\pi \Z}(t) := \argmin \{ |t - s| : s \in 2\pi \Z \} \, ,
\end{equation*} 
choosing the $\argmin$ with minimal modulus when it is not unique. We consider the angle between the vectors $v(x)$ and $v(x')$ given by
\begin{equation} \label{def:de} 
    \d^e \varphi(x,x') := \varphi(x') - \varphi(x) - P_{2\pi \Z}\big( \varphi(x') - \varphi(x) \big) \, .
\end{equation}
We remark that $\d^e \varphi$ does not depend on the choice of $\varphi$. Moreover, $\d^e \varphi(x,x') = - \d^e \varphi(x',x)$.  Let now $(x_1,x_2,x_3)$ be the vertices of $T$ in counterclockwise order. Then we set  
\begin{equation} \label{def:vorticity}
    \mu_v(T) := \frac{1}{2 \pi} \big( \d^e \varphi(x_1,x_2) + \d^e \varphi(x_2,x_3) + \d^e \varphi(x_3,x_1)\big) \, .
\end{equation}
Since $| \d^e \varphi(x,x') | \leq \pi$, we immediately deduce that $\mu_v(T) \in \{-1,0,1\}$. 
Finally, we define the measure  
\begin{equation*}
    \mu_v := \sum_{T \in \T_\e(\R^2)} \mu_v(T) \delta_{b(T)} \, ,
\end{equation*}
where $b(T) \in \R^2$ denotes the barycenter of the triangle $T$.

\begin{remark} \label{rmk:XY controls mass}
    There exists a constant $C > 0$ such that for every $v \in \SF_\e$ and $\Omega' \subset \subset \Omega$ with $\mathrm{dist}(\Omega^\prime, \partial \Omega) > \e$
    \begin{equation*}
        |\mu_{v}|(\Omega') \leq \frac{C}{\e^2}  XY_\e(v,\Omega) \, .
    \end{equation*}
    Indeed, we start by observing that, by the definition of $\mu_v$,
    \begin{equation*}
        |\mu_{v}|(\Omega') \leq \sum_{T \in \T_\e(\Omega)} |\mu_v(T)| \, .
    \end{equation*}
    Thanks to the previous inequality, it is enough to prove that there exists a universal constant $C > 0$ such that $|\mu_v|(T) \leq \frac{C}{\e^2} XY_\e(v,T) $ for every $T \in \T_\e(\Omega)$. Given $T = \conv\{\e i, \e j, \e k\} \in \T_\e(\Omega)$, two cases are possible: either $\mu_v(T) = 0$ or $\mu_v(T) \neq 0$. In the former case, we trivially have $|\mu_v|(T) \leq \frac{1}{\e^2} XY_\e(v,T)$. In the latter case, let $\varphi(x) \in \R$ be such that $v(x) = \exp(\iota \varphi(x))$ for $x \in \{\e i, \e j, \e k\}$. Then one value between $|\d^e \varphi(\e i, \e j)|$, $|\d^e \varphi(\e j, \e k)|$, or $|\d^e \varphi(\e k, \e i)|$ is greater than or equal to $\frac{2\pi}{3}$. Since $|v(x) - v(x')|^2 \geq \frac{4}{\pi^2} |\d^e \varphi(x, x')|^2$, we conclude that $\frac{1}{\e^2} XY_\e(v,T) \geq  \frac{8}{9}  |\mu_v(T)|$. 
\end{remark}

\subsection*{Jacobians and degree}
We recall here some definitions and basic results concerning topological singularities. Let $U \subset \R^2$ be an open set and let $v = (v_1,v_2) \in W^{1,1}(U;\R^2) \cap L^\infty(U;\R^2)$. We define the {\em pre-Jacobian} (also known as {\em current}) of $v$ by 
\begin{equation*}
    j(v) := \frac{1}{2}( v_1 \nabla v_2 - v_2 \nabla v_1 ) \, .
\end{equation*}
The {\em distributional Jacobian} of $v$ is defined by 
\begin{equation*}
    J(v) := \curl (j(v)) \, ,
\end{equation*}
in the sense of distributions, \ie 
\begin{equation*}
    \langle J(v), \psi \rangle = - \int_U j(v) \cdot \nabla^\perp \psi \, \d x \quad \text{for every } \psi \in C^\infty_c(U) \, ,
\end{equation*}
where $\nabla^\perp = (-\de_2, \de_1)$. Note that $J(v)$ is also well-defined when $v \in H^1(U;\R^2)$, and, in that case, it coincides with the $L^1$ function $\det \nabla v$.
 
Given $v =(v_1,v_2) \in H^{\frac{1}{2}}(\de B_\rho(x_0);\S^1)$, its {\em degree} is defined by 
\begin{equation} \label{def:degree}
    \deg(v,\de B_\rho(x_0)) := \frac{1}{2 \pi} \big( \langle \nabla_{\de B_\rho(x_0)} v_2 , v_1 \rangle_{H^{-\frac{1}{2}},H^\frac{1}{2}} - \langle \nabla_{\de B_\rho(x_0)} v_1 , v_2 \rangle_{H^{-\frac{1}{2}},H^\frac{1}{2}} \big) \, ,
\end{equation}
where $\langle \cdot  , \cdot  \rangle_{H^{-\frac{1}{2}},H^\frac{1}{2}}$ denotes the duality between $H^{-\frac{1}{2}}(\de B_\rho(x_0);\S^1)$ and $H^{\frac{1}{2}}(\de B_\rho(x_0);\S^1)$ and we let  $\nabla_{\de B_\rho(x_0)}$ denote the derivative on $\de B_\rho(x_0)$ with respect to the unit speed parametrization of $\de B_\rho(x_0)$. Note that, by definition, the map $v \in H^{\frac{1}{2}}(\de B_\rho(x_0);\S^1) \mapsto \deg(v,\de B_\rho(x_0))$ is continuous. We remark that
\begin{equation*}
    \deg(v,\de B_\rho(x_0)) = \frac{1}{2\pi}\int_{\de B_\rho(x_0)}   v_1 \nabla_{\de B_\rho(x_0)} v_2  - v_2 \nabla_{\de B_\rho(x_0)} v_1 \, \d \H^1 \quad \text{if } v \in H^1(\de B_\rho(x_0);\S^1)
\end{equation*} 
(and thus $v$ is continuous) and this notion coincides with the classical notion of degree.\footnote{One can see this by noticing that
\begin{equation*}
    \frac{1}{2\pi}\int_{\de B_\rho}   v_1 \nabla_{\de B_\rho} v_2  - v_2 \nabla_{\de B_\rho} v_1 \, \d \H^1 = \fint_{\de B_\rho}   v^* \omega_{\de B_\rho}  = \deg(v,\de B_\rho) \fint_{\de B_\rho} \omega_{\de B_\rho}  = \deg(v,\de B_\rho) \, ,
\end{equation*} 
where $v^* \omega_{\de B_\rho}$ is the pull-back through $v$ of the volume form $\omega_{\de B_\rho}$ on $\de B_\rho$ and the second equality is due to the topological definition of degree. 
}
Also when $v \in H^{\frac{1}{2}}(\de B_\rho(x_0);\S^1)$ is discontinuous, the degree defined in~\eqref{def:degree} inherits from the continuous setting some characterizing properties. In particular, a result due to L.\ Boutet de Monvel \& O.\ Gabber~\cite[Theorem~A.3]{BdMBGeoPur} ensures that\footnote{In~\cite[Theorem~A.3]{BdMBGeoPur} the degree formula is written in an alternative form, equivalent to~\eqref{def:degree}, by interpreting~$v$ as a complex-valued function.} $\deg(v,\de B_\rho(x_0)) \in \Z$.

A further fundamental property of the degree is the following. Let $v \in H^1(\A{r}{R}(x_0);\S^1)$. By the trace theory, $v|_{\de B_{\rho}(x_0)} \in H^{\frac{1}{2}}(\de B_{\rho}(x_0);\S^1)$ for every $\rho \in [r, R]$. Then 
\begin{equation} \label{eq:degree independent}
    \deg(v,\de  B_{\rho}(x_0)) = \deg(v,\de  B_{\rho'}(x_0)) \quad \text{for every } \rho, \rho' \in [r,R] \, .
\end{equation}
This follows from the fact that $\deg(v,\de  B_{\rho}(x_0)) \in \Z$, by the continuity of the degree with respect to the $H^\frac{1}{2}$ norm, and by the continuity of the map
 \begin{equation*}  
     \rho \in [r,R] \mapsto v( x_0 + \rho \, \cdot \, )|_{\de B_{1}} \in H^{\frac{1}{2}}(\de B_{1};\S^1) \, , 
 \end{equation*}
 which is a consequence of the trace theory for Sobolev functions.

We conclude this summary about the degree by recalling the following property. Let $v \in H^1(\A{r}{R}(x_0);\S^1)$. By the theory of slicing of Sobolev functions (\cf~\cite[Proposition~3.105]{AmbFusPal} with a change of coordinates), for a.e.\ $\rho \in (r,R)$ the restriction $v|_{\de B_\rho(x_0)}$ belongs to $H^1(\de B_\rho(x_0); \S^1)$ and $\nabla_{\de B_\rho(x_0)}(v|_{\de B_\rho(x_0)})(y) = \nabla v(y) \tau_{\de B_\rho(x_0)}(y)$ for $\H^1$-a.e.\ $y \in \de B_\rho(x_0)$, where $\tau_{\de B_\rho(x_0)}(y)$ is the unit tangent vector to $\de B_\rho(x_0)$ at $y$. Therefore
 \begin{equation} \label{eq:degree for H1}
     \deg(v, \de B_\rho(x_0)) = \frac{1}{\pi} \int_{\de B_\rho(x_0)} j(v)|_{\de B_\rho(x_0)} \cdot \tau_{\de B_\rho(x_0)} \, \d \H^1  \quad \text{for a.e.\ } \rho \in (r,R) \, ,
 \end{equation}
 which relates the degree to the pre-jacobian and, by Stokes' Theorem, to the distributional Jacobian. 

 \subsection*{An interpolation of discrete spin fields} In the following remark we relate the discrete vorticity of a spin field with the Jacobian of a suitable interpolation. 

    \begin{remark}[$\S^{1}$-interpolation] \label{rmk:arcwise interpolation}
        To every $v \in \SF_\e$ we associate a map $\overarc v \colon \R^2 \to \S^1$ with the following properties:
        \begin{enumerate}
            \item  $\overarc v = v$ on the lattice $\L_\e$;
            \item  $\overarc v \in W^{1,\infty}_\loc(\R^2 \sm \supp \mu_v;\S^1)$; 
            \item $J (\overarc v) = \pi \mu_v$;
            \item $\e^2 \int_T |\nabla \overarc v|^2 \, \d x \leq  \frac{\pi^2}{4\sqrt{3}}  XY_\e(v,T)$ for every $T$ with $\mu_v(T) = 0$. 
        \end{enumerate}
           We define $\overarc v$ in every $T =\conv\{\e \ell_1, \e \ell_2 ,\e \ell_3 \} \in \T_\e(\R^2)$  with $(\e \ell_1,\e \ell_2, \e \ell_3)$ ordered counterclockwise by distinguishing two cases: $\mu_v(T)=0$ or $\mu_v(T) \in \{-1,1\}$. For every vertex $x \in \{\e \ell_1, \e \ell_2, \e \ell_3 \}$ let $\varphi(x)\in \R$ be such that $v(x) = \exp(\iota \varphi(x))$. We set
           \begin{equation*}        
               \phi(\e \ell_1)   := \varphi(\e \ell_1) \,, \quad   \ \phi(\e \ell_2)  := \phi(\e \ell_1) + \d^e \varphi(\e \ell_1, \e \ell_2)  \, , \quad   \ \phi(\e \ell_3)    :=  \phi(\e \ell_2) + \d^e \varphi(\e \ell_2, \e \ell_3)   \, ,
        \end{equation*}
         where $\d^e \varphi$ is defined in~\eqref{def:de}.  In this way,
        \begin{equation}\label{eq:relation phi k phi i}
            \phi(\e \ell_3) + \d^e \varphi(\e \ell_3, \e \ell_1)  = \phi(\e \ell_1) + 2 \pi \mu_v(T)  \, .
        \end{equation}
         If $\mu_v(T)=0$, we let $\hat \phi$ be the function that is affine in $T$ and satisfies $\hat \phi(x) =  \phi(x)$ for every vertex $x \in \{\e \ell_1,\e \ell_2,\e \ell_3\}$. We set $\overarc v(x) := \exp(\iota \hat \phi(x))$ for every $x \in T$. Since $\mu_v(T)=0$,  \eqref{eq:relation phi k phi i} implies that $\phi(\e \ell_1) = \phi(\e \ell_3) + \d^e \varphi(\e \ell_3 , \e \ell_1)$.  Note that
        \begin{equation} \label{eq:arcwise controlled with XY}
            \begin{split}
                \sqrt{3}\int_T |\nabla \overarc v|^2 \, \d x & =  \sqrt{3} \int_T |\nabla \hat \phi|^2 \, \d x = \frac{1}{2} \e^2 \big( |\nabla \hat \phi \cdot \tb_1|^2 + |\nabla \hat \phi \cdot \tb_2|^2 + |\nabla \hat \phi \cdot \tb_3|^2 \big) \\
                & = \frac{1}{2}   \big( |\phi(\e \ell_2) - \phi(\e \ell_1)|^2 + |\phi(\e \ell_3) - \phi(\e \ell_2)|^2 + |\phi(\e \ell_1) - \phi(\e \ell_3)|^2 \big) \\
                & = \frac{1}{2}   \big( |\d^e \varphi(\e \ell_1, \e \ell_2)|^2 + |\d^e \varphi(\e \ell_2, \e \ell_3)|^2 + |\d^e \varphi(\e \ell_3, \e \ell_1)|^2 \big) \\
                & = \frac{1}{2}   \big( \d_{\S^1}(v(\e \ell_1), v(\e \ell_2))^2 + \d_{\S^1}(v(\e \ell_2), v(\e \ell_3))^2  +\d_{\S^1}(v(\e \ell_3), v(\e \ell_1))^2  \big) \\
                & \leq   \frac{\pi^2}{8} \big( |v(\e \ell_2) - v(\e \ell_1)|^2 + |v(\e \ell_3) - v(\e \ell_2)|^2 + |v(\e \ell_1) - v(\e \ell_3)|^2 \big) \\
                &  =   \frac{\pi^2}{4\e^2} XY_\e(v,T) \, .
            \end{split}
        \end{equation}
        We remark that $J( \overarc v) \mres T = 0$ (using the area formula and noticing that the image of $T$ through the smooth map $\overarc v$ is $\S^1$). 
        
        If $\mu_v(T) = z$ with $z \in \{-1,1\}$ we define $\overarc v$ in a different way.        
        Namely, on $\de T$ we define the function $\mathring \phi$ by 
        \begin{equation*}
            \mathring \phi(x) := \begin{cases}
               \phi(\e \ell_1) + s \d^e \varphi(\e \ell_1, \e \ell_2) \, , & \text{for } x = \e \ell_1 + s (\e \ell_2 - \e \ell_1) \, , \ s \in [0,1)  \, , \\
               \phi(\e \ell_2) + s \d^e \varphi(\e \ell_2, \e \ell_3) \, , & \text{for } x = \e \ell_2 + s (\e \ell_3 - \e \ell_2) \, ,\  s \in [0,1)  \, , \\ 
               \phi(\e \ell_3) + s \d^e \varphi(\e \ell_3, \e \ell_1) \, , & \text{for } x = \e \ell_3 + s (\e \ell_1 - \e \ell_3) \, , \ s \in [0,1)  \, .
            \end{cases}
        \end{equation*}
        Let $b(T) \in T$ be the barycenter of $T$. We extend $\mathring \phi$ to $T \sm \{b(T)\}$ making it 0-homogeneous with respect to $b(T)$. Notice that $\mathring \phi$ is continuous outside the segment $[\e \ell_1, b(T)]$, where in view of~\eqref{eq:relation phi k phi i} it has a jump of~$\phi(\e \ell_3) + \d^e \varphi(\e \ell_3, \e \ell_1) - \phi(\e \ell_1) = 2 \pi z$. We define $\overarc v(x) := \exp(\iota \mathring \phi(x))$ for every $x \in T \sm \{b(T)\}$, observing that $\overarc v \in W^{1,1}(T;\S^1)$ and  that $\overarc v \in W^{1,\infty}_{\mathrm{loc}}(T \sm \{b(T)\};\S^1)$. Then the Jacobian of $\overarc v$ is defined in the sense of distributions.  In fact, $J (\overarc v)$ is a measure and $J (\overarc v) \mres T = \pi \mu_v \mres T$.\footnote{The proof of this fact is standard: one can consider for every $\rho \in (0,1)$ the scaled triangle $T^\rho := \rho(T - b(T)) + b(T)$ and  define the ``conical'' approximation  $\overarc v^\rho \colon T \to \R^2$ given by
        \begin{equation*}
            \overarc v^\rho(x) := \begin{cases}
                \overarc v(x) \, , & \text{ if } x \in T \sm T^\rho , \\
                \big(1  - \frac{\dist(x, \de T^\rho)}{\dist(b(T), \de T^\rho)} \big)\overarc v(x) \, , & \text{ if } x \in  T^\rho . \\
            \end{cases}
        \end{equation*}
        On the one hand, $j(\overarc v^\rho) \to j(\overarc v)$ in $L^1(T;\R^2)$ and thus $J(\overarc v^\rho) \wto J(\overarc v)$ in the sense of distributions. On the other hand, $J(\overarc v^\rho) = 0$ in $T \sm T^\rho$ and $\|J(\overarc v^\rho)\|_{L^1(T)} = \|J(\overarc v^\rho)\|_{L^1(T^\rho)} \leq C$, which implies that~$J(\overarc v^\rho)$ converges weakly* to a multiple of the Dirac delta at $b(T)$. Moreover $J(\overarc v^\rho)(T) \to J(\overarc v)(T)$ and
        \begin{equation*}
            \begin{split}
              \int_{T} J(\overarc v^\rho) \, \d x & = \int_{T^\rho} J(\overarc v^\rho) \, \d x  = \int_{T} J(\overarc v^1) \, \d x = \int_{T} \curl(j(\overarc v^1)) \, \d x =  \int_{\de T} j(\overarc v^1) \cdot \tau \, \d \H^1 = \frac{1}{2}\int_{\de T} \nabla \mathring \phi \cdot \tau \, \d \H^1 \\
            &   = \frac{1}{2} \Big( \int_{[\e \ell_1, \e \ell_2]} \frac{\d^e \varphi(\e \ell_1, \e \ell_2) }{\e} \, \d \H^1 +  \int_{[\e \ell_2, \e \ell_3]} \frac{\d^e \varphi(\e \ell_2, \e \ell_3) }{\e} \, \d \H^1 +  \int_{[\e \ell_3, \e \ell_1]} \frac{\d^e \varphi(\e \ell_3, \e \ell_1) }{\e} \, \d \H^1 \Big)   = \pi z \, .
            \end{split}
        \end{equation*}
        } 
    
        The map $\overarc v$ is well-defined. Indeed, it satisfies in any case the following property: if $\varphi \colon \L_\e \to \R$ is such that $v(x) = \exp(\iota \varphi(x))$ for $x \in \L_\e$, then $\overarc v(\e \ell_1 + s (\e \ell_2 - \e \ell_1)) = \exp(\iota \varphi(\e \ell_1)) \exp(\iota s \d^e \varphi(\e \ell_1, \e \ell_2))$ for every $\e \ell_1, \e \ell_2 \in \L_\e$ with $|i-j| = 1$. The curve $s \in [0,1] \mapsto \exp(\iota \varphi(\e \ell_1)) \exp(\iota s \d^e \varphi(\e \ell_1, \e \ell_2))$ parametrizes a geodesic arc in $\S^1$ that connects~$v(\e \ell_1)$ to~$v(\e \ell_2)$.  
    \end{remark}

\subsection*{Flat convergence} We recall here the notion of convergence relevant for the discrete vorticity of spin fields and for Jacobian of maps.  Given a distribution $T \in \mathcal{D}'(U)$, we define its {\em flat norm}\footnote{The name comes from the theory of currents. Interpreting $T \in \mathcal{D}'(U)$ as a 0-current, its flat norm is given by $\mathbb{F}(T) := \inf \{ \mathbb{M}(R) + \mathbb{M}(\de S) : R + \de S = T\}$, where $\mathbb{M}(\cdot)$ denotes the mass. Then, it holds true that $\mathbb{F}(\cdot) = \|\cdot\|_{\mathrm{flat},U}$ (see \cite[4.1.12]{Federer}).} by   
\begin{equation} \label{def:flat}
    \| T \|_{\mathrm{flat},U} := \sup \{ \langle T, \psi \rangle : \psi \in C^\infty_c(U) \, , \ \|\psi\|_{L^\infty(U)} \leq 1 \, , \ \|\nabla \psi\|_{L^\infty(U)}  \leq 1 \} 
\end{equation}   
If $ \| T \|_{\mathrm{flat},U} < \infty $, then the duality $\langle T, \psi \rangle$ can be extended to Lipschitz functions with compact support $\psi \in C^{0,1}_c(U)$. If $T_n$ is a sequence of distributions  such that $\| T_n \|_{\mathrm{flat},U} \to 0$, then $\langle T_n, \psi \rangle \to 0$ for every $\psi \in C^{0,1}_c(U)$.

\subsection*{Lifting of discrete spin fields}  In this subsection we discuss the conditions sufficient to define the lifting of a discrete spin field. 

\begin{lemma} \label{lemma:lifting in annulus}
    Let $x_0 \in \R^2$,  let $v \in \SF_\e$, and let $\overarc v$ be defined as in Remark~\ref{rmk:arcwise interpolation}. Let $0 < r < R$ and assume that $\mu_v(B_{r}(x_0)) = 0$ and $|\mu_{v}|(T)= 0$ for every $T \in \T_\e(\R^2)$ such that $T \cap \A{r}{R}(x_0) \neq 0$.  Then there exists $\phi \in W^{1,\infty}(\A{r}{R}(x_0))$ such that $\overarc v(x) = \exp(\iota \phi(x))$ for all $x \in \A{r}{R}(x_0)$ and 
    \begin{equation} \label{eq:energy of lifting and of v}
         |\nabla \phi(x)|  =  |\nabla \overarc v(x)| \quad \text{ for a.e.\ } x \in \A{r}{R}(x_0) \,. 
    \end{equation}
\end{lemma}
\begin{proof}
    We assume, without loss of generality, that $x_0 = 0$ (the arguments in the proof will never use the fact that $0 \in \L_\e$).  From the fact that $|\mu_{v}|(T)= 0$ for every $T \in \T_\e(\R^2)$ such that $T \cap \A{r}{R}  \neq 0$ and by the definition of $\overarc v$ in Remark~\ref{rmk:arcwise interpolation}, we have that $\overarc v$ is a continuous function in $\A{r}{R}$. We write the annulus $\A{r}{R}$ as the union of the two simply connected sets $S_{r,R}^{0,\frac{3\pi}{2}}$ and $S_{r,R}^{\pi,\frac{5\pi}{2}}$ (see Figure~\ref{fig:lifting}), where
    \begin{equation*}
        S_{r,R}^{\theta_1,\theta_2}  := \{ (\rho \cos \theta, \rho \sin \theta) : \rho \in (r,R) \, , \ \theta \in (\theta_1,\theta_2) \} \,.
    \end{equation*}
    By the simple connectedness of $S_{r,R}^{0,\frac{3\pi}{2}}$ and $S_{r,R}^{\pi,\frac{5\pi}{2}}$, there exist two continuous functions $\phi \colon S_{r,R}^{0,\frac{3\pi}{2}}  \to \R$ and $\phi' \colon S_{r,R}^{\pi,\frac{5\pi}{2}} \to \R$ such that $\overarc v(x) = \exp(\iota \phi(x))$ for $x \in S_{r,R}^{0,\frac{3\pi}{2}}$ and $\overarc v(x) = \exp(\iota \phi'(x))$ for $x \in S_{r,R}^{\pi,\frac{5\pi}{2}}$. We shall prove that $\phi$ and $\phi'$ coincide (up to translating $\phi'$ of an integer multiple of $2\pi$), so that a unique lifting is defined in the annulus~$\A{r}{R}$.

    \begin{figure}[H]
        \includegraphics{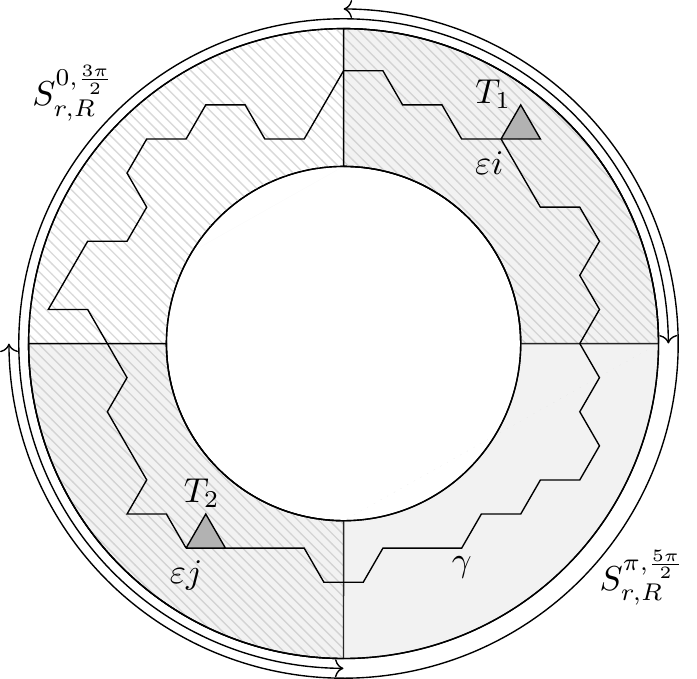}
            \caption{Example of a path $\gamma$ used in the proof.}
            \label{fig:lifting}
    \end{figure}

    We observe that, since $\overarc v \in W^{1,\infty}(\A{r}{R};\S^1)$, also $\phi \in W^{1,\infty}(S_{r,R}^{0,\frac{3\pi}{2}})$ and $\phi' \in W^{1,\infty}(S_{r,R}^{\pi,\frac{5\pi}{2}})$. Moreover,  by the chain rule,
    \begin{equation*} 
        2 j(\overarc v)(x) = \nabla \phi(x) \quad \text{for } x \in S_{r,R}^{0,\frac{3\pi}{2}} \, , \quad 2 j(\overarc v)(x) = \nabla \phi'(x) \quad \text{for } x \in S_{r,R}^{\pi,\frac{5\pi}{2}} \, .
    \end{equation*}
    
    By the uniqueness of the lifting up to integer multiples of $2 \pi$, there exist $z_1 \in \Z$ and $z_2 \in \Z$ such that 
    \begin{equation*}
        \phi(x)  = \phi'(x) + 2 \pi z_1 \, , \text{ for every } x \in S_{r,R}^{0,\frac{\pi}{2}}  \, , \quad 
        \phi(x)  = \phi'(x) + 2 \pi z_2 \, , \text{ for every } x \in S_{r,R}^{\pi,\frac{3\pi}{2}} \, .
    \end{equation*}
    Let us prove that $z_1 = z_2$ by exploiting the assumption $\mu_v(B_{r}) = 0$. Let $T_1 \in \T_\e(S_{r,R}^{0,\frac{\pi}{2}})$ with a vertex $\e i$ and let $T_2 \in \T_\e(S_{r,R}^{\pi,\frac{3\pi}{2}})$ with a vertex $\e j$. Let $\gamma \colon [0,1] \to \A{r}{R}$ be a path such that $\supp(\gamma)$ is the union of edges in the triangular lattice, specifically, $\supp(\gamma) = \bigcup_{h=1}^{M} [\e \ell_{h-1}, \e \ell_{h}]$ with $|\ell_{h} - \ell_{h-1}| = 1$. Assume that $\e \ell_0 = \e \ell_{M} = \e i$, $\e \ell_N = \e j$ with $N < M$, $\bigcup_{h=1}^{N} [\e \ell_{h-1}, \e \ell_{h}] \subset S_{r,R}^{0,\frac{3\pi}{2}}$, and $\bigcup_{h=N+1}^{M} [\e \ell_{h-1}, \e \ell_{h}] \subset  S_{r,R}^{\pi,\frac{5\pi}{2}}$. Moreover, assume that~$\gamma$ is the oriented boundary of an open set $U \subset B_{R}$ with $0 \in U$ (See Figure~\ref{fig:lifting}). Then, by Stokes' Theorem, by the definition of Jacobian, and by Remark~\ref{rmk:arcwise interpolation} we infer that
    \begin{equation*}
        \begin{split}
        2 \pi (z_2 - z_1) & = \phi'(\e i) - \phi(\e i) + 2 \pi z_2 = \phi'(\e i) - \phi'(\e j) + \phi(\e j) - \phi(\e i)\\
         & = \sum_{h=1}^N \phi'(\e \ell_h) - \phi'(\e \ell_{h-1}) + \sum_{h=N+1}^{M} \phi(\e \ell_h) - \phi(\e \ell_{h-1}) \\
         & = \sum_{h=1}^N \int_{[\e \ell_{h-1}, \e \ell_{h}]} \nabla \phi'(x) \cdot \tau(x) \, \d \H^1(x) + \sum_{h=N+1}^{M}\int_{[\e \ell_{h-1}, \e \ell_{h}]} \nabla \phi(x) \cdot \tau(x) \, \d \H^1(x) \\
         & = 2 \sum_{h=1}^{M} \int_{[\e \ell_{h-1}, \e \ell_{h}]} j(\overarc v)(x) \cdot \tau(x) \, \d \H^1(x) = 2 \int_\gamma j(\overarc v) \cdot \tau  \, \d \H^1 = 2 \int_U J(\overarc v)\, \d x = 2 \pi \mu(U) \, .
        \end{split} 
    \end{equation*}
    Since $\mu(U) = 0$, we conclude that $z_1 = z_2$. Therefore, we can extend $\phi$ to the whole annulus $\A{r}{R}$ by setting $\phi(x) := \phi'(x) + 2\pi z_1$ for every $x \in S_{r,R}^{\pi,\frac{5\pi}{2}}$. It satisfies~\eqref{eq:energy of lifting and of v} by the chain rule.
\end{proof}

\begin{remark} \label{rmk:properties of lifting}
    We point out some properties of the lifting $\phi$ of $\overarc v$ provided by Lemma~\ref{lemma:lifting in annulus}. 
    
    The first property is the following:  
    \begin{equation} \label{eq:geodesic distance is angle}
        \d_{\S^1}(v(\e i), v(\e j)) = |\phi(\e i) - \phi(\e j)| \, ,   \quad \text{for every } \e i,\e j \in \L_\e \cap T \text{ with } T \in \T_\e(\A{r}{R}(x_0)) \, . 
    \end{equation}
    Indeed, on the one hand, we have that $\overarc v(\e i + s (\e j- \e i)) = \exp(\iota \phi(\e i + s (\e j- \e i)))$  for $s \in [0,1]$. On the other hand, since $\mu_v(T) = 0$, by Remark~\ref{rmk:arcwise interpolation}, the map $\overarc v$ is given by $\overarc v(\e i + s (\e j- \e i)) = \exp(\iota \phi(\e i)) \exp(\iota s \d^e\phi(\e i, \e j))$ for $s \in [0,1]$. Hence there exists $z \in \Z$ such that
    \begin{equation*}
        \phi(\e i + s (\e j- \e i)) = \phi(\e i) + s \, \d^e \phi(\e i, \e j) + 2 \pi z \, , \quad \text{for every } s \in [0,1] \, .
    \end{equation*}
    Evaluating the previous formula at $s=0$, we infer that $z=0$; evaluating it at $s=1$, we obtain that $\phi(\e j) - \phi(\e i) = \d^e \phi(\e i, \e j) \in [-\pi,\pi]$. In particular, $\d_{\S^1}(v(\e i), v(\e j)) = |\phi(\e j) - \phi(\e i)|$, which concludes the proof of~\eqref{eq:geodesic distance is angle}.
    
    The second property is the following: let $T = \conv \{\e i, \e j, \e k\} \in \T_\e(\A{r}{R}(x_0))$. Then $\phi|_T$ is an affine function. Indeed, the previous property yields that $\phi(\e j) = \phi(\e i) + \d^e \phi(\e i, \e j)$, $\phi(\e k) = \phi(\e j) + \d^e \phi(\e j, \e k)$. Since $\mu_v(T) = 0$, by the definition of $\overarc v$ in Remark~\ref{rmk:arcwise interpolation} we have that $\overarc v(x) = \exp(\hat \phi(x))$ for $x \in T$, where $\hat \phi$ is the affine function in $T$ such that $\hat \phi(x) = \phi(x)$ for every vertex $x \in \{\e i,\e j,\e k\}$. Then $\hat \phi(x) = \phi(x)$ for every $x \in T$. In particular, from~\eqref{eq:geodesic distance is angle} and~\eqref{eq:energy of lifting and of v}, we deduce that
    \begin{equation}\label{est:XY-Gradient}
    \frac{1}{\e^2}XY_\e(v,T)\leq\sqrt{3}\int_{T}|\nabla\phi|^2\dx=\sqrt{3}\int_{T}|\nabla\overarc{v}|^2\dx\leq\frac{\pi^2}{4\e^2}XY_\e(v,T)\,,
    \end{equation}
    where the last inequality follows from Remark~\ref{rmk:arcwise interpolation}.
\end{remark}

\subsection*{Extension of discrete spin fields} We prove now an extension lemma. It is the discrete version of a standard result in the continuum, which states the following: if $v \in H^1(\A{r}{R}(x_0);\S^1)$ satisfies $\deg(v,\de B_\rho(x_0)) = 0$, then it can be extended to a $\overline v \in H^1(B_R(x_0;\S^1))$ such that $\int_{B_R(x_0)} |\nabla \overline v|^2 \d x \leq C \int_{\A{r}{R}(x_0)} |\nabla v|^2 \d x$. In the proof we exploit the interpolation introduced in Remark~\ref{rmk:arcwise interpolation}. 

\begin{lemma}\label{lemma:extension}
There exists a universal constant $C_0>0$ such that the following holds true. Let $\e>0$, $x_0\in\R^2$, and $R>r>\e$, let $C_1>1$ and $v_\e \in \SF_\e$ with  $XY_\e(v_\e,\A{r}{R}(x_0)) \leq C_1 \e^2$, $\mu_{v_\e}(B_{r}(x_0)) = 0$, and $|\mu_{v_\e}|(T)= 0$ for every $T \in \T_\e(\R^2)$ such that $T \cap \A{r}{R}(x_0) \neq 0$. Then there exists  $\overline v_\e \in \SF_\e$  such that for $\e<\frac{R-r}{C_0C_1} \big(\frac{2\pi}{3}\big)^2$:
    \begin{itemize}
        \item  $\overline v_\e = v_\e$ on $\L_\e \cap(\R^2 \sm \overline B_{\frac{r+R}{2}}(x_0))$;  
        \item $|\mu_{\overline v_\e}|( B_{R}(x_0)) = 0$;
        \item $XY_\e(\overline v_\e, B_{R}(x_0)) \leq C(r,R) XY_\e(v_\e, \A{r}{R}(x_0))$, where $C(r,R) = C_0 \frac{R}{R-r}$.  
    \end{itemize}
\end{lemma}

\begin{remark}\label{rmk:extension}  If there exists $\beta>1$ such that $R=\beta r$, then the extension constant $C(r,R)$  given in Lemma~\ref{lemma:extension} is independent of $r$. Indeed, $C(r,R) = C_0 \frac{\beta}{\beta-1} =: C(\beta)$.  Moreover, the extension $\overline v_\e$ satisfies the properties in the statement for $\frac{\e}{r} <\frac{\beta-1}{C_0C_1} \big(\frac{2\pi}{3}\big)^2$.  
    
    The lemma is stated for $\e$ fixed, thus the result can be applied also when $r = r_\e$ and $R = R_\e$.
    
    The particular geometry of the triangular lattice does not play a major role in the proof: it is crucial that \ref{est:XY-Gradient} holds true. For instance an analogous result holds true in the case of the square lattice. 
\end{remark}

\begin{proof}[Proof of Lemma \ref{lemma:extension}] 
    We assume, without loss of generality, that $x_0 = 0$  (the arguments in the proof will never use the fact that $0 \in \L_\e$) and $2\e \leq \frac{R-r}{36}$  (Note that this is {\em a fortiori} satisfied if~$C_0$ is chosen sufficiently large in $\e <\frac{R-r}{C_0C_1} \big(\frac{2\pi}{3}\big)^2 < \frac{R-r}{C_0} \big(\frac{2\pi}{3}\big)^2$).  Let  $\overarc v_\e \in W^{1,\infty}(\A{r}{R};\S^1)$  be defined as in Remark~\ref{rmk:arcwise interpolation}. By Lemma~\ref{lemma:lifting in annulus} there exists $\phi_\e \in W^{1,\infty}(\A{r}{R})$ such that $\overarc v_\e(x) = \exp(\iota \phi_\e(x))$ for all $x \in \A{r}{R}$ and~\eqref{eq:energy of lifting and of v} holds true. To define $\overline v_\e$, we start by extending $\phi_\e$ from the annulus $\A{r}{R}$ to a function $\phi_\e'$ on the ball~$B_{R}$ via a 1-homogeneous extension that starts from a layer of triangles suitably chosen inside~$\A{r}{R}$. More precisely, we fix  $r < r' < R' < \frac{r+R}{2} <  R$ such that $R'-r' \geq \frac{R-r}{4}$ and we subdivide $\A{r'}{R'}$ into the union of annuli 
    \begin{equation*}
        \A{r'}{R'} = \bigcup_{k=1}^{K_\e} A^k \, , \quad A^k := \A{r_{k-1}}{r_k} \, , \quad r_k := r'+k\frac{R'-r'}{K_\e} \, ,\quad  K_\e := \lfloor \tfrac{R'-r'}{9\e} \rfloor  \, .
    \end{equation*} 

    Note that $K_\e \geq \frac{R-r}{36 \e}-1\geq\frac{1}{2}\frac{R-r}{36\e} \geq 1$ and the width of each annulus $A^k$ is $\frac{R'-r'}{K_\e} \geq 9 \e$. For every $\e$ we find~$k_\e \in \{1,\dots, K_\e\}$ such that 
    \begin{equation}  \label{eq:averaging on annulus}
        \begin{split}
            C_1 & \geq \frac{1}{\e^2} XY_\e(v_\e,\A{r}{R})  \geq  C  \int_{\A{r'}{R'}} |\nabla \overarc v_\e|^2 \, \d x  =  C \sum_{k=1}^{K_\e} \int_{A^k} |\nabla \overarc v_\e|^2 \, \d x \\
            & \geq  C K_\e \int_{A^{k_\e}} |\nabla \overarc v_\e|^2 \, \d x  \geq C\frac{R-r}{\e^3} XY_\e(v_\e, A^{k_\e}) \, ,
        \end{split}
    \end{equation}
     where the second and the last inequality follow from~\eqref{est:XY-Gradient}. 
    Recalling the first property in Remark~\ref{rmk:properties of lifting}, an immediate consequence of the previous inequality is that
    \begin{equation}  \label{eq:close angles}
        |\phi_\e(\e i) - \phi_\e(\e j)|^2  \leq C |v_\e(\e i) - v_\e(\e j)|^2 \leq \frac{C}{\e^2}  XY_\e(v_\e, A^{k_\e}) \leq
          \frac{C}{(R-r)\e}XY_\e(v_\e,\A{r}{R})\leq  C\frac{C_1\e}{R-r}  
    \end{equation}
    for every $\e i, \e j \in A^{k_\e}$ with $|i-j|=1$.

Using that $\phi_\e \in W^{1,\infty}(\A{r}{R})$,\footnote{ Notice that $\phi_\e |_{\de B_\rho} \in W^{1,\infty}(\de B_\rho)$ for every $\rho \in (r',R')$: the function~$\phi_\e$ is piecewise affine by Remark~\ref{rmk:arcwise interpolation} and thus it is $C^1$ outside a finite union of segments, which intersect $\partial B_\rho$ only in a finite number of points (depending on~$\e$).  \label{foot:actually every radius}}   we find  $\rho_\e \in (r_{k_\e-1} + 4\e , r_{k_\e} -4\e )$  such that the restriction $y \in \de B_{\rho_\e} \mapsto \phi_\e |_{\de B_{\rho_\e}}(y)$ belongs to  $W^{1,\infty}(\de B_{\rho_\e})$,  $\nabla_{\de B_{\rho_\e}}(\phi_\e |_{\de B_{\rho_\e}})(y) = \nabla \phi_\e(y) \cdot \tau_{\de B_{\rho_\e}}(y)$ for $\H^1$-a.e.\ $y \in \de B_{\rho_\e}$, and
    \begin{equation*}
        \begin{split}
            \int_{A^{k_\e}} |\nabla \overarc v_\e|^2 \, \d x & = \int_{A^{k_\e}} |\nabla \phi_\e|^2 \, \d x \geq \int_{r_{k_\e-1} + 4\e }^{r_{k_\e} - 4 \e} \int_{\de B_{\rho}} |\nabla \phi_\e(y)\cdot \tau_{\de B_{\rho}}(y)|^2 \, \d \H^1(y) \, \d \rho \\
            & \geq  \Big( \frac{R'-r'}{K_\e} - 8\e \Big) \int_{\de B_{\rho_\e}} |\nabla \phi_\e(y)\cdot \tau_{\de B_{\rho_\e}}(y)|^2 \, \d \H^1(y)  \\
            & \geq  \e  \int_{\de B_{\rho_\e}} \big|\nabla_{\de B_{\rho_\e}}(\phi_\e |_{\de B_{\rho_\e}})(y)\big|^2 \, \d \H^1(y)   \, .
        \end{split}
    \end{equation*}
    In particular, by the previous formula and by~\eqref{eq:averaging on annulus}, we infer that
    \begin{equation} \label{eq:bound on tangential derivative}
        \int_{\de B_{\rho_\e}} \big|\nabla_{\de B_{\rho_\e}}(\phi_\e |_{\de B_{\rho_\e}})\big|^2 \, \d \H^1  \leq   \frac{1}{\e^3 K_\e} XY_\e(v_\e,\A{r}{R})  \leq   \frac{C}{(R-r)\e^2} XY_\e(v_\e,\A{r}{R})\leq C\frac{C_1}{R-r}\,.
    \end{equation} 
    Setting $a_\e := \fint_{\de B_{\rho_\e}} \phi_\e \, \d \H^1$, by Poincare's Inequality on $\de B_{\rho_\e}$ we have that 
    \begin{equation} \label{eq:poincare on circle}
        \int_{\de B_{\rho_\e}} \big|\phi_\e|_{\de B_{\rho_\e}} - a_\e \big|^2 \, \d \H^1 \leq C \rho_\e^2 \int_{\de B_{\rho_\e}} \big|\nabla_{\de B_{\rho_\e}}(\phi_\e |_{\de B_{\rho_\e}}) \big|^2 \, \d \H^1 ,
    \end{equation}
    for a scale-independent constant $C$. By \eqref{eq:bound on tangential derivative} and the Sobolev Embedding Theorem in one dimension, we have that $\phi_\e$ is $\frac{1}{2}$-H\"older continuous and 
    \begin{equation} \label{eq:holder continuity}
        \begin{split}
        \|\phi_\e-a_\e\|_{L^\infty(\de B_{\rho_\e})}^2 & \leq  \frac{C\rho_\e}{(R-r)\e^2} XY_\e(v_\e,\A{r}{R})  \leq C\frac{C_1}{R-r} \rho_\e \,, \\
         \sup_{\substack{x,y \in \de B_{\rho_\e} \\ x \neq y}} \frac{|\phi_\e(x) - \phi_\e(y)|^2}{|x-y|} & \leq  \frac{C}{(R-r)\e^2} XY_\e(v_\e,\A{r}{R})  \leq C\frac{C_1}{R-r} \, . 
        \end{split}
    \end{equation}
    The function $\phi_\e|_{\de B_{\rho_\e}}$ is also Lipschitz continuous, but its Lipschitz constant might depend on $\e$. We define the auxiliary function $\phi_\e' \in W^{1,\infty}(B_R)$ via the 1-homogeneous extension
    \begin{equation*}
        \phi_\e'(x)  := \begin{cases}
            \phi_\e(x) \, , & \text{if } x \in \A{\rho_\e}{R} \, , \\
            a_\e + \frac{|x|}{\rho_\e} \big( \phi_\e\big( \frac{x}{|x|} \rho_\e \big)  - a_\e \big) \, , & \text{if } x \in \overline B_{\rho_\e} \, ,
        \end{cases} \quad  
        v_\e'(x)  := \exp( \iota \phi_\e'(x) ) \quad \text{for } x \in B_{R} \, .
    \end{equation*} 
    Note that $v_\e' = \overarc v_\e$ in $\A{\rho_\e}{R}$. To define the spin field $\overline v_\e$, we suitably sample $v_\e'$. Applying Lemma~\ref{lemma:sampling}  below,  we find $\overline x_\e \in \R^2$  with $|\overline x_\e| \leq \e$ such that  
    \begin{equation} \label{eq:good sampling}
        \frac{1}{\e^2} XY_\e(v_\e'(\, \cdot \, + \overline x_\e), B_{\rho_\e}) \leq C \int_{B_{R'}} |\nabla v_\e'|^2 \d x 
    \end{equation}
    and for $\e i \in \L_\e \cap B_R$ we set 
    \begin{equation*}
        \overline \phi_\e(\e i)  := \begin{cases}
             \phi_\e'(\e i)   \, , & \text{if } \e i \in \A{\rho_\e}{R} \, , \\
           \phi_\e'(\e i + \overline x_\e) \, , & \text{if } \e i \in \overline B_{\rho_\e} \, ,
        \end{cases} \quad  
        \overline v_\e(\e i)  := \exp( \iota \overline \phi_\e( \e i) ) \quad \text{for } \e i \in B_{R} \, .
    \end{equation*}
     We extend $\overline v_\e$ outside $B_R$ by setting $\overline v_\e(\e i) := v_\e(\e i)$ for $\e i \in \R^2 \sm B_R$. By construction we have that $\overline v_\e =  v_\e$ on $\L_\e \cap (\R^2 \sm \overline B_{\rho_\e})$  and thus on $\L_\e \cap (\R^2 \sm \overline B_{\frac{r+R}{2}})$.  
    
 Let us prove that $|\mu_{\overline v_\e}|( B_{R}) = 0$.  Let $T = \conv \{\e i, \e j, \e k\}$ be such that $T \cap B_{R} \neq \emptyset$. If $T \subset \R^2 \sm \overline B_{\rho_\e}$, there is nothing to prove, as $\overline v_\e = v_\e$ on $\L_\e \cap T$ and $|\mu_{v_\e}|(T)= 0$ (since $T \cap \A{r}{R} \neq \emptyset$). Let us thus assume that $T \cap \overline B_{\rho_\e} \neq \emptyset$ and let us prove that $|\overline \phi_\e(\e i) - \overline \phi_\e(\e j)|, |\overline \phi_\e(\e j) - \overline \phi_\e(\e k)|, |\overline \phi_\e(\e k) - \overline \phi_\e(\e i)| < \frac{2\pi}{3}$, which implies $|\mu_{\overline v_\e}|(T) = 0$. We only prove it for $|\overline \phi_\e(\e i) - \overline \phi_\e(\e j)|$, the other inequalities being analogous. We start by observing that    
\begin{equation} \label{eq:3 phi}
|\overline{\phi}_\e(\e i)- \overline{\phi}_\e(\e j)|^2 \leq 3|\phi_\e'(\e i+\overline{x}_\e)- \phi_\e'(\e i)|^2 + 3|\phi_\e'(\e j)- \phi_\e'(\e i)|^2 + 3|\phi_\e'(\e j+\overline{x}_\e)- \phi_\e'(\e j)|^2
\end{equation}
and $\e i,\e i + \overline{x}_\e, \e j, \e j + \overline{x}_\e \in \overline B_{\rho_\e +2\e} $.
 Therefore, to conclude it is enough to prove that for all $x,y \in \overline{B}_{\rho_\e +2\e}$ such that $|x-y| \leq \e$  we have
    \begin{equation}\label{ineq:phiprimehoelder}
  |\phi_\e'(x)-\phi_\e'(y)|^2 \leq \frac{C}{(R-r)\e}XY_\e(v_\e,\A{r}{R})\leq C\frac{C_1\e}{R-r}
    \end{equation}
    from which it follows that
    \begin{equation}\label{ineq:philess2pi3}
    |\overline{\phi}_\e(\e i)- \overline{\phi}_\e(\e j)|^2 \leq  9 C  \frac{C_1\e}{R-r} < \Big(\frac{2\pi}{3}\Big)^{\! 2}
    \end{equation}
    for $\e < \big(\frac{2\pi}{3}\big)^2 \frac{R-r}{C_0C_1}$ and $C_0 > 9C$. To prove \eqref{ineq:phiprimehoelder} we distinguish three cases.
    
        {\em Case 1}: $x,y \in \overline{B}_{\rho_\e+2\e}\setminus B_{\rho_\e}$. Since $B_{\rho_\e+4\e}\setminus \overline{B}_{\rho_\e-4\e} \subset A^{k_\e}$, we find  $T',T''  \in \mathcal{T}_\e(A^{k_\e})$ such that $x \in T'$, $y \in T''$,  and $T' \cap T'' \neq \emptyset$. Let $z \in T'\cap T''$. Since $\phi_\e|_{T'}$ and $\phi_\e|_{T''}$ are affine and using~\eqref{eq:close angles} we obtain that
        \begin{equation}\label{eq:07101415}
        | \phi_\e'(x) -  \phi_\e'(y)|^2 =  | \phi_\e(x) -  \phi_\e(y)|^2 \leq 2| \phi_\e(x) -  \phi_\e(z)|^2 + 2| \phi_\e(z) -  \phi_\e(y)|^2  \leq   \frac{C}{(R-r)\e}XY_\e(v_\e,\A{r}{R}) \,.
        \end{equation}
        
        {\em Case 2}: $x,y \in \overline{B}_{\rho_\e}$.  By~\eqref{eq:holder continuity} we get that 
   \begin{equation*} 
        \begin{split}
            & | \phi_\e'(x) -  \phi_\e'(y)|^2 =  \big|\tfrac{|x|}{\rho_\e} \big(\phi_\e\big( \tfrac{x}{|x|} \rho_\e \big)-a_\e \big) - \tfrac{|y|}{\rho_\e} \big(\phi_\e\big( \tfrac{y}{|y|} \rho_\e \big)-a_\e\big)\big|^2\\
            & \quad \leq 2 \big| \tfrac{|x|-|y|}{\rho_\e} \big|^2 \|\phi_\e-a_\e\|_{L^\infty(\de B_{\rho_\e})}^2 + 2 \tfrac{|x|^2}{\rho_\e^2} \big| \phi_\e\big( \tfrac{x}{|x|} \rho_\e \big) - \phi_\e\big( \tfrac{y}{|y|} \rho_\e \big) \big|^2 \\
            & \quad \leq \Big(\frac{\e^2}{\rho_\e} +\e\Big) \frac{C}{(R-r)\e^2}XY_\e(v_\e,\A{r}{R})  \leq   \frac{C}{(R-r)\e}XY_\e(v_\e,\A{r}{R})   \,.
        \end{split}
    \end{equation*}
        
          {\em Case 3}: $x \in  \overline{B}_{\rho_\e}$, $y \in \A{\rho_\e}{R}$. We find $z \in \partial B_{\rho_\e}$ such that $|x-z|+|y-z| = |x-y|$. Using Case~1 for $x,z$ and Case~2 for $z,y$ we obtain
          \begin{equation}\label{eq:07101417}
          |\phi_\e'(x) -\phi_\e'(y)|^2 \leq 2|\phi_\e'(x) -\phi_\e'(z)|^2 + 2|\phi_\e'(z) -\phi_\e'(y)|^2 \leq \frac{C}{(R-r)\e}XY_\e(v_\e,\A{r}{R})   \,.
          \end{equation}

    This concludes the proof of the fact that $|\mu_{\overline v_\e}|(T) = 0$. For the next estimates it is worth to mention that, as a byproduct of~\eqref{eq:3 phi} and \eqref{eq:07101415}--\eqref{eq:07101417}, we also obtain that
    \begin{equation}  \label{eq:T at boundary}
        \frac{1}{\e^2} XY_\e(\overline v_\e, T) \leq \frac{C}{(R-r)\e} XY_{\e}(v_\e,\A{r}{R}) \quad \text{for every } T \in \T_\e(B_R) \text{ such that } T \cap \de B_{\rho_\e} \neq \emptyset \, .
    \end{equation}  

    It remains to prove that $XY_\e(\overline v_\e, B_{R}) \leq C(r,R) XY_\e(v_\e, \A{r}{R})$.  First of all, we observe that~\eqref{eq:good sampling} and the definition of $v_\e'$ imply
    \begin{equation*}
        \frac{1}{\e^2}XY_\e(\overline v_\e, B_{\rho_\e}) \leq C \int_{B_{\rho_\e}} |\nabla v_\e'|^2 \d x + C \int_{\A{\rho_\e}{R'}} |\nabla \overarc v_\e|^2 \d x  \leq C \int_{B_{\rho_\e}} |\nabla v_\e'|^2 \d x +   \frac{C}{\e^2}XY_\e(v_\e, \A{r}{R})  \, .
    \end{equation*}
    Let us estimate $\int_{B_{\rho_\e}} |\nabla v_\e'|^2 \d x $. 
    Using that $\nabla \phi_\e'(x) = \tfrac{1}{\rho_\e}\tfrac{x}{|x|} \big( \phi_\e\big( \tfrac{x}{|x|} \rho_\e \big)  - a_\e \big) + \nabla \phi_\e\big( \tfrac{x}{|x|} \rho_\e \big)  \tfrac{\, x^\perp}{|x|}  \otimes \tfrac{\, x^\perp}{|x|}$, by Fubini's Theorem, by~\eqref{eq:poincare on circle}, and by~\eqref{eq:bound on tangential derivative}, we obtain that
    \begin{equation*}
        \begin{split}
         \int_{B_{\rho_\e}} |\nabla v_\e'|^2 \, \d x  & = \int_{B_{\rho_\e}} |\nabla \phi_\e'|^2 \, \d x  = \int_{B_{\rho_\e}} \big|\nabla \phi_\e'(x) \cdot \tfrac{x}{|x|}\big|^2 + \big|\nabla \phi_\e'(x)  \cdot \tfrac{\, x^\perp}{|x|}\big|^2\, \d x  \\
            &   = \int_{B_{\rho_\e}} \tfrac{1}{\rho_\e^2} \big|  \phi_\e\big( \tfrac{x}{|x|} \rho_\e \big)  - a_\e  \big|^2 + \big|\nabla \phi_\e\big( \tfrac{x}{|x|} \rho_\e \big)  \cdot \tfrac{\, x^\perp}{|x|}\big|^2\, \d x  \\
            &   = \int_0^{\rho_\e}\int_{\de B_{\rho_\e}} \Big( \tfrac{1}{\rho_\e^2} \big|  \phi_\e|_{\de B_{\rho_\e}}(y)  - a_\e  \big|^2 +  \big|\nabla_{\de B_{\rho_\e}} (\phi_\e|_{\de B_{\rho_\e}})(y) \big|^2 \Big)\tfrac{\rho}{\rho_\e} \, \d \H^1(y) \, \d \rho  \\
            &   \leq CR  \int_{\de B_{\rho_\e}} \Big( \tfrac{1}{\rho_\e^2} \big|  \phi_\e|_{\de B_{\rho_\e}}(y)  - a_\e  \big|^2 +  \big|\nabla_{\de B_{\rho_\e}} (\phi_\e|_{\de B_{\rho_\e}})(y) \big|^2 \Big)  \, \d \H^1(y)  \\
            & \leq CR  \int_{\de B_{\rho_\e}}  \big|\nabla_{\de B_{\rho_\e}} (\phi_\e|_{\de B_{\rho_\e}})(y) \big|^2    \, \d \H^1(y) \leq
            \frac{CR}{(R-r)\e^2} XY_\e(v_\e, \A{r}{R})\, .
        \end{split}
    \end{equation*}
     To conclude, let us estimate the energy on triangles that are not contained in $B_{\rho_\e}$. Let us fix $T =\conv\{\e i, \e j, \e k\} \in \T_\e(B_R)$. If $T \subset \A{\rho_\e}{R}$, then $\overline v_\e = v_\e$ on $\L_\e \cap T$ and thus $XY_\e(\overline v_\e, T) = XY_\e(v_\e, T)$. Finally, since $\#\{T \in \T_\e(B_r) : T \cap \de B_{\rho_\e}\neq \emptyset\} \leq \frac{C\rho_\e}{\e} \leq \frac{C R}{\e}$, inequality~\eqref{eq:T at boundary} yields
    \begin{equation*}
        \sum_{\substack{T \in \T_\e(B_R)\\ T \cap \de B_{\rho_\e} \neq \emptyset }} \frac{1}{\e^2} XY_{\e}(\overline v_\e,T) \leq \frac{CR}{(R-r)\e^2} XY_\e(v_\e,\A{r}{R})  \, .
    \end{equation*}
    Putting together the previous estimates yields $XY_\e(\overline v_\e, B_{R}) \leq \frac{CR}{(R-r)} XY_\e(v_\e, \A{r}{R})$. Choosing $C_0\geq C$ such that  \eqref{ineq:philess2pi3} is satisfied yields the statement of the lemma.      
\end{proof}

We prove a lemma concerning sampling of $H^1$ functions used in the previous proof. 
 
\begin{lemma} \label{lemma:sampling}
    Let $T_0 := \conv\{0,\e \tb_1 , \e \tb_2\}$. There exists a universal constant $C > 0$ such that the following holds true: given $U \subset \R^2$, $v \in H^1(U;\S^1)$, $U' \subset \subset U$ and $\e>0$ with $\dist(U',\de U) > \e$, there exists a point $\overline x \in T_0$ (possibly depending on $U'$) such that 
    \begin{equation*}
        \frac{1}{\e^2} XY_\e(v(\, \cdot \, + \overline x), U') \leq C \int_{U} |\nabla v|^2 \d x \, .
    \end{equation*}
\end{lemma}
\begin{proof} 
Let $T = \conv\{\e i, \e j , \e k\}\in \T_\e(U')$. Note that, if $T' \cap T \neq \emptyset$, then $T' \subset U$. Moreover, by the theory of slicing of Sobolev functions (\cf~\cite[Proposition~3.105]{AmbFusPal}), for a.e.\ $x \in T_0$ (actually, for $\H^1$-a.e.\ $x \in T_0$) we have that $v|_{[x + \e i, x + \e j]} \in H^1([x + \e i, x + \e j];\S^1)$ and $\frac{\d}{\d t} v \big(x+\e i + t(\e j- \e i)\big) = \nabla v\big(x+\e i + t(\e j- \e i)\big)  (\e j-\e i) $  for $t\in(0,1)$ . Thus, by Jensen's Inequality and Fubini's Theorem,
\begin{equation*}
    \begin{split}
        \int_{T_0} |v(\e i + x) - v(\e j + x)|^2 \, \d x 
        & = \int_{T_0} \Big| \int_0^1 \nabla v\big(x+\e i + t(\e j- \e i)\big)  (\e j-\e i) \, \d t \Big|^2 \, \d x \\
        & \leq \int_{T_0}   \int_0^1 \e^2 \big| \nabla v\big(x+\e i + t(\e j- \e i)\big) \big|^2 \, \d t   \, \d x  \\
        &      \leq \e^2   \int_{T+T_0}  \big| \nabla v(y) \big|^2  \, \d y \leq \e^2 \sum_{T' \cap T \neq \emptyset} \int_{T'}  \big| \nabla v(y) \big|^2  \, \d y \, .
    \end{split}
    \end{equation*}
    Arguing analogously with the other vertices of the triangle, we get that for every $T \in \T_\e(U')$
    \begin{equation*}
        \int_{T_0} \frac{1}{\e^2} XY_\e(v(\, \cdot \, + x),T) \, \d x    \leq  3  \e^2   \sum_{T' \cap T \neq \emptyset} \int_{T'}  \big| \nabla v(y) \big|^2  \, \d y \, .
    \end{equation*}
    Hence, there exists $\overline x \in T_0$ such that
\begin{equation*}
    \begin{split}
        XY_\e(v(\, \cdot \, + \overline x),U')  & \leq \int_{T_0} \frac{4}{\sqrt{3}\e^2}XY_\e(v( \, \cdot \, + x),U') \, \d x = \sum_{T \in \T_\e(U')} \frac{4}{\sqrt{3}} \int_{T_0} \frac{1}{\e^2} XY_\e(v(\, \cdot \, + x),T) \, \d x \\
        & \leq \sum_{T \in \T_\e(U')} 4\sqrt{3}   \e^2   \sum_{T' \cap T \neq \emptyset} \int_{T'}  \big| \nabla v(y) \big|^2  \, \d y \leq  52\sqrt{3}  \e^2   \int_{U}  \big| \nabla v(y) \big|^2  \, \d y \, ,
    \end{split}
\end{equation*} 
where we used that each triangle $T'$ is counted at most 13 times. This concludes the proof.
\end{proof}

\subsection*{Relations between chirality and vorticity} We describe here the relations between the chirality of $u \in \SF_\e$ and the vorticity of the auxiliary spin field $v \in \SF_\e$ defined as in~\eqref{def:from u to v}.

\begin{remark} \label{rmk:chirality and vorticity}
    If the chirality of $u$ is close enough to 1 in a triangle of the lattice, there the auxiliary spin field $v$ cannot have vorticity. More precisely, there exists $\eta \in (0,1)$ such that $\mu_v(T) = 0$ for every $u \in \SF_\e$ and $T \in \T_\e(\R^2)$ with $\chi(u,T) > 1 - \eta$. Indeed, by Remark~\ref{rmk:chirality angles} there exists $\eta \in (0,1)$ such that $|\theta_1 - \frac{2\pi}{3}| < \frac{\pi}{2}$ and $|\theta_2 + \frac{2\pi}{3}| < \frac{\pi}{2}$ for every $\theta_1, \theta_2 \in [-\pi,\pi]$ with $\chi(\theta_1,\theta_2) > 1-\eta$. Let $T=\conv\{\e i, \e j, \e k\}$ with $\e i \in \L^1_\e$, $\e j \in \L^2_\e$, $\e k \in \L^3_\e$ and let $u(x) = \exp(\iota \theta(x))$ for $x \in \{\e i, \e j, \e k\}$ be such that $\theta(\e j) - \theta(\e i) \in [-\pi,\pi]$  and $\theta(\e k) - \theta(\e i) \in [-\pi,\pi]$. If $\chi(u,T) = \chi(\theta(\e j) - \theta(\e i), \theta(\e k) - \theta(\e i)) > 1 - \eta$, then $|\theta(\e j) - \theta(\e i) - \frac{2\pi}{3}| < \frac{\pi}{2}$ and $|\theta(\e k) - \theta(\e i) + \frac{2\pi}{3}| < \frac{\pi}{2}$. Let 
        \begin{equation*}
            \varphi(\e i) := \theta(\e i) \, , \quad \varphi(\e j) := \theta(\e j) - \frac{2\pi}{3} \, , \quad \varphi(\e k) := \theta(\e k) + \frac{2\pi}{3} \, ,
        \end{equation*}    
     so that $v(x) = \exp(\iota \varphi(x))$ for $x \in \{\e i, \e j, \e k\}$. In particular, $|\varphi(\e j)  - \varphi(\e i)| < \frac{\pi}{2}$, $|\varphi(\e k)  - \varphi(\e i)| < \frac{\pi}{2}$, and $|\varphi(\e k)  - \varphi(\e j)| < \pi$. The latter conditions imply that 
     \begin{equation*} 
     		\d^e\varphi(\e i, \e j) + \d^e\varphi(\e j, \e k) + \d^e\varphi(\e k, \e i)=0= \d^e\varphi(\e i,\e k)+\d^e\varphi(\e k,\e j)+\d^e\varphi(\e j,\e i)
     	\end{equation*}
      and thus $\mu_v(T) = 0$, independent of the ordering of the vertices $\e i,\e j,\e k$.    
\end{remark}

\begin{remark} \label{rmk:vorticity and chirality}
    Conversely to Remark~\ref{rmk:chirality and vorticity}, if the vorticity of the auxiliary spin field $v$ is 0 in a triangle of the lattice, there the chirality of $u$ cannot be close to $-1$. More precisely, there exists  $\eta' \in (0,1)$  such that $\chi(u,T) \geq -1+\eta'$ for every $u \in \SF_\e$ and $T \in \T_\e(\R^2)$ with $\mu_v(T) =0$. Indeed, as in Remark~\ref{rmk:chirality angles} there exists  $\eta' \in (0,1)$   such that $|\theta_1 + \frac{2\pi}{3}| < \frac{\pi}{6}$ and $|\theta_2 - \frac{2\pi}{3}| < \frac{\pi}{6}$ for every $\theta_1, \theta_2 \in [-\pi,\pi]$ with $\chi(\theta_1,\theta_2) < -1+\eta'$. Let $T=\conv\{\e i, \e j, \e k\}$ with $\e i \in \L^1_\e$, $\e j \in \L^2_\e$, $\e k \in \L^3_\e$, let $u(x) = \exp(\iota \theta(x))$ for $x \in \{\e i, \e j, \e k\}$.  If $\chi(u,T) < -1 +\eta'$,  then $|\theta(\e j) - \theta(\e i) + \frac{2\pi}{3}| < \frac{\pi}{6}$ and $|\theta(\e k) - \theta(\e i) - \frac{2\pi}{3}| < \frac{\pi}{6}$.  Let now $\varphi(\e i) := \theta(\e i)$, $\varphi(\e j) := \theta(\e j) - \frac{2\pi}{3}$, $\varphi(\e k) := \theta(\e k) + \frac{2\pi}{3}$, so that $v(x) = \exp(\iota \varphi(x))$ for $x \in \{\e i, \e j, \e k\}$. Then 
    \begin{equation*}
        |\varphi(\e j) - \varphi(\e i) + \tfrac{4\pi}{3}| < \tfrac{\pi}{6} \, , \quad |\varphi(\e i) - \varphi(\e k) + \tfrac{4\pi}{3}| <\tfrac{\pi}{6} \, , \quad |\varphi(\e k) - \varphi(\e j) - \tfrac{8\pi}{3}| < \tfrac{\pi}{3} \, ,
    \end{equation*}
    which yields $\d^e \varphi(\e i, \e j) = \varphi(\e j) - \varphi(\e i) + 2 \pi$,  $\d^e \varphi(\e j, \e k) = \varphi(\e k) - \varphi(\e j) - 2 \pi$, $\d^e \varphi(\e k, \e i) = \varphi(\e i) - \varphi(\e k) + 2 \pi$, whence $\mu_v(T) = 1$,  assuming that the counterclockwise order of the vertices is $(\e i,\e j,\e k)$. If instead the counterclockwise order is $(\e i,\e k,\e j)$, then the  antisymmetry condition $\d^e\varphi(x,x')=-\d^e\varphi(x',x)$ implies that $\mu_v(T)=-1$. 
\end{remark}

 Although a control $XY_\e \leq C E_\e$ does not hold true, we show that it is feasible in the regions where the spin field has no vorticity. 

\begin{lemma} \label{lemma:rough XY bound}
    There exists a constant $C>0$ such that the following holds true. Let $u \in \SF_\e$, let $v \in \SF_\e$ the auxiliary spin field defined according to~\eqref{def:from u to v}, and let $T \in \T_\e(\R^2)$. Then 
    \begin{equation} \label{eq:estimate with XY when zero vorticity}
        \mu_v(T)=0 \quad \implies \quad  XY_\e(v,T) \leq C E_\e(u,T) \, .
    \end{equation}
\end{lemma}
\begin{proof}
     Let us fix $\lambda = \frac{1}{2}\in (0,1)$ and let $\eta > 0$ be the corresponding number given by Lemma~\ref{lemma:bounds with XY}. If $\chi(u,T) > 1-\eta$, by Lemma~\ref{lemma:bounds with XY} we have that $\frac{1}{2}XY_\e(v,T) \leq E_\e(u,T)$. Therefore, we are left to prove the bound when $\chi(u,T) \leq 1-\eta$. Let now $\eta' > 0$ be as in Remark~\ref{rmk:vorticity and chirality}. From the condition $\mu_v(T) = 0$, we infer that $\chi(u,T) \geq -1+\eta'$.  Let $\eta'' := \min\{\eta,\eta'\}$. Since $ -1+\eta'' \leq \chi(u,T) \leq 1- \eta''$, by Remark~\ref{rmk:chiralityenergy} we deduce the existence of a constant $C_{\eta''} > 0$ such that $E_\e(u , T) \geq C_{\eta''} \e^2$. On the other hand $XY_\e(v,T) \leq 6 \e^2 \leq \frac{6}{C_{\eta''}} E_\e(u , T)$. This concludes the proof. 
\end{proof}

\begin{remark}
    The constant $C$ found in~\eqref{eq:estimate with XY when zero vorticity} is not optimal. 
\end{remark}

\section{\texorpdfstring{$\Gamma$}{}-limit  in the bulk scaling}

In this section we compute the $\Gamma$-limit of the AFXY-energy in the bulk scaling. We start with a lemma concerning the spin fields that cannot overcome the energetic barrier of the chirality transition (of order $\e$). 

\begin{lemma} \label{lemma:chi converges to 1}  Let $\Omega \subset \R^2$ be an open, bounded, and connected set and  let $u_\e \in \SF_\e$ be such that $E_\e(u_\e,\Omega) \leq C \delta_\e$ with $\delta_\e \ll \e$.  Then, up to a subsequence, either $\chi(u_\e) \to 1$ or $\chi(u_\e) \to -1$ in $L^1(\Omega)$. 
\end{lemma}
\begin{proof}
    Let 
\begin{equation*}
    \hat \chi_\e := \begin{cases}
        1 \, , & \text{if } \chi(u_\e)  > 0 \, , \\
        -1 \, , & \text{if } \chi(u_\e) \leq  0 \, .
    \end{cases}
\end{equation*}
For every $\Omega' \subset \subset \Omega$  connected set  such that  $\mathrm{dist}(\Omega',\partial \Omega) >\sqrt{3}\e$,  due to Lemma \ref{lemma:energy of two triangles} applied with $\eta=1$,  there exists $C>0$ such that 
\begin{equation*}
    \H^1(\de \{\hat \chi_\e = 1 \} \cap \Omega') \leq   \frac{C}{\e} E_\e(u_\e,\Omega) \leq C  \frac{\delta_\e}{\e} \, ,
\end{equation*}
for $\e$ small enough.  Since $\Omega'$ is connected, this implies that, up to a subsequence, either $\hat \chi_\e \to 1$ or $\hat \chi_\e \to -1$ strongly in $L^1(\Omega)$. Moreover, due to Remark \ref{rmk:chiralityenergy}, for every  $\eta \in (0,1)$  and every $\Omega' \subset \subset \Omega$ with $\mathrm{dist}(\Omega',\partial \Omega)>3\e$ we have that  
\begin{equation*}
    \lim_{\e \to 0} |\{ |\hat \chi_\e - \chi(u_\e) | > \eta \} \cap \Omega'| = 0 \, ,
\end{equation*}
which implies that either $\chi(u_\e) \to 1$ or $\chi(u_\e) \to -1$  in measure and therefore also  strongly in~$L^1(\Omega)$.
\end{proof}

 We compute the $\Gamma$-limit in the bulk scaling in the case where $\chi(u_\e) \sim 1$. An analogous statement holds true if $\chi(u_\e) \sim -1$ (in that case, the auxiliary variable has to be redefined accordingly). 
 
 We state the next theorem for $\Omega$ connected. In case $\Omega$ is not connected, the result holds true in every connected component of $\Omega$.

\begin{theorem} \label{thm:bulk}
    Assume that $\Omega$ is an open, bounded,  and connected  set with Lipschitz boundary.  The following results hold:
    \begin{enumerate}[label=\roman*)]
        \item {\em (Compactness)} Let $u_\e \in \SF_\e$ be such that $E_\e(u_\e,\Omega) \leq C \e^2$.  Then, up to a subsequence, either $\chi(u_\e) \to 1$ or $\chi(u_\e) \to -1$ in $L^1(\Omega)$. Assume that $\chi(u_\e) \to 1$, let $v_\e \in \SF_\e$ be the auxiliary spin field defined as in~\eqref{def:from u to v}, and let $\hat v_\e$ be its piecewise affine interpolation. Then there exists a subsequence (not relabeled) and~$v \in H^1(\Omega;\S^1)$ such that $\hat v_\e \to v$ strongly in $L^2(\Omega;\R^2)$ and~$\hat v_\e \wto v$ in $H^1_{\mathrm{loc}}(\Omega;\R^2)$.
        \item  {\em ($\liminf$ inequality)} Let $u_\e \in \SF_\e$ be such that $\chi(u_\e) \to 1$ in $L^1(\Omega)$, let $v_\e \in \SF_\e$ be the auxiliary spin field defined as in~\eqref{def:from u to v}, and let $\hat v_\e$ be its piecewise affine interpolation. Let~$v \in H^1(\Omega;\S^1)$ and assume that $\hat v_\e \to v$ strongly in $L^2(\Omega;\R^2)$. Then 
        \begin{equation*}
            \sqrt{3} \int_\Omega |\nabla v|^2 \, \d x \leq \liminf_{\e \to 0} \frac{1}{\e^2} E_\e(u_\e, \Omega) \, .
        \end{equation*}
        \item {\em ($\limsup$ inequality)} Let $v \in H^1(\Omega;\S^1)$. Then there exist $u_\e \in \SF_\e$ such that  $\chi(u_\e) \to 1$ in $L^1(\Omega)$ and $\hat v_\e \to v$ strongly in $L^2(\Omega;\R^2)$ and 
        \begin{equation*}
            \limsup_{\e \to 0} \frac{1}{\e^2} E_\e(u_\e, \Omega)  \leq \sqrt{3} \int_\Omega |\nabla v|^2 \, \d x  \, ,
        \end{equation*}
        where $v_\e \in \SF_\e$ is the auxiliary spin field defined as in~\eqref{def:from u to v} and  $\hat v_\e$ is its piecewise affine interpolation.
    \end{enumerate}
\end{theorem}

\begin{proof} 
Let us prove {\em i)}. The fact that either $\chi(u_\e) \to 1$ or $\chi(u_\e) \to -1$ in $L^1(\Omega)$ (up to a subsequence) follows from Lemma~\ref{lemma:chi converges to 1}. In the following, we assume that $\chi(u_\e) \to 1$.

Let us fix $\Omega' \subset \subset  \Omega'' \subset \subset \Omega$ be such that both $\Omega'$ and $\Omega''$ have Lipschitz boundary. For $\e$ small enough, $\sqrt{3} \e < \dist(\Omega'', \de \Omega)$ holds true. We fix $\lambda \in (0,1)$ and we consider the corresponding $\eta \in (0,1)$ provided by Lemma~\ref{lemma:bounds with XY}. By Lemma~\ref{lemma:counting} we get that there exists $C_\eta >0$ depending also on $\Omega''$ such that
\begin{equation*}
    \# \{ T \in  \T_\e(\Omega'' ) : \chi(u_\e, T) \leq 1 - \eta \} \leq  C_\eta\Big(\frac{1}{\e^2} E_\e(u_\e,\Omega)\Big)^2  \leq C_\eta.
\end{equation*}
Therefore, up to a subsequence, we can assume that $\#  \{ T \in  \T_\e(\Omega'' ) : \chi(u_\e, T) \leq 1 - \eta \} = \overline M$, the number $\overline M$ possibly depending on $\eta$ and $\Omega''$. This yields 
\begin{equation*}
    \sum_{\substack{T \in \T_\e(\Omega'')\\ \chi(u_\e,T) \leq 1 - \eta}} \frac{1}{\e^2} XY_\e(v_\e,T) \leq 6 \# \{T \in \T_\e(\Omega'') : \chi(u_\e,T) \leq 1 - \eta \}   \leq 6 \overline M \, .
\end{equation*}
We apply Lemma~\ref{lemma:bounds with XY} to obtain that
\begin{equation*}
    \begin{split}
        C &\geq \frac{1}{\e^2}E_\e(u_\e,\Omega) \geq \sum_{\substack{T \in \T_\e(\Omega'')\\ \chi(u_\e,T) > 1 - \eta}} \frac{1}{\e^2}E_\e(u_\e,T)  \geq (1-\lambda)\sum_{\substack{T \in \T_\e(\Omega'')\\ \chi(u_\e,T) > 1 - \eta}} \frac{1}{\e^2} XY_\e(v_\e,T) \\
        & \geq (1-\lambda) \frac{1}{\e^2}XY_\e(v_\e, \Omega'') -  (1-\lambda)  6 \overline M \, .
    \end{split}
\end{equation*}
In conclusion, applying Remark~\ref{rmk:XY is integral},
\begin{equation*}
    C \geq \frac{1}{\e^2} XY_\e(v_\e, \Omega'')  \geq \sqrt{3} \int_{\Omega'} |\nabla \hat v_\e|^2 \, \d x \, .
 \end{equation*}  
From this we deduce that, up to a subsequence, $\hat v_\e \wto v$ in $H^1(\Omega';\R^2)$ and $\hat v_\e \to v $ a.e.\ in $\Omega'$, with $v \in H^1(\Omega';\R^2)$. To prove that $|v| = 1$ we apply Lemma~\ref{lemma:potential part} to infer that 
\begin{equation*}
    \int_{\Omega'} (1 - |\hat v_\e|^2)^2 \d x \leq C XY_\e(v_\e, \Omega) \leq C \e^2 \to 0 \, 
\end{equation*}
to obtain, due to Fatou's Lemma,  that $|v| = 1$ a.e. in $\Omega^{\prime\prime}$. 

By a diagonal argument we find a $v \in H^1(\Omega;\S^1)$ and we extract a subsequence such that for every $\Omega' \subset \subset \Omega$ we have $\hat v_\e \wto v$ in $H^1(\Omega';\R^2)$. Finally, $\hat v_\e \to v$ in $L^2(\Omega;\R^2)$ by the Dominated Convergence Theorem. 

Let us prove {\em ii)}. We let $\Omega' \subset \subset  \Omega'' \subset \subset \Omega$ and  $\lambda \in (0,1)$,  $\eta \in (0,1)$ as in the proof of {\em i)}. We prove in the same way that $\#  \{ T \in  \T_\e(\Omega'' ) : \chi(u_\e, T) \leq 1 - \eta \} = \overline M$ with $\overline M$ possibly depending on~$\eta$ and~$\Omega''$.  Let $b_\e^1,\dots, b_\e^{\overline M}$ be the barycenters of the triangles in $\{ T \in  \T_\e(\Omega'' ) : \chi(u_\e, T) \leq 1 - \eta \}$. There exist $b^1, \dots, b^{M} \in  \Omega$ with $M \leq \overline M$ such that, up to a subsequence, each of the points~$b_\e^1,\dots, b_\e^{\overline M}$ converges to one of the points in $\{b^1, \dots, b^{M}\}$. Let us fix $\rho > 0$.
 For $\e$ small enough,  every triangle $T \in \T_\e\big( \Omega'' \sm \bigcup_{h=1}^{M} B_{\rho}(b^h) \big)$ satisfies $\chi(u_\e,T) > 1 - \eta$. In particular, Lemma~\ref{lemma:bounds with XY} and Remark~\ref{rmk:XY is integral} yield
\begin{equation*}
    \begin{split}
         \liminf_{\e \to 0} \frac{1}{\e^2}E_\e(u_\e, \Omega) & \geq \liminf_{\e \to 0} \frac{1}{\e^2} E_\e\Big(u_\e, \Omega'' \sm \bigcup_{h=1}^{M} B_{\rho}(b^h) \Big) \\
         & \geq  (1-\lambda) \liminf_{\e \to 0}  \frac{1}{\e^2} XY_\e\Big(v_\e, \Omega'' \sm \bigcup_{h=1}^{M} B_{\rho}(b^h) \Big) \\
        & \geq  (1-\lambda)  \liminf_{\e \to 0} \sqrt{3} \int_{\Omega' \sm \bigcup_{h=1}^{M} B_{2\rho}(b^h)} |\nabla \hat v_\e|^2 \, \d x  \\
        & \geq  (1-\lambda)  \sqrt{3} \int_{\Omega' \sm \bigcup_{h=1}^{M} B_{2\rho}(b^h)} |\nabla v|^2 \, \d x \, , 
    \end{split}
\end{equation*}
where we used that $\hat v_\e  \wto v$  in $H^1(\Omega';\R^2)$. The claim is proven by letting, in the order, $\rho \to 0$,  $\lambda \to 0$ , and $\Omega' \nearrow \Omega$.  

Let us prove {\em iii)}. Let $v \in H^1(\Omega;\S^1)$. Thanks to the regularity of the boundary we find  $\tilde \Omega \supset \supset \Omega$ open, bounded set with Lipschitz boundary and we extend $v$ to a map in $H^1(\tilde \Omega;\S^1)$, which we still denote, with a slight abuse of notation, by~$v$. This can be achieved via a reflection argument in an open neighborhood of $\de \Omega$. More details can be found, \eg in~\cite[Step 2 in proof of Proposition~4.3]{CicOrlRuf-I}. For the moment, let us assume that $v \in C^\infty(\tilde \Omega;\S^1) \cap H^1(\tilde \Omega;\S^1)$. Later we will prove the result for a generic $v \in H^1(\tilde \Omega;\S^1)$ with a regularization argument.   We define the discrete spin field $u_\e \in \SF_\e$ as follows:
    \begin{gather*}
        v_\e(\e i) := v(\e i) \, , \quad v_\e(\e j) := v(\e j) \, , \quad v_\e(\e k) := v(\e k) \, , \\
        u_\e(\e i) := v_\e(\e i) \, , \quad u_\e(\e j) := R[\tfrac{2\pi}{3}](v_\e(\e j)) \, , \quad u_\e(\e k) := R[-\tfrac{2\pi}{3}](v_\e(\e k)) \, ,
    \end{gather*}
    for $\e i \in \L^1_\e \cap \tilde \Omega$, $\e j \in \L^2_\e \cap \tilde \Omega$, $\e k \in \L^3_\e \cap \tilde \Omega$. For points of $\L_\e$ outside $\tilde \Omega$, we define $u_\e$ arbitrarily.
    
    Let~$\hat v_\e$ be the affine interpolation of $v_\e$ and let us prove that $\hat v_\e \to v$ in $L^2(\Omega;\R^2)$ and
    \begin{equation} \label{eq:limsup}
        \limsup_{\e \to 0} \frac{1}{\e^2} E_\e(u_\e, \Omega) \leq \sqrt{3} \int_{\Omega} |\nabla v|^2 \, \d x \, .
    \end{equation}
    
    Let $\Omega \subset \subset U \subset \subset \tilde \Omega$ and let $T \in \T_\e(U)$. Then, for $x \in T$ we have that $|\hat v_\e(x) - v(x)| \leq 3 \|\nabla v\|_{L^\infty(U)} \e$. This yields $\|\hat v_\e - v\|_{L^2(\Omega)} \to 0$. 
    
   Let now $\alpha \in \{1,2,3\}$ and let $\e i, \e j$ be two vertices of $T \in \T_\e(\Omega)$ (not necessarily $\e i \in \L^1_\e$ and $\e j \in \L^2_\e$) with $j - i = \tb_\alpha$. By a Taylor expansion there exists $\xi$ belonging to the segment $[\e i, \e j]$ such that 
    \begin{equation*}
        \begin{split}
           | \nabla \hat v_\e(x) \tb_\alpha - \nabla v(x)\tb_\alpha | & = \Big| \frac{v(\e j) - v(\e i)}{\e} - \nabla v(x)\tb_\alpha \Big| \\
            & = \Big| \nabla v(\e i)  \tb_\alpha + \frac{1}{2} \D^2 v(\xi)(\e j - \e i) \cdot ( j - i) - \nabla v(x)\tb_\alpha \Big| \\
            & \leq \Big| \nabla v(\e i)  \tb_\alpha  - \nabla v(x)\tb_\alpha \Big| + \frac{1}{2} \| \D^2 v \|_{L^\infty( U)} \e \\
            & \leq C \| \D^2 v \|_{L^\infty(U)} \e 
        \end{split}
    \end{equation*}
    for every $x \in T$, which yields $\|\nabla \hat v_\e - \nabla v\|_{L^2(\Omega)} \to 0$. Let us fix  $\lambda \in (0, 1)$ and let  $\eta \in (0,1)$  be as in Lemma~\ref{lemma:bounds with XY}. Let $T = \conv\{\e i, \e j , \e k\} \in \T_\e(\Omega)$ with $\e i \in \L^1_\e$, $\e j \in \L^2_\e$, $\e k \in \L^3_\e$. For $\e$ small enough we have that 
    \begin{align*}
        \d_{\S^1}(v_\e(\e i), v_\e(\e j)) & \leq \frac{\pi}{2} |v_\e(\e i) - v_\e(\e j)| \leq \frac{\pi}{2} \|\nabla v\|_{L^\infty(\Omega)} \e  <  \min\big\{ \eta , \tfrac{\pi}{3} \big\}  \, , \\
        \d_{\S^1}(v_\e(\e i), v_\e(\e k)) & \leq \frac{\pi}{2} |v_\e(\e i) - v_\e(\e k)| \leq \frac{\pi}{2} \|\nabla v\|_{L^\infty(\Omega)} \e  <  \min\big\{ \eta, \tfrac{\pi}{3} \big\}  \, .
    \end{align*}
    In particular, this implies that $\chi(u_\e) > 0$,  as $(u_\e(\e i), u_\e(\e j), u_\e(\e k))$ are in a counterclockwise order (see \cite[Remark 2.3]{BacCicKreOrl}). By Lemma~\ref{lemma:bounds with XY}, by Remark~\ref{rmk:XY is integral}, and since $\|\nabla \hat v_\e - \nabla v\|_{L^2(\Omega)} \to 0$,
    \begin{equation*}
        \begin{split}
            \limsup_{\e \to 0}\frac{1}{\e^2}E_\e(u_\e,\Omega) &  \leq  (1+\lambda)\limsup_{\e \to 0} \frac{1}{\e^2}XY_\e(v_\e,\Omega) \leq  (1+\lambda) \limsup_{\e \to 0} \sqrt{3} \int_{\Omega} |\nabla \hat v_\e|^2 \, \d x  \\
            & \leq  (1+\lambda) \sqrt{3} \int_{\Omega} |\nabla   v|^2 \, \d x  \, .
        \end{split}
    \end{equation*}
    Letting $\lambda \to 0$ we conclude the proof of~\eqref{eq:limsup}. Let us prove that $\chi(u_\e) \to 1$ in $L^1(\Omega)$. From~\eqref{eq:limsup} we get that $\frac{1}{\e^2} E_\e(u_\e,\Omega) \leq C$. Using Lemma \ref{lemma:chi converges to 1} and using the fact that $\chi(u_\e) > 0$ (independent of the subsequence), we conclude that  $\chi(u_\e) \to 1$ in $L^1(\Omega)$. 

    We assume now that $v \in H^1(\tilde \Omega;\S^1)$ and we regularize it. By Schoen-Uhlenbeck's approximation theorem for Sobolev maps between manifolds~\cite[Section 4]{SchUhl}, there exists a sequence $v_n \in C^\infty(\tilde \Omega;\S^1) \cap H^1(\tilde \Omega;\S^1)$ such that $\| v_n - v \|_{H^1(\tilde \Omega;\R^2)} \leq \frac{1}{n}$ (see also~\cite[5.1, Theorem~3 and Remark~1]{GiaModSou-I}). Then we conclude the proof of the limsup inequality by a standard diagonal argument. 
\end{proof}

 A consequence of the $\Gamma$-limit result in the bulk scaling is the following lower bound for the energy under a degree constraint on the spin field. To properly set the constraint, we define the set of admissible spin fields with degree $d$ in an annulus $\A{r}{R}$ by
\begin{equation*}
    \Adm{r}{R}^\e(d) := \Big\{ u \in \SF_\e : \mu_v(T) = 0 \text{ for every } T \in \T_\e(\mathbb{R}^2) \text{  with  } T\cap \A{r}{R} \neq \emptyset \,, \  \mu_v(B_r) = d  \Big\} \, ,
\end{equation*} 
where $v \in \SF_\e$ is the auxiliary spin field associated to $u$ defined as in~\eqref{def:from u to v}.

\begin{proposition} \label{prop:convergence of minima}
    For every $0 < r < R$, we have 
    \begin{equation*}
        \liminf_{\e \to 0}   \ \inf \Big\{ \frac{1}{\e^2} E_{\e}(u,\A{r}{R}) : u \in \Adm{r}{R}^\e(d) \Big\} \geq 2\sqrt{3} \pi |d|^2 \log \frac{R}{r} \, .
    \end{equation*}
\end{proposition}
\begin{proof}
    For every $\e > 0$ let $u_\e \in \Adm{r}{R}^\e(d)$ be such that 
    \begin{equation}\label{ineq:almostmin}
        \frac{1}{\e^2} E_{\e}(u_\e,\A{r}{R}) \leq \inf \Big\{ \frac{1}{\e^2} E_{\e}(u,\A{r}{R}) : u \in \Adm{r}{R}^\e(d) \Big\} + \e \, .
    \end{equation}
    Without loss of generality, we assume that $\frac{1}{\e^2} E_{\e}(u_\e,\A{r}{R})$ is equibounded. By Remark~\ref{rmk:vorticity and chirality}, there exists $\eta'\in (0,1)$ such that $\chi(u_\e) > -1 + \eta'$ in $\A{r}{R}$. Let $r < r' < R' < R$. By Theorem~\ref{thm:bulk}-{\em i)}, up to a subsequence, either $\chi(u_\e) \to 1$ or $\chi(u_\e) \to -1$ in $L^1(\A{r'}{R'})$. Since $\chi(u_\e) > -1 + \eta'$, the latter possibility is ruled out. Via a diagonal argument, we obtain that $\chi(u_\e) \to 1$ in $L^1(\A{r'}{R'})$ for every $r < r' < R' < R$. By Theorem~\ref{thm:bulk}-{\em i)} and via a diagonal argument we find $v \in H^1(\A{r}{R};\S^1)$ such that $\hat v_\e \wto v$ in $H^1_{\mathrm{loc}}(\A{r}{R};\R^2)$, up to a subsequence that we do not relabel. 
    
    Let us prove that $\deg(v,\de B_\rho) = d$ for every $\rho \in [r,R]$. By~\eqref{ineq:almostmin}, by Theorem~\ref{thm:bulk}--{\em ii)},  since $2 |j(v)| = |\nabla v|$, and~\eqref{eq:degree for H1},  this yields that 
    \begin{equation*}
         \begin{split}
           & \liminf_{\e \to 0}   \ \inf \Big\{ \frac{1}{\e^2} E_{\e}(u,\A{r}{R}) : u \in \Adm{r}{R}^\e(d) \Big\} \geq \liminf_{\e \to 0} \frac{1}{\e^2}E_\e(u_\e,\A{r}{R}) \geq \sqrt{3}\int_{\A{r}{R}}|\nabla v|^2\,\mathrm{d}x \\
           &  \geq 4 \sqrt{3} \int_r^R \int_{\partial B_\rho} \big|j(v)|_{\partial B_{\rho}} \cdot \tau_{\partial B_{\rho}} \big|^2\,\mathrm{d}\mathcal{H}^1 \, \d \rho \geq  \sqrt{3} \int_r^R \frac{2}{\pi\rho} \Big|\int_{\partial B_\rho} j(v)|_{\partial B_\rho}\cdot \tau_{\partial B_{\rho}}\,\mathrm{d}\mathcal{H}^1\Big|^2 \, \d \rho \\
           & = 2 \sqrt{3} \pi \int_{r}^R \frac{1}{\rho}|d|^2\,\mathrm{d}\rho = 2 \sqrt{3}  \pi |d|^2\log \frac{R}{r}\,.
        \end{split}  
    \end{equation*} 
    To prove that $\deg(v,\de B_\rho) = d$ for every $\rho \in [r,R]$, by~\eqref{eq:degree independent} it is enough to show that $\deg(v,\de B_r) = d$. Let $\overarc v_\e \colon \R^2 \to \S^1$ be the map associated to $v_\e$ given by Remark~\ref{rmk:arcwise interpolation}. Let $r < r' < R' < R$. We claim that $j(\overarc v_\e) - j(\hat v_\e) \to 0$ in $L^1(\A{r'}{R'};\R^2)$. Indeed, given $T =\conv\{\e i, \e j, \e k\} \in \T_\e(\A{r}{R})$, since $\mu_{v_\e}(T) = 0$, we get $\overarc v_\e = \exp(\iota \hat \phi_\e)$, where $\hat \phi_\e$ is the affine function in $T$ with $v_\e(x) = \exp(\iota \hat \phi_\e(x))$ for $x \in \{\e i, \e j, \e k\}$, see~Remark~\ref{rmk:arcwise interpolation}.  Let $\alpha \in \{1,2,3\}$ and let $j - i = \tb_\alpha$. For every $x \in T$ we have that
     \begin{equation*}
        \begin{split}
            | \nabla \overarc v_\e(x)  \tb_\alpha - \nabla \hat v_\e(x)   \tb_\alpha  | & = | \exp(\iota \hat \phi_\e(x))^\perp  \nabla \hat \phi_\e(x) \cdot \tb_\alpha - \nabla \hat v_\e(x)   \tb_\alpha | \\
            & = \Big| \exp(\iota \hat \phi_\e(x))^\perp  \frac{\phi_\e(\e j) - \phi_\e(\e i)}{\e} -  \frac{\exp(\iota \phi_\e(\e j)) - \exp(\iota \phi_\e(\e i))}{\e}\Big| \\
            & =  \Big| \exp(\iota \hat \phi_\e(x))^\perp  \frac{\phi_\e(\e j) - \phi_\e(\e i)}{\e} -  \exp(\iota \xi_{i,j})^\perp  \frac{\phi_\e(\e j) - \phi_\e(\e i)}{\e} \Big| \\
            & = |\exp(\iota \hat \phi_\e(x))^\perp  - \exp(\iota \xi_{i,j})^\perp | | \nabla \hat \phi_\e(x) \cdot \tb_\alpha|  \\
            & \leq | \hat \phi_\e(x)  -   \xi_{i,j}  | |  \nabla \overarc v_\e(x)  \tb_\alpha|  \leq \e |\nabla  \hat \phi_\e(x)| |  \nabla \overarc v_\e(x)  \tb_\alpha| \leq \e  |  \nabla \overarc v_\e(x) |^2,
        \end{split}
     \end{equation*}
     where $\xi_{i,j}$ belongs to the segment $[\phi_\e(\e i), \phi_\e(\e j)]$. Moreover, by a straightforward computation\footnote{\eg in the case $j - i=\tb_1$ and $k-i = \tb_2$, one writes $x = \e i + s \e \tb_1 + t \e \tb_2$ with $s,t \in [0,1]$, $\overarc v_\e(x) = \exp(\iota \hat \phi_\e(x)) = \exp(\iota \phi_\e(\e i))   \exp\big(\iota s (\phi_\e(\e j) - \phi_\e(\e i))\big)  \exp\big(\iota t (\phi_\e(\e k) - \phi_\e(\e i))\big)$,2 and $\hat v_\e(x) = \exp(\iota \phi_\e(\e i)) +  s \big( \exp(\iota \phi_\e(\e j)) - \exp(\iota \phi_\e(\e i)) \big) + t \big(\exp(\iota \phi_\e(\e k)) - \exp(\iota \phi_\e(\e i)) \big)$.} one shows that for every $x \in T$ 
     \begin{equation*}
         |\overarc v_\e(x) - \hat v_\e(x)| \leq C \e |\nabla \hat \phi_\e(x)| = C \e |\nabla \overarc v_\e(x)| \, .
     \end{equation*}
    The previous inequalities and \eqref{eq:arcwise controlled with XY}  yield
     \begin{equation*}
        \begin{split}
            \int_{T} |j(\overarc v_\e) - j(\hat v_\e)| \, \d x  & \leq C \int_{T} | \overarc v_\e - \hat v_\e | |\nabla \overarc v_\e| + |\hat v_\e|| \nabla \overarc v_\e - \nabla \hat v_\e | \, \d x \\
            &  \leq C \int_{T} | \overarc v_\e - \hat v_\e | |\nabla \overarc v_\e| + | \nabla \overarc v_\e - \nabla \hat v_\e | \, \d x  \leq C \e \int_{T} |\nabla \overarc v_\e|^2 \, \d x  \\
            & \leq  \frac{C}{\e} XY_\e(v_\e, T) \leq \frac{C}{\e} E_\e(u_\e, T) \, ,
        \end{split}
    \end{equation*}
    where in the last inequality we used the fact that $\mu_{v_\e}(T) = 0$ and we applied Lemma~\ref{lemma:rough XY bound}. Summing over all triangles $T \in \T_\e(\A{r}{R})$ that intersect $\A{r'}{R'}$ we conclude that 
    \begin{equation*}
        \int_{\A{r'}{R'}} |j(\overarc v_\e) - j(\hat v_\e)| \, \d x  \leq \frac{C}{\e} E_\e(u_\e, \A{r}{R}) \leq C \e \to 0 \, ,
    \end{equation*}
     which in turn implies that $j(\overarc{v}_\e)\to j(v)$ in $L^1(\A{r'}{R'};\R^2)$. 
    We are now in a position to prove that $\deg(v,\de B_r) = d$. Let $\psi(x) := 1- \min\{\frac{1}{R'-r'}\dist(x, B_{r'}),1\}$. Using the convergence of $j(\overarc{v}_\e)$ together with the fact that $v\in H^1(\A{r}{R};\S^1)$, by~\eqref{eq:degree for H1} and~\eqref{eq:degree independent} we have that 
    \begin{equation*}
        \begin{split}
            & \pi d = \pi \mu_{v_\e}(B_r) = \int_{B_R} J (\overarc v_\e) \psi \, \d x  = -\int_{\A{r'}{R'}} j(\overarc v_\e) \cdot  \nabla^\perp \psi \, \d x \to  \\
            & \quad \to  -\int_{\A{r'}{R'}} j( v) \cdot  \nabla^\perp \psi \, \d x  = \frac{1}{R'-r'}\int_{r'}^{R'} \int_{\de B_\rho}j( v)|_{\de B_\rho} \cdot \tau_{\de B_\rho} \, \d \H^1  \, \d \rho  =  \pi\deg(v,\de B_r)   \, ,
        \end{split}
    \end{equation*}
    which concludes the proof.
\end{proof}

\section{Ball Construction}

 In this section we prove a variant of the well-known ball construction~\cite{San,Jer} suited for our arguments. 

Let $\mathcal{B}= \{B_{r_i}(x_i)\}_{i=1}^N$ be a finite family of open balls such that $B_{r_i}(x_i)\cap B_{r_j}(x_j) = \emptyset$ for every $i,j \in \{1,\ldots,N\}$, $i\neq j$. Let $\mu=\sum_{i=1}^N d_i \delta_{x_i}$, $d_i \in \mathbb{Z} \setminus \{0\}$, $x_i \in \mathbb{R}^2$, and let $\mathcal{E}(\mathcal{B},\mu,\cdot) \colon \mathcal{A}(\mathbb{R}^2) \to [0,+\infty]$ be an increasing set function satisfying the following properties:
\begin{itemize}
\item[(B1)] $\mathcal{E}(\mathcal{B},\mu, U\cup V) \geq \mathcal{E}(\mathcal{B},\mu,U) + \mathcal{E}(\mathcal{B},\mu,V)$ for every $U,V \in \mathcal{A}(\mathbb{R}^2)$ such that $U\cap V = \emptyset$.
\item[(B2)] for every annulus $\A{r}{R}(x) = B_R(x)\setminus \overline{B}_r(x)$, $0<r<R$ with $\A{r}{R}(x) \cap \bigcup_{i=1}^N B_{r_i}(x_i) = \emptyset$, it holds
\begin{align}\label{ineq:lowerboundannulus}
\mathcal{E}(\mathcal{B},\mu,\A{r}{R}(x)) \geq c_0 |\mu(B_r(x))| \log \frac{R}{r}\,,
\end{align}
for some constant $c_0 >0$.
\end{itemize}
 Given a ball $B$, we let $r(B)$ denote its radius. For a family of balls $\mathcal{B}$, we let $\mathcal{R}(\mathcal{B}) := \sum_{B \in \mathcal{B}} r(B)$. 

\begin{lemma}[Ball construction] \label{lemma:ball construction} Let $\mathcal{B}$, $\mu$, and $\mathcal{E}$ be as above. Let $\sigma >0$. Then there exists a one-parameter family $\{ \mathcal{B}(t) \}_{t \geq 0}$ of balls such that
\begin{itemize}
\item[{\rm (1)}] the following inclusions hold:
\begin{align*}
\bigcup_{B \in \mathcal{B}} B \subset \bigcup_{B \in \mathcal{B}(t_1)} \!  B \subset \bigcup_{B \in \mathcal{B}(t_2)}\! B \, , \quad  \text{for every } 0 \leq t_1 \leq t_2\, ;
\end{align*}
\item[{\rm (2)}]  $\overline B\cap \overline B^\prime = \emptyset$ for every $B,B^\prime \in \mathcal{B}(t)$, $B\neq B^\prime$,  and $t \geq 0$;  
\item[{\rm (3)}] for every $0\leq t_1 \leq t_2$ and every $U \in \mathcal{A}(\mathbb{R}^2)$ we have that
\begin{align*}
\mathcal{E}\Big(\mathcal{B},\mu, U \cap \Big( \bigcup_{B \in \mathcal{B}(t_2)}B \setminus \bigcup_{B \in \mathcal{B}(t_1)}\overline{B}\Big)\Big) \geq c_0 \underset{B \subset U}{\sum_{B \in\mathcal{B}(t_2)}} |\mu(B)| \log\frac{1+t_2}{1+t_1}\, ;
\end{align*}
\item[{\rm (4)}] $|\mu|(B_{r+\sigma}(x)\setminus B_{r-\sigma}(x))=0$ for every $B=B_r(x) \in \mathcal{B}(t)$ and for every $t\geq 0$;
\item[{\rm (5)}] for every $t\geq 0$ we have that $\mathcal{R}(\mathcal{B}(t)) \leq (1+t) \left(\mathcal{R}(\mathcal{B})  +N\sigma \right)$;
\item[{\rm (6)}]  for every $t \geq 0$, $B \in \mathcal{B}$, and $B' \in \mathcal{B}(t)$ with $B \subset B'$ we have that $r(B') \geq (1+t) r(B)$.
\end{itemize}
\end{lemma}
\begin{proof}  In order to construct the family $\mathcal{B}(t)$, we closely follow the strategy of the ball construction due to Sandier~\cite{San} and Jerrard~\cite{Jer}. We adapt the argument in order to be sure that condition~({\rm 4}) holds true, \ie that the measure $\mu$ is supported far from the boundaries of the balls of the constructed family. 
  
  The ball construction consists in letting the balls alternatively expand and merge into each other.  We let $T_0 := 0$ and we define the family $\mathcal{B}(T_0)$ by distinguishing the following two cases: If $\overline B_{r_i+\sigma}(x_i)\cap \overline B_{r_j+\sigma}(x_j) \neq \emptyset$ for some of the starting balls with $i,j \in \{1,\ldots,N\}$, $i\neq j$, then the construction starts with a merging phase and $T_0 = 0 $ is the first merging time. This phase consists in identifying a suitable partition $\{S_j^0\}_{j=1,\ldots,N_0}$  of the family $\{B_{r_i+\sigma}(x_i)\}_{i=1}^N$ which satisfies the following: for each  $j \in \{1,\dots,N_0\}$ there exists a ball $B_{r_j^0}(x_j^0)$ which contains all the balls in $S_j^0$ and such that
  \begin{itemize}
  \item[{\em i}\,)]  $\overline B_{r_j^0}(x_j^0) \cap \overline B_{r_\ell^0}(x_\ell^0) = \emptyset$   for every $j,\ell \in \{1,\dots, N_0\}$, $j \neq \ell$,
  \item[{\em ii}\,)] $r_j^0 \leq \sum_{B \in S_j^0} r(B)\,.$
  \end{itemize}
   We then define
  \begin{equation}\label{eq:ball-construction-Tzero}
  \mathcal{B}(T_0)\defas\{B_{r_j^0}(x_j^0)\colon j=1\,,\ldots\,, N_0\}\,.
  \end{equation}
  If, instead, $\overline B_{r_i+\sigma}(x_i)\cap \overline B_{r_j+\sigma}(x_j) = \emptyset$ for every $i,j \in \{1,\ldots,N\}$, $i\neq j$, then we let $N_0 := N$, $B_{r_j^0}(x_j^0) := B_{r_j+\sigma}(x_j)$ for $j = 1, \dots, N$ in ~\eqref{eq:ball-construction-Tzero}, and we start with an expansion phase. During this first expansion phase, we let the balls expand without changing their centres, in such a way that the  new radius $r_j^0(t)$ of the ball centred in  $x_j^0$ satisfies
\begin{align*}
 \frac{r_j^0(t)}{r_j^0} = \frac{1+t}{1+T_0} = 1+t\,,
\end{align*}
 for every $t\geq T_0 = 0$ and every $j \in \{1,\dots,N_0\}$. 
We continue the first expansion phase as a long as
\begin{align}\label{ineq:distancealpha}
 \overline B_{r_j^0(t)}(x_j) \cap \overline B_{r_\ell^0(t)}(x_\ell) =\emptyset\, \text{ for every } j,\ell \in \{1,\dots,N_0\}\, , \ j \neq \ell\,, 
\end{align}
and we let $T_1$ denote the smallest  $t\geq T_0 = 0$ such that \eqref{ineq:distancealpha} is violated. (Note that $T_1 > 0$.)  At time $T_1$, following the same procedure described above, a merging phase starting from the balls $\{B_{r_j^0(T_1)}(x_j^0)\}_{j=1}^{N_0}$ begins, that defines a new family of balls $\{B_{r_j^1}(x_j^1)\}_{j=1}^{N_1}$. 

We iterate this procedure by alternating merging and expansion phases to obtain the following: a discrete set of times $\{T_0,\ldots,T_K\}$, $K \leq N$; for each $k \in \{1,\ldots,K\}$, a partition $\{S_j^k\}_{j=1}^{N_k}$ of~$\{B_{r_j^{k-1}(T_k)}(x_j^{k-1})\}_{j=1}^{N_{k-1}}$; for each subclass $S_j^k$, a ball $B_{r_j^k}(x_j^k)$, which contains the balls in $S_j^k$ and such that the following properties are satisfied:
\begin{itemize}
\item[{\em i}\,)]  $\overline B_{r_j^k}(x_j^k) \cap \overline B_{r_\ell^k}(x_\ell^k) = \emptyset$ for every $j,\ell \in \{1,\dots, N_k\}$, $j \neq \ell$,
\item[{\em ii}\,)] $r_j^k \leq \sum_{B \in S_j^k} r(B) $.
\end{itemize}
For $t\geq 0$, the family $\mathcal{B}(t)$ is given by $\{B_{r_j^k(t)}(x_j^k)\}_{j=1}^{N_k}$ for $t \in [T_k,T_{k+1})$ and $k =0,\ldots,K$, where we set $T_{K+1}:=+\infty$ (in other words, it consists of a single expanding ball for $t \geq T_K$). For every $t \in [T_k,T_{k+1})$ and for $j = 1, \dots, N_k$, the radii satisfy  
\begin{equation}\label{eq:rj1t}
  \frac{r_j^k(t)}{r_j^k} = \frac{1+t}{1+T_k} \, .
\end{equation}
Note that 
\begin{equation} \label{eq:initial radii}
  \mathcal{R}(\mathcal{B}(T_0)) = \sum_{j=1}^{N_0} r_j^0 \leq \mathcal{R}(\mathcal{B}) + N \sigma \, .
\end{equation} 

It remains to check that conditions {\rm (1)}--{\rm (5)} hold true.  By construction, it is clear that {\rm (1)} and~{\rm (2)} are satisfied.

Let us prove {\rm (3)}. We note that,  by {\rm (1)},  
\begin{align}\label{ineq:tau1tau2}
\underset{B \subset U}{\sum_{B \in \mathcal{B}(\tau_1)}} |\mu(B)| \geq \underset{B \subset U}{\sum_{B \in \mathcal{B}(\tau_2)}} |\mu(B)| \quad \text{ for every } 0 < \tau_1 < \tau_2\,.
\end{align}
Let $t_1 < \overline{t} < t_2$. In view of \eqref{ineq:tau1tau2}, since $\mathcal{E}$ is an increasing sub-additive set-function, if we show that {\rm (3)} holds true for the pairs $(t_1,\overline{t})$ and $(\overline{t},t_2)$, then {\rm (3)} also follows for $t_1$ and $t_2$. Therefore we can assume, without loss of generality, that $T_k \notin (t_1,t_2)$ for every $k=1,\ldots,K$. Let $t_1 < \tau < t_2$ and let $B \in \mathcal{B}(\tau)$. Then, there exists a unique ball $B^\prime \in \mathcal{B}(t_1)$ such that $B^\prime \subset B$. By construction $\mu(B) = \mu(B^\prime)$ and, by \eqref{ineq:lowerboundannulus}, we have that
\begin{align*}
\mathcal{E}(\mathcal{B},\mu, B\setminus B^\prime) \geq c_0 |\mu(B^\prime)| \log \frac{1+\tau}{1+t_1}=  c_0 |\mu(B)| \log \frac{1+\tau}{1+t_1}\,.
\end{align*}
Summing up over all $B \in \mathcal{B}(\tau)$ with $B \subset U$ and using \eqref{ineq:tau1tau2} yields
\begin{equation*}
\mathcal{E}\Big(\mathcal{B},\mu, U \cap \Big(\bigcup_{B \in \mathcal{B}(t_2)} B \setminus \bigcup_{B \in \mathcal{B}(t_1)} B \Big)\Big) \geq c_0\underset{B \subset U}{\sum_{B \in \mathcal{B}(\tau)}} |\mu(B)| \log \frac{1+\tau}{1+t_1} \geq c_0\underset{B \subset U}{\sum_{B \in \mathcal{B}(t_2)}} |\mu(B)| \log \frac{1+\tau}{1+t_1}\,.
\end{equation*}
Property {\rm (3)} follows by letting $\tau \to t_2$.

Let us prove {\rm (4)}.  Let $t \geq 0$ and let $B = B_{r}(x) \in \mathcal{B}(t)$. Let us fix an initial ball $B_{r_i}(x_i)$. By construction, $B_{r_i+\sigma}(x_i)$ is contained in some ball $B_{r'}(y) \in \mathcal{B}(t)$, \ie $B_{r_i}(x_i) \subset B_{r'-\sigma}(y)$. Then $B_{r_i}(x_i) \cap B_{r+\sigma}(x) \subset B_{r-\sigma}(x)$, since condition~{\rm (2)} implies that $\overline B_{r'-\sigma}(y) \cap \overline B_{r+\sigma}(x) = \emptyset$ whenever~$y \neq x$. This yields
\begin{equation*}
  B_{r+\sigma}(x) \cap \bigcup_{i=1}^N B_{r_i}(x_i) \subset B_{r-\sigma}(x) \quad \implies \quad B_{r+\sigma}(x) \sm B_{r-\sigma}(x) \subset  B_{r+\sigma}(x) \sm   \bigcup_{i=1}^N B_{r_i}(x_i)   \, .
\end{equation*}
Therefore
\begin{align*}
|\mu|(B_{r+\sigma}(x)\setminus B_{r-\sigma}(x)) \leq |\mu| \Big(\mathbb{R}^2 \setminus \bigcup_{i=1}^N B_{r_i}(x_i)\Big) = 0\,,
\end{align*}
where we used the fact that $\mu$ is supported on $\{x_1,\dots,x_N\}$. This proves~{\rm (4)}.

To prove {\rm (5)}, we start by observing that, by~\eqref{eq:rj1t},
\begin{equation*}
  \mathcal{R}(\mathcal{B}(t)) = \sum_{j=1}^{N_k} r_j^k(t) = \sum_{j=1}^{N_k} \frac{1+t}{1+T_k} r_j^k = \frac{1+t}{1+T_k} \mathcal{R}(\mathcal{B}(T_k)) 
\end{equation*}
for every $t \in [T_k,T_{k+1})$ and every $k\in\{0\,,\ldots\,, K\}$. It thus suffices to show that $\mathcal{R}(\mathcal{B}(T_k))\leq(1+T_k)(\mathcal{R}(\mathcal{B})+N\sigma)$ for every $k\in\{0\,,\ldots\,, K\}$. For $k=0$ this is a direct consequence of~\eqref{eq:initial radii}. For $k\geq 1$, it follows inductively by applying {\rm (5)} for $t\in[T_{k-1},T_k)$ and observing that
\begin{align*}
\mathcal{R}(\mathcal{B}(T_k)) = \sum_{j=1}^{N_k} r_j^k \leq \sum_{j=1}^{N_k}\sum_{B \in S_j^k} r(B)  = \sum_{j=1}^{N_{k-1}} r_{j}^{k-1}(T_k) =  \limsup_{t \nearrow T_k} \, \mathcal{R}(\mathcal{B}(t)) \leq (1+T_k)\left(\mathcal{R}(\mathcal{B}) +N\sigma\right)\,,
\end{align*}
 which follows from~{\em ii}\,).

 Finally, property~(6) holds true by construction.
\end{proof}

\section{Proof of Theorem~\ref{thm:main}-{\normalfont \em i)} and {\normalfont \em ii)}}

In this section we prove Theorem~\ref{thm:main}-{\em i)} and {\em ii)}. We start by proving  a first estimate on the XY-energy of the auxiliary spin field, from which, however, the compactness statement does not follow straightforwardly. 

\begin{lemma} \label{lemma:XY is log2} Let $\Omega \subset \R^2$ be open, bounded, and connected let $\Omega' \subset \subset \Omega$  with  Lipschitz boundary. Let $u_\e \in \SF_\e$ be such that $E_\e(u_\e,\Omega) \leq C \e^2 |\log \e|$ and $\chi(u_\e)\to 1$, and let $v_\e \in \SF_\e$ be the auxiliary spin field defined as in~\eqref{def:from u to v}. There exists a constant $C>0$ depending on $\Omega'$ such that
    \begin{equation*}
         XY_\e(v_\e,\Omega') \leq C  \e^2 |\log \e|^2 ,
    \end{equation*}
    for $\e$ sufficiently small.
\end{lemma}  
\begin{proof}
 Let $\Omega' \subset \subset \Omega$ with  Lipschitz boundary  and assume  $\mathrm{dist}(\Omega^\prime,\partial \Omega) > \sqrt{3}\e$. Fix $\lambda \in (0,1)$  and let $\eta \in (0,1)$ be given by Lemma~\ref{lemma:bounds with XY}. For every $T \in \T_\e(\Omega')$ with $\chi(u_\e, T) > 1 - \eta$, by Lemma~~\ref{lemma:bounds with XY} we have that
        \begin{equation} \label{eq:good triangles}
             (1 - \lambda)  XY_\e(v_\e,T) \leq E_\e(u_\e,T) \, .
        \end{equation}
        For $T \in \T_\e(\Omega')$ with $\chi(u_\e, T) \leq 1 - \eta$ we estimate
        \begin{equation*} 
            XY_\e(v_\e,T) \leq 6 \e^2 .
        \end{equation*}
        and we count the number of such triangles. By Lemma~\ref{lemma:counting} we obtain that there exists $C>0$ depending on $\eta$ and $\Omega'$ such that
        \begin{equation} \label{eq:number bad triangles}
            \# \{ T \in  \T_\e(\Omega' ) : \chi(u_\e, T) \leq 1 - \eta \} \leq  \frac{ C_\eta }{\e^4}  E_\e(u_\e,\Omega)^2 \leq  C |\log \e|^2
        \end{equation} 
        for $\e$ sufficiently small.
        Putting together~\eqref{eq:good triangles}--\eqref{eq:number bad triangles}, we infer that
        \begin{equation*}
            \begin{split}
                XY_\e(v_\e, \Omega') & \leq \sum_{\substack{T \in \T_\e(\Omega') \\ \chi(u_\e,T) > 1-\eta} } \!\!\!\! XY_\e(v_\e, T) \ \ +  \! \! \sum_{\substack{T \in \T_\e(\Omega') \\ \chi(u_\e,T) \leq 1-\eta} }  \!\!\!\! XY_\e(v_\e, T) \\
                & \leq C E_\e(u_\e,\Omega) +  C \e^2 |\log \e|^2 \leq C   \e^2 |\log \e|^2 ,
            \end{split}
        \end{equation*}
        thus concluding the proof of the lemma.
\end{proof}

We are now in a position to prove the compactness statement Theorem~\ref{thm:main}-{\em i)}. Let $\Omega$ be an open, bounded set. Let $u_\e \in \SF_\e$ be such that $E_\e(u_\e,\Omega) \leq C\e^2 |\log \e|$. The fact that either $\chi(u_\e) \to 1$ or $\chi(u_\e) \to -1$ in $L^1(\Omega)$ (up to a subsequence) follows from Lemma~\ref{lemma:chi converges to 1}. In the following, we assume that $\chi(u_\e) \to 1$ and we let $v_\e \in \SF_\e$ be the auxiliary spin field defined as in~\eqref{def:from u to v}. 
 We plan to apply the Ball Construction of Lemma~\ref{lemma:ball construction} to the measures~$\mu_{v_\e}$.  

 \begin{figure}[H]
    \includegraphics{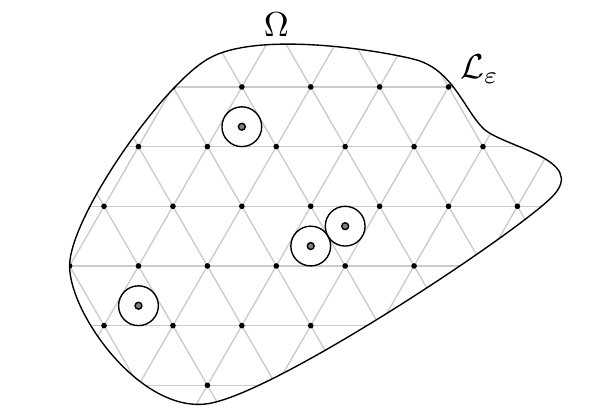}



         
        
          
          
           
         
         
         

        
        
        
        
        
        

        \caption{Example of a possible family of balls $\mathcal{B}_\e$, from which the ball construction starts.}
        \label{fig:beps}
    \end{figure}

    Let us fix $\Omega' \subset\subset \Omega'' \subset \subset \Omega$ with Lipschitz boundary. By Remark~\ref{rmk:XY controls mass} and by Lemma~\ref{lemma:XY is log2}, for $\e$ sufficiently small we have that there exists a constant $C >0$ depending on $\Omega''$ such that 
    \begin{align}\label{ineq: massboundmu}
     \#\mathrm{supp}(\mu_{v_\e}) \cap \Omega' \leq |\mu_{v_\e}|(\Omega') \leq \frac{C}{\e^2} XY_\e(v_\e, \Omega'') \leq  C  |\log \e|^2. 
    \end{align}
     We consider the family of balls 
    \begin{align}\label{def:Beps}
     \mathcal{B}_\e := \{ B_{\frac{\e}{2\sqrt{3}}}(x) : x \in \mathrm{supp}(\mu_{v_\e}) \cap \Omega' \} \, . 
    \end{align}
    Notice that each of these balls is fully contained in a triangle of the lattice, see Figure~\ref{fig:beps}.
    For every $0<r<R$ and for every $x \in \mathbb{R}^2$ such that $\A{r}{R}(x) \cap \bigcup_{B \in \mathcal{B}_\e} B =\emptyset$ we set 
    \begin{align*}
    \mathcal{E}(\mathcal{B}_\e,\mu_{v_\e},\A{r}{R}(x)) := |\mu_{v_\e}(B_r(x))| \log\frac{R}{r}\,,
    \end{align*}
     and we extend $\mathcal{E}$ to every  $A \in  \mathcal{A}(\mathbb{R}^2)$ by  
    \begin{align}\label{def: setfunction}
    \begin{split}
    \mathcal{E}(\mathcal{B}_\e,\mu_{v_\e},A) := \sup\Big\{ \sum_{j=1}^N \mathcal{E}(\mathcal{B}_\e,\mu_{v_\e},A^j)  : & \ N \in \mathbb{N} \, , \ A^j =\A{r_j}{R_j}(x_j)\, , \  A^j \cap \bigcup_{B \in \mathcal{B}_\e} B =\emptyset \, ,\\
    &  \hspace{2em} A^j \cap A^k =\emptyset \text{ for } j \neq k \, , \ A^j \subset A  \text{ for all } j \Big\}\,.
    \end{split}
    \end{align}
     We apply Lemma~\ref{lemma:ball construction} with $\sigma = 3 \e$ to $\mathcal{B} = \mathcal{B}_\e$, $\mu= \mu_{v_\e}$, and $\mathcal{E}$ defined in \eqref{def: setfunction}, which satisfy the assumptions~(B1) and~(B2) with $c_0=1$. Hence, there exists a family of balls $\{\mathcal{B}_\e(t)\}_{t \geq 0}$ satisfying {\rm (1)}--{\rm (6)} of Lemma~\ref{lemma:ball construction}. Due to~\eqref{ineq: massboundmu} and~\eqref{def:Beps}, we have that
    \begin{align}\label{ineq:radbound}
    \mathcal{R}(\mathcal{B}_\e) \leq  C \e |\log\e|^2.
    \end{align}
     Moreover, by property~(6) in Lemma~\ref{lemma:ball construction},
    \begin{equation} \label{eq:lower bound on radius}
        r(B) \geq (1+t)  \frac{\e}{2 \sqrt{3}} \quad \text{for every } B \in \mathcal{B}_\e(t) \, .
    \end{equation}

     In  the next lemma we deduce an upper bound for the set function $\mathcal{E}$. 

    \begin{lemma} Let $\mathcal{E}$, $\mathcal B_\e$, and $\mu_{v_\e}$ be as above. Let $U_\e(t) := \bigcup_{B \in \mathcal{B}_\e(t)} B$ for all $\e >0$ and $t \geq 0$. Then we have the following inequalities
        \begin{align}\label{ineq:setfunctionbound}
            \mathcal{E}(\mathcal{B}_\e,\mu_{v_\e}, \Omega''\setminus \overline{U}_\e(0)) \leq C \int_{\Omega^{\prime\prime} \setminus \overline{U}_\e(0)} |\nabla \overarc v_\e|^2 \, \d x  \leq \frac{C}{\e^2} E_\e(u_\e, \Omega) \leq C|\log \e|\,.
        \end{align}
    \end{lemma}
        \begin{proof}
         We set $U_\e  := \bigcup_{B \in \mathcal{B}_\e} B$ and we let $0<r<R$ and $x_0 \in \mathbb{R}^2$ be such that $\A{r}{R}(x_0) \cap U_\e =\emptyset$. Since $J(\overarc v_\e) = \pi\mu_{v_\e}$ and by Stokes' theorem (see also~\eqref{eq:degree for H1}), we have that
        \begin{equation*}
         \pi\mu_{v_\e}(B_s(x_0)) = J(\overarc v_\e)(B_s(x_0)) = \int_{B_s(x_0)} \curl(j(\overarc v_\e)) \, \d x = \int_{\partial B_s(x_0)} j(\overarc v_\e)  \cdot \tau_{\de B_s(x_0)} \, \d \H^1 \,
        \end{equation*} 
        for a.e.\ $s \in (r,R)$.\footnote{In fact, $\overarc v_\e |_{\de B_s(x_0)} \in H^1(\de B_s(x_0);\S^1)$ for every $s \in (r,R)$. See also Footnote~\ref{foot:actually every radius}. } Note that, since $|\overarc v_\e| =1$, we have $2|j(\overarc v_\e) | = |\nabla \overarc v_\e|$. Therefore, by Jensen's inequality,
        \begin{equation*}
            \Big|\int_{\partial B_s(x_0)} j(\overarc v_\e) \cdot \tau_{\de B_s(x_0)} \, \d \H^1 \Big|^2 \leq \frac{1}{4}\Big( \int_{\partial B_s(x_0)} |\nabla \overarc v_\e| \, \d \H^1 \Big)^{\! 2} \leq \frac{\pi}{2} s\int_{\partial B_s(x_0)} |\nabla \overarc v_\e|^2 \, \d \H^1 .
        \end{equation*}
        Since $\A{r}{R} \cap U_\e = \emptyset$, $\mu_{v_\e}(B_s(x_0)) = \mu_{v_\e}(B_r(x_0)) \in \mathbb{Z}$ for every $s \in (r,R)$. Thus, the two previous inequalities imply that 
        \begin{equation*}
            \frac{2\pi}{ s} |\mu_{v_\e}(B_r(x_0))| \leq \frac{2\pi}{ s} |\mu_{v_\e}(B_r(x_0))|^2  \leq \int_{\partial B_s(x_0)} |\nabla \overarc v_\e|^2 \, \d \H^1 .
        \end{equation*}
        Integrating in $s$ from $r$ to $R$, by the coarea formula we obtain 
        \begin{equation}\label{ineq:muboundannulus}
            \mathcal{E}(\mathcal{B}_\e,\mu_{v_\e},\A{r}{R}(x_0)) \leq C \int_{\A{r}{R}(x_0)} |\nabla \overarc v_\e|^2\mathrm{d}x\,.
        \end{equation}
        Let now $A \in \mathcal{A}(\mathbb{R}^2)$. For all $A^j$ admissible in \eqref{def: setfunction}, we have $A^j \subset A \setminus \overline U_\e$ and $A^j\cap A^k =\emptyset$ for~$j \neq k$. Therefore, using \eqref{ineq:muboundannulus}, we get
        \begin{align*}
        \sum_{j} \mathcal{E}(\mathcal{B}_\e,\mu_{v_\e},A^j) \leq \sum_{j} C \int_{A^j} |\nabla \overarc v_\e|^2 \, \d x \leq C \int_{A\setminus \overline U_\e} |\nabla \overarc v_\e|^2 \, \d x\,.
        \end{align*}
        Taking the supremum over all admissible $A^j$, we infer that 
        \begin{align*} 
            \mathcal{E}(\mathcal{B}_\e,\mu_{v_\e},A) \leq C\int_{A\setminus \overline U_\e} |\nabla \overarc v_\e|^2 \, \d x\, .
            \end{align*}
        We are now in a position to prove~\eqref{ineq:setfunctionbound}. 
        Note that Lemma~\ref{lemma:ball construction} gives 
        \begin{equation*}
            \Omega'' \setminus \overline{U}_\e(0) \subset \bigcup_{\substack{T \in \mathcal{T}_\e(\Omega) \\ |\mu_{v_\e}|(T)=0  }} T
        \end{equation*}
        thanks to the choice $\sigma=3\e$. Hence, by Lemma~\ref{lemma:rough XY bound} and the properties of $\overarc v_\e$, we obtain  
        \begin{equation*} 
        \begin{split} 
        \mathcal{E}(\mathcal{B}_\e,\mu_{v_\e}, \Omega^{\prime\prime}\setminus \overline{U}_\e(0)) &\leq C\int_{\Omega^{\prime\prime} \setminus \overline{U}_\e(0)} |\nabla \overarc v_\e|^2 \, \d x \leq C \sum_{\substack{T \in \mathcal{T}_\e(\Omega) \\ |\mu_{v_\e}|(T)=0 }} \int_T |\nabla \overarc v_\e|^2 \, \d x \\
        &\leq  \frac{C}{\e^2}\sum_{\substack{T \in \mathcal{T}_\e(\Omega) \\ |\mu_{v_\e}|(T)=0 }} \hspace{-1em} XY_\e(v_\e,T) \leq  \frac{C}{\e^2} \sum_{\substack{T \in \mathcal{T}_\e(\Omega) \\ |\mu_{v_\e}|(T)=0 }} \hspace{-1em} E_\e(u_\e,T) \leq \frac{C}{\e^2} E_\e(u_\e,\Omega) \leq C|\log \e|\,. 
        \end{split}
        \end{equation*}
        This concludes the proof of~\eqref{ineq:setfunctionbound}. 
        \end{proof}

        In the next lemma we estimate the number of merging times in the ball construction and show that the trivial estimate of order $|\log\e|^{2}$ can be improved to become of order $|\log\e|$.  By inspecting the proof, we get a better insight on the structure of the vorticity measure $\mu_{v_\e}$: the possible $|\log \e|^2$ short dipoles in the region $\chi(u_\e) \sim -1$ are annihilated at the first step of the ball construction. 

        \begin{lemma}
            Let $\mathcal{B}_\e(t)$ be as above  and let
        \begin{align*}
        \mathrm{T}^{\mathrm{merg}}_\e  := \{t \in [0,+\infty) : \#\mathcal{B}_\e(t^+) < \#\mathcal{B}_\e(t^-)\}
        \end{align*}
        denote the set of merging times.  Then there exists $M>0$ such that
        \begin{align}\label{ineq:Tmergbound}
        \# \mathrm{T}^{\mathrm{merg}}_\e  \leq M|\log\e|\,.
        \end{align}
        \end{lemma}
        \begin{proof}
        We start by proving that there exists $c>0$ such that
    \begin{align}\label{ineq:energyB}
    E_\e(u_\e,B) \geq c\e^2 \quad \text{for every } B \in \mathcal{B}_\e(0) \,.
    \end{align}
    Given $B = B_r(x) \in \mathcal{B}_\e(0)$, there exists $T_1 \in \mathcal{T}_\e(B)$ such that $|\mu_{v_\e}|(T_1)=1$.  Letting $\eta \in (0,1)$ be given by Remark~\ref{rmk:chirality and vorticity}, we have that   $\chi(u_\e,T_1) \leq 1-\eta$ (otherwise, the vorticity of $v_\e$ would be zero in~$T_1$). If additionally $-1+\eta \leq \chi(u_\e,T_1) \leq 1-\eta$, then,  by Remark~\ref{rmk:chiralityenergy},  $E_\e(u_\e,T_1) \geq C_\eta \e^2$ for some constant $C_\eta>0$ and thus~\eqref{ineq:energyB} holds true. If, instead, $\chi(u_\e,T_1) < -1+\eta$, then we argue as follows. Thanks to the choice $\sigma = 3 \e$, there exists $T' \in \T_\e(B_{r}(x) \sm \overline{B}_{r-\sigma}(x))$. Property~{(4)} in Lemma~\ref{lemma:ball construction} implies that $|\mu_{v_\e}|(T') = 0$.  Letting $\eta' \in (0,1)$ be given by Remark~\ref{rmk:vorticity and chirality}, we have that  $-1 + \eta' \leq \chi(u_\e,T')$. If $-1 + \eta' \leq \chi(u_\e,T') \leq 1- \eta'$, then $E_\e(u_\e,T') \geq C_{\eta'} \e^2$ for some constant $C_{\eta'}>0$ and thus~\eqref{ineq:energyB} holds true. Then we assume $1- \eta' < \chi(u_\e,T')$. We find now a chain of triangles $\{T_1,\ldots, T_L = T^\prime\} \subset \mathcal{T}_\e(B)$ with $T_{\ell+1} \in \mathcal{N}_\e(T_\ell)$ for all $\ell=1,\ldots,L-1$, see Figure~\ref{fig:chain}. Since $\chi(u_\e, T_1) < -1+\eta$ and $1-\eta' < \chi(u_\e, T_N)$, there exists $\ell \in \{1,\ldots,L - 1\}$ such that $\chi(u_\e,T_\ell) < 0$ and $\chi(u_\e,T_{\ell+1}) \geq 0$. Then~\eqref{ineq:energyB} follows from Lemma~\ref{lemma:energy of two triangles}.   

    \begin{figure}[H]
        \includegraphics{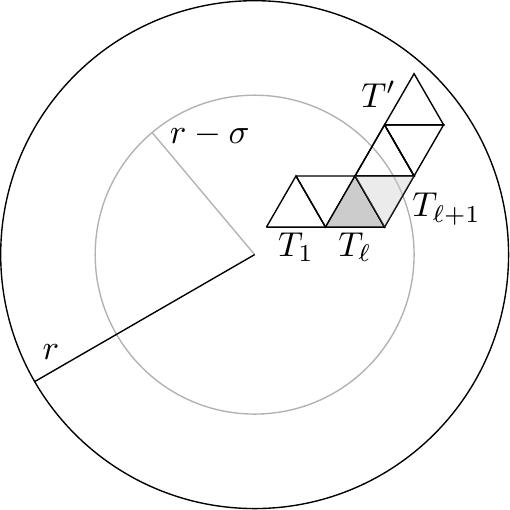}
                        













        \caption{Choice of a chain of triangles $\{T_1,\ldots, T_L = T^\prime\} \subset \mathcal{T}_\e(B)$.}
        \label{fig:chain}
    \end{figure}

    Estimate~\eqref{ineq:Tmergbound} is a consequence of~\eqref{ineq:energyB} since   
        \begin{align*}
            c \e^2 \#\mathrm{T}^{\mathrm{merg}}_\e \leq  c \e^2 \#\mathcal{B}_\e(0) \leq  \sum_{B \in \mathcal{B}_\e(0)} E_\e(u_\e,B)  \leq E_\e(u_\e,\Omega) \leq C\e^2|\log \e|\,,
        \end{align*}
        hence~\eqref{ineq:Tmergbound} follows.
\end{proof}

Let us fix $p \in (0,1)$.  (At the very end of the proof we will let $p \to 1$.)  We construct an auxiliary sequence of measures $\mu_{\e,p}$ such that~$\mu_{\e,p}$ have equibounded mass and $\mu_{\e,p}$ are close to $\mu_{v_\e}$ in the flat norm.  For  $k = 0,\ldots, \lfloor 2M|\log \e| \rfloor$  we set\footnote{ The choice of these particular expansion times will become clearer later when we deduce~\eqref{ineq:ineqXYarcannulus}. Similar arguments can be found, \eg in \cite{DLGarPon,AliBraCicDLPia}.}
    \begin{align}\label{def:time}
        \beta_p  := \exp\Big(\frac{\sqrt{p}(1 - \sqrt{p})}{2M}\Big) \, , \quad t_{\e,p}^{k} := (\beta_p)^k \e^{\sqrt{p}-1}-1 \, , 
    \end{align}  
    and 
    \begin{align}\label{def:Keps0}
    \K_\e:= \Big\{k \in \{1,\ldots,\lfloor 2M|\log \e|
    \rfloor \} : (t_{\e,p}^{k-1}, t_{\e,p}^{k}] \cap \mathrm{T}^{\mathrm{merg}}_\e = \emptyset \Big\}\,.
    \end{align}
   By~\eqref{ineq:Tmergbound}, we have that $\# \K_\e  \geq M|\log \e|$. We choose $k_\e \in \K_\e$ (depending also on $p$) such that
    \begin{align*}
    \int_{\Omega''\cap U_\e(t_{\e,p}^{k_\e}) \setminus \overline U_\e(t_{\e,p}^{k_\e-1})} |\nabla \overarc v_\e|^2 \, \d x \leq \frac{1}{\# \K_\e} \sum_{k \in \K_\e} \int_{\Omega''\cap U_\e(t_{\e,p}^{k}) \setminus \overline U_\e(t_{\e,p}^{k-1})}|\nabla \overarc v_\e|^2 \, \d x \,.  
    \end{align*}
    By conditions (1)--(2) in Lemma~\ref{lemma:ball construction} and by~\eqref{ineq:setfunctionbound} we have that
    \begin{equation*}
    \sum_{k \in \K_\e} \int_{\Omega''\cap  U_\e(t_{\e,p}^{k}) \setminus \overline U_\e(t_{\e,p}^{k-1})}|\nabla \overarc v_\e|^2\, \d x \leq \int_{\Omega'' \setminus \overline U_\e(0)}|\nabla \overarc v_\e|^2 \, \d x \leq C|\log \e|\,,
    \end{equation*}  
    whence
    \begin{align}\label{ineq:energyboundtkeps}
    \int_{\Omega''\cap U_\e(t_{\e,p}^{k_\e}) \setminus \overline U_\e(t_{\e,p}^{k_\e-1})} |\nabla \overarc v_\e|^2 \, \d x \leq C\frac{|\log \e|}{\# \K_\e} \leq C_1\,.
    \end{align}
    We define 
    \begin{align}\label{def:muepsp}
     \mu_{\e,p} :=  \sum_{B \in \mathcal{B}_\e(t_{\e,p}^{k_\e})} \mu_{v_\e}(B) \delta_{x_B}\,,
    \end{align}
    where we let $x_B$ denote the center of the ball $B$. 

    \begin{lemma} \label{lemma:muepsp}
        Let $\mu_{v_\e}$ be as above and let $\mu_{\e,p}$ be the measure defined in~\eqref{def:muepsp}. Then 
        \begin{align}\label{ineq:Muepspmassbound}
            |\mu_{\e,p}|(\Omega^\prime) \leq \frac{C}{1-\sqrt{p}} =: C_p \quad \text{and} \quad \|\mu_{v_\e} -\mu_{\e,p}\|_{\mathrm{flat},\Omega'} \to 0 \,.
            \end{align}
    \end{lemma}
    \begin{proof}    
     We start by estimating the radii of the balls in the family $\mathcal{B}_\e(t_{\e,p}^{k_\e})$ used in the definition of~$\mu_{\e,p}$. Recalling that $\sigma = 3 \e$ and that the number of balls at the start of the ball construction is $N \leq C|\log \e|^2$, by condition~(5) in Lemma~\ref{lemma:ball construction} and by~\eqref{ineq:radbound}, we infer that
    \begin{equation}\label{ineq:radboundt}
        \mathcal{R}(\mathcal{B}_\e(t_{\e,p}^{k_\e}))\leq (1+t_{\e,p}^{k_\e}) (\mathcal{R}(\mathcal{B}_\e) + C \e |\log \e|^2) \leq  C  (\beta_p)^{ 2M|\log \e| } \e^{\sqrt{p}-1} \e  |\log \e|^2  =  C \e^{p}|\log \e|^2,
    \end{equation}
  where $C$ depends on $\Omega''$.
    In particular, the balls of the family $\mathcal{B}_\e(t_{\e,p}^{k_\e})$ have infinitesimal radius as $\e \to 0$. Hence, by~\eqref{ineq:setfunctionbound} and by property~(3) in Lemma~\ref{lemma:ball construction}, for $\e$ small enough
     \begin{align*}
     & C|\log \e| \geq \mathcal{E}(\mathcal{B}_\e,\mu_{v_\e}, \Omega^{\prime\prime}\setminus \overline{U}_\e(0)) \geq \mathcal{E}(\mathcal{B}_\e,\mu_{v_\e}, \Omega^{\prime\prime} \cap U_\e(t_{\e,p}^{k_\e})\setminus \overline{U}_\e(0)) \\ 
     & \quad \geq  {\sum_{\substack{ B \in \mathcal{B}_\e(t_{\e,p}^{k_\e}) \\ x_B \in \Omega' }}} \hspace{-.8em} |\mu_{v_\e}(B)|\log (1+ t_{\e,p}^{k_\e}) \geq  |\mu_{\e,p}|(\Omega')  \, |\log (1+ t_{\e,p}^{0}) |\geq  |\mu_{\e,p}|(\Omega') (1-\sqrt{p}) |\log \e|\,,
    \end{align*} 
    which yields the estimate in~\eqref{ineq:Muepspmassbound}.
    
    To deduce the convergence in~\eqref{ineq:Muepspmassbound}, we estimate the flat distance between $\mu_{\e,p}$ and $\mu_{v_\e}$.  The argument to do this is standard (see, \eg \cite[Lemma~2.2]{DLPon}). One lets $\psi \in C^{0,1}_c(\Omega')$ be such that $\|\psi\|_{L^\infty(\Omega')} \leq 1$, $\|\nabla \psi\|_{L^\infty(\Omega')} \leq 1$. Since the balls in $\mathcal{B}_\e(t_{\e,p}^{k_\e})$ are pairwise disjoint,
   \begin{equation*}
       \begin{split}
           \langle \mu_{v_\e} -\mu_{\e,p}, \psi \rangle & = \sum_{\substack{ B \in \mathcal{B}_\e(t_{\e,p}^{k_\e}) \\ x_B \in \Omega' }} \int_B \psi \, \d (\mu_{v_\e} -\mu_{\e,p}) + \sum_{\substack{ B \in \mathcal{B}_\e(t_{\e,p}^{k_\e}) \\ x_B \notin \Omega' }} \int_{B \cap \supp(\psi)} \psi \, \d \mu_{v_\e}  \\
           & \leq \sum_{\substack{ B \in \mathcal{B}_\e(t_{\e,p}^{k_\e}) \\ x_B \in \Omega' }} \hspace{-1em} \mathrm{osc}_B(\psi)\big(|\mu_{v_\e}|+ |\mu_{\e,p}| \big)(\Omega') + \mathcal{R}(\mathcal{B}_\e(t_{\e,p}^{k_\e}))|\mu_{v_\e}|(\Omega') \\
           & \leq 2\mathcal{R}(\mathcal{B}_\e(t_{\e,p}^{k_\e})) \big(|\mu_{v_\e}|+ |\mu_{\e,p}| \big)(\Omega') \,.
       \end{split}
   \end{equation*} 
   Taking the supremum over $\psi$ in the previous inequality, by~\eqref{ineq:radboundt}, \eqref{ineq: massboundmu}, and the uniform bound in~\eqref{ineq:Muepspmassbound}, we get that
    \begin{align*}
    \|\mu_{v_\e} -\mu_{\e,p}\|_{\mathrm{flat},\Omega'} \leq C \mathcal{R}(\mathcal{B}_\e(t_{\e,p}^{k_\e})) \big(|\mu_{v_\e}|+ |\mu_{\e,p}| \big)(\Omega') \leq  C \e^{p} |\log \e|^4 \to 0 \, ,
    \end{align*}
     hence the convergence in~\eqref{ineq:Muepspmassbound} is proved.
\end{proof}

Thanks to the previous lemma, we conclude the proof of the compactness statement Theorem~\ref{thm:main}-{\em i)}. Indeed, by~\eqref{ineq:Muepspmassbound} the measures $\mu_{\e,p}\mres{\Omega'}$ converge weakly* to some measure $\mu$ in $\Omega'$, up to a subsequence. Moreover, $\mu$ is a finite sum of Dirac deltas with centers in $\Omega'$ and with integer weights, because of the structure of $\mu_{\e,p}$ in~\eqref{def:muepsp} and the uniform bound on the mass~\eqref{ineq:Muepspmassbound}. Finally, we have that\footnote{This is due to the fact that the flat norm metrizes the weak convergence of measures with equibounded mass.}  $\| \mu_{\e,p} - \mu\|_{\mathrm{flat},\Omega'} \to 0$  and thus, by~\eqref{ineq:Muepspmassbound}, $\|\mu_{v_\e} - \mu\|_{\mathrm{flat},\Omega'} \to 0$. We argue for every $\Omega' \subset \subset \Omega$ and by a diagonal argument to obtain that $\|\mu_{v_\e} - \mu\|_{\mathrm{flat},\Omega'} \to 0$ for every $\Omega' \subset \subset \Omega$. The finiteness of $|\mu|(\Omega)$ will follow from Theorem~\ref{thm:main}-{\em ii)}.

\vspace{1em}

Let us now prove Theorem~\ref{thm:main}-{\em ii)}. 
Let $u_\e \in \SF_\e$ and assume that $\chi(u_\e) \to 1$. We let $v_\e \in \SF_\e$ be the auxiliary spin field defined as in~\eqref{def:from u to v}. Let $\mu = \sum_{h = 1}^N d_h \delta_{x_h}$ with $d_h \in \Z$, $x_h \in \Omega$ and assume that $\|\mu_{v_\e} - \mu\|_{\mathrm{flat},\Omega^\prime} \to 0$ for all $\Omega^\prime\subset\subset \Omega$. Let us prove that 
\begin{equation} \label{claim:liminf}
    \liminf_{\e \to 0} \frac{1}{\e^2|\log \e|} E_\e(u_\e, \Omega)  \geq  2\sqrt{3} \pi |\mu|(\Omega) \, .
\end{equation} 
We can assume, without loss of generality, that
\begin{align*}
\liminf_{\e \to 0}\frac{1}{\e^2|\log \e|}E_\e(u_\e,\Omega) =  \lim_{\e \to 0}\frac{1}{\e^2|\log \e|}E_\e(u_\e,\Omega) <+\infty\,.
\end{align*}
Let us fix $\Omega' \subset \subset \Omega'' \subset \subset \Omega$  with Lipschitz boundary. We assume that $0 \in \Omega^\prime$ and $\mu \mres \Omega^\prime = d \delta_{0}$ for some $d \in \mathbb{Z} \setminus \{0\}$, hence $\| \mu_{v_\e} - d \delta_0\|_{\mathrm{flat},\Omega'} \to 0$. (The fact that $\mu$ is supported in $0$ is not relevant for the discussion.) Thanks to the superadditivity of the $\liminf$ and the non-negativity of the energy, it will be enough to prove the claim in $\Omega'$. 

We apply the ball construction and we define $\mu_{\e,p}$ as done above for the compactness result. By Lemma~\ref{lemma:muepsp} and the assumptions made above, we have that 
\begin{equation} \label{eq:convergence of muepsp}
    \mu_{\e,p}\mres \Omega' \wsto d\delta_0 \, .
\end{equation}

We classify the balls of the family $\mathcal{B}_\e(t_{\e,p}^{k_\e})$ into two subclasses 
    \begin{align}\label{def:subclasses}
    \begin{split}
    &\mathcal{B}_\e^{=0} := \{ B \in \mathcal{B}_\e(t_{\e,p}^{k_\e}) : \mu_{v_\e}(B)=0 \, , \ x_B \in \Omega^{\prime}\}\,, \\
     & \mathcal{B}_\e^{\neq0} := \{ B \in \mathcal{B}_\e(t_{\e,p}^{k_\e}) : \mu_{v_\e}(B)\neq0\}\,.
    \end{split}
    \end{align}
 We modify the spin field $u_\e$ in such a way that we can assume $\mathcal{B}_\e^{=0} = \emptyset$ without loss of generality. Then we will work only with balls in the family~$\mathcal{B}_\e^{\neq0}$, which are relevant from the energetic point of view.  

\begin{lemma} \label{lemma:replace}
    Let $u_\e$ be as above, let $\mathcal{B}_\e^{=0}$ be as in~\eqref{def:subclasses}, and let $c_p := \frac{\beta_p+1}{2 \beta_p} \in (0,1)$. Then there exists $\overline{u}_\e \in \mathcal{SF}_\e$ such that $\overline{u}_\e= u_\e$ on $\Omega' \setminus \bigcup_{B_R(x) \in \mathcal{B}_\e^{=0}} B_{c_p R}(x)$, $|\mu_{\overline{v}_\e}|(B)=0$ for all $B \in \mathcal{B}_\e^{=0}$, and 
    \begin{align*}
        \liminf_{\e \to 0}\frac{1}{\e^2|\log \e|} E_\e(\overline{u}_\e, \Omega') \leq  \liminf_{\e \to 0}\frac{1}{\e^2|\log \e|} E_\e(u_\e, \Omega^\prime)\,.
       \end{align*}
\end{lemma}
\begin{proof}
Let $B_{R_\e}(x_\e) \in \mathcal{B}_\e^{=0}$. Since $k_\e \in \K_\e$, by~\eqref{def:Keps0} no merging occurs in the interval $(t_{\e,p}^{k_\e-1},t_{\e,p}^{k_\e}]$ and  therefore there exists $B_{r_\e}(x_\e) \in \mathcal{B}_\e(t_{\e,p}^{k_\e-1})$ (\ie a ball with the same center). Note that, by~\eqref{eq:lower bound on radius}, 
\begin{equation} \label{eq:useful ratio}
    \frac{\e}{r_\e} \leq \frac{ C  }{1+t_\e^{k_\e-1,p}} = \frac{ C  \e^{1-\sqrt{p}}}{(\beta_p)^{k_\e -1}} \leq  C  \e^{1-\sqrt{p}} \to 0 \, .
\end{equation}
  Let $r_\e'$ be the radius of the ball centred in $x_\e$ at the last merging time  $T \leq t_{\e,p}^{k_\e-1}$  (in the case no merging occurred before  $t_{\e,p}^{k_\e-1}$,  let $T=0$). By construction, recalling~\eqref{def:time}, we have that
    \begin{align*}
    \frac{r_\e}{r_\e'} = \frac{1+t_{\e,p}^{k_\e-1}}{1+T}\,,\quad \frac{R_\e}{r_\e'} = \frac{1+t_{\e,p}^{k_\e}}{1+T} \implies \frac{R_\e}{r_\e} = \beta_p\,.
    \end{align*}
    Note that $\mu_{v_\e}(B_{r_\e}(x_\e))=0$ and, by property~(4) in Lemma~\ref{lemma:ball construction} and due to the choice $\sigma = 3 \e$, $|\mu_{v_\e}|(\A{r_\e-3\e}{R_\e+3\e}(x_\e))=0$. Furthermore, due to \eqref{est:XY-Gradient} and to \eqref{ineq:energyboundtkeps}, we have that
    \begin{align}\label{ineq:ineqXYarcannulus}
    \frac{1}{\e^2} XY_\e(v_\e,\A{r_\e}{R_\e}(x_\e)) \leq \sqrt{3}\int_{\A{r_\e}{R_\e}(x_\e)}|\nabla \overarc v_\e|^2 \, \d x \leq C_1\,.
    \end{align}
     Therefore, we are in a position to apply Lemma~\ref{lemma:extension}, see also Remark~\ref{rmk:extension}.  We obtain $\overline{v}_\e \in \mathcal{SF}_\e$ such that $\overline{v}_\e = v_\e$ on  $ \A{c_p R_\e}{R_\e}(x_\e)$ (observe that $\frac{r_\e+R_\e}{2} = c_p R_\e$),   $|\mu_{\overline{v}_\e}|(B_{R_\e}(x_\e)) =0$, and 
     \begin{align}\label{ineq:extension}
     XY_\e(\overline{v}_\e,B_{R_\e}(x_\e)) \leq  C(\beta_p)  XY_\e(v_\e,\A{r_\e}{R_\e}(x_\e)) 
     \end{align}
      for $\e$ small enough (\ie such that  $\frac{\e}{r_\e} < \frac{\beta_p-1}{C_0 C_1}(\frac{2\pi}{3})^2$, \cf~\eqref{eq:useful ratio}, where $C_0$ is given by Lemma~\ref{lemma:extension}). We set 
     \begin{equation*}
        \overline{u}_\e(\e i) := \overline{v}_\e(\e i) \, , \quad  \overline{u}_\e(\e j) := R\big[\tfrac{2\pi}{3}\big](\overline{v}_\e(\e j)) \, , \quad  \overline{u}_\e(\e k) := R\big[-\tfrac{2\pi}{3}\big](\overline{v}_\e(\e k)) 
     \end{equation*}
    for $\e i \in \mathcal{L}_\e^1$, $\e j  \in \mathcal{L}_\e^2$, $\e k \in \mathcal{L}_\e^3$ in accordance with~\eqref{def:from u to v}. By~\eqref{eq:E less than XY}, \eqref{ineq:extension}, and~\eqref{ineq:ineqXYarcannulus}, we get 
     \begin{align}\label{ineq:tildeu}
     \frac{1}{\e^2}E_\e(\overline{u}_\e, B_{R_\e}(x_\e))\leq \frac{C}{\e^2}  XY_\e(\overline{v}_\e,B_{R_\e}(x_\e)) \leq \frac{C}{\e^2} XY_\e(v_\e,\A{r_\e}{R_\e}(x_\e))\leq C\int_{\A{r_\e}{R_\e}(x_\e)} |\nabla \overarc v_\e|^2 \, \d x\, .
     \end{align}
    We apply this construction for all $B  \in \mathcal{B}_\e^{=0}$ in order to obtain $\overline{u}_\e \in \mathcal{SF}_\e$ such that $\overline{u}_\e= u_\e$ on $\Omega' \setminus \bigcup_{B_R(x) \in \mathcal{B}_\e^{=0}} B_{c_p R}(x)$, $|\mu_{\overline{v}_\e}|(B)=0$ for all $B \in \mathcal{B}_\e^{=0}$, and 
    \begin{align}\label{ineq: Beps0}
    \frac{1}{\e^2} E_\e\Big(\overline{u}_\e,\bigcup_{B \in \mathcal{B}_\e^{=0}}B\Big) \leq C \int_{U_\e(t_{\e,p}^{k_\e})\setminus \overline U_\e(t_{\e,p}^{k_\e-1})} |\nabla \overarc v_\e|^2 \, \d x \leq C\,,
    \end{align}
     where we exploited~\eqref{ineq:tildeu} and~\eqref{ineq:energyboundtkeps}.  Using \eqref{ineq: Beps0}, we therefore obtain
    \begin{align*}
     \liminf_{\e \to 0}\frac{1}{\e^2|\log \e|} E_\e(\overline{u}_\e, \Omega') &\leq   \liminf_{\e \to 0}\frac{1}{\e^2|\log \e|} \Big( E_\e(u_\e, \Omega^\prime) +E_\e\Big(\overline{u}_\e,\bigcup_{B \in \mathcal{B}_\e^{=0}}B \Big)  \Big) 
     \\&\leq \liminf_{\e \to 0}\frac{1}{\e^2|\log \e|} E_\e(u_\e, \Omega^\prime) + \limsup_{\e \to 0} \frac{C}{|\log \e|}   \\&=  \liminf_{\e \to 0}\frac{1}{\e^2|\log \e|} E_\e(u_\e, \Omega^\prime)\,.
    \end{align*}
    This concludes the proof.
\end{proof}

    Thanks to Lemma~\ref{lemma:replace}, we replace $u_\e$ by $\overline{u}_\e$ and thus we can assume hereafter that the collection~$\mathcal{B}_\e^{=0}$ is empty without loss of generality. Hence, it remains to prove the lower bound for the sequence $u_\e$ using only the family of balls $\mathcal{B}_\e^{\neq 0}$. Before going further with the proof, we obtain the lower bound in a simpler framework. Afterwards, we shall reduce to this setting. We recall that 
\begin{equation*}
    \Adm{r}{R}^\e(d) := \Big\{ u \in \SF_\e : \mu_v(T) = 0 \text{ for every } T \in \T_\e(\mathbb{R}^2)\,, \ T\cap \A{r}{R} \neq \emptyset \,,   \  \mu_v(B_r) = d  \Big\} \, ,
\end{equation*} 
where $v \in \SF_\e$ is the auxiliary spin field associated to $u$ defined as in~\eqref{def:from u to v}.

        \begin{lemma} \label{lem:liminf with exponents}
        Let $d \in \Z \sm \{0\}$ and let $0 < q_1 < q_2 < 1$. Then 
        \begin{equation*}
            \liminf_{\e \to 0} \frac{1}{\e^2|\log \e|} \inf \{ E_\e(u,\A{\e^{q_2}}{\e^{q_1}}) : u \in \Adm{\e^{q_2}}{\e^{q_1}}^\e(d) \} \geq (q_2 - q_1) 2 \sqrt{3} \pi |d|^2 .
        \end{equation*}  
        \end{lemma}
        \begin{proof}
            For every $\e$ let $u'_\e \in \Adm{\e^{q_2}}{\e^{q_1}}^\e(d)$ be such that 
            \begin{equation*}
                E_\e(u'_\e ,\A{\e^{q_2}}{\e^{q_1}}) \leq \inf \{ E_\e(u,\A{\e^{q_2}}{\e^{q_1}}) : u \in \Adm{\e^{q_2}}{\e^{q_1}}^\e(d) \} + \e^2 \, .
            \end{equation*}
            We fix $R > 1$, we set $M_{\e,R} := \lfloor  (q_2-q_1)\frac{|\log \e|}{\log R}\rfloor$ and $A^{m,\e} := \A{R^{m-1}\e^{q_2}}{R^m \e^{q_2}}$. We remark that $ \bigcup_{m=1}^{M_{\e,R}} A^{m,\e} \subset \A{\e^{q_2}}{\e^{q_1}}$. Let $\overline{m} = \overline{m}_{\e,R}$ be such that 
            \begin{equation*}
                E_\e(u'_\e ,A^{\overline{m},\e}) \leq E_\e(u'_\e ,A^{m,\e})  \, , \quad \text{ for } m=1,\dots, M_{\e,R} \, .
            \end{equation*}
            We let $\eta_\e := \e/R^{\overline{m}-1}\e^{q_2}$ and we define $u'_{\eta_\e}(\eta_\e i) := u'_\e(\e i)$ for every $i \in \L$. Then we have
            \begin{equation*}
                \frac{1}{\e^2}E_\e(u'_\e ,\A{\e^{q_2}}{\e^{q_1}}) \geq \sum_{m=1}^{M_{\e,R}} \frac{1}{\e^2} E_\e(u'_\e ,A^{m,\e})  \geq  \frac{M_{\e,R}}{\e^2} E_\e(u'_\e ,A^{\overline{m},\e}) =  \frac{M_{\e, R}}{\eta_\e^2}E_{\eta_\e}(u'_{\eta_\e}, \A{1}{R})\, .
            \end{equation*}
            Since $M_{\e,R} \geq (q_2-q_1)\frac{|\log \e|}{\log R} - 1$, from the previous inequalities, and by Proposition~\ref{prop:convergence of minima} it follows that 
            \begin{equation*}
                \begin{split}
                    & \liminf_{\e \to 0} \frac{1}{\e^2|\log \e|} \inf \{ E_\e(u,\A{\e^{q_2}}{\e^{q_1}}) : u \in \Adm{\e^{q_2}}{\e^{q_1}}^\e(d) \}  \geq \liminf_{\e \to 0} \frac{1}{\e^2|\log \e|}  E_\e(u'_\e,\A{\e^{q_2}}{\e^{q_1}}) \\
                    &\quad \geq \liminf_{\e \to 0} \frac{M_{\e, R}}{|\log \e|}  \frac{1}{\eta_\e^2} E_{\eta_\e}(u'_{\eta_\e},\A{1}{R}) \geq \liminf_{\e \to 0} \Big[ \Big( \frac{q_2-q_1}{\log R} - \frac{1}{|\log \e|}\Big) \frac{1}{\eta_\e^2} E_{\eta_\e}(u'_{\eta_\e},\A{1}{R}) \Big] \\
                    & \quad \geq  \frac{q_2-q_1}{\log R} \liminf_{\eta \to 0} \ \inf \Big\{ \frac{1}{\eta^2} E_{\eta}(u,\A{1}{R}) : u \in \Adm{1}{R}^\eta(d) \Big\} \geq (q_2-q_1) 2 \sqrt{3} \pi |d|^2 .
                \end{split}
            \end{equation*}
            This concludes the proof.  
        \end{proof}

        In view of \eqref{ineq:Muepspmassbound}, we have that $\# \mathcal{B}_\e^{\neq 0} \leq C_p$ and therefore we can assume that (up to a subsequence) $\#\mathcal{B}_\e^{\neq 0} =L$ for all $\e >0$ for some $L \in \mathbb{N}$.  Let $\mathcal{B}_\e^{\neq 0} = \{ B_{r_\e^\ell}(x_\e^\ell)\}_{\ell=1}^{L}$.  By definition~\eqref{def:muepsp}, we have that $\{x_\e^1,\ldots,x_\e^L\}$ is the support of the measure $\mu_{\e,p}$. The points~$x_\e^\ell$ converge (up to a subsequence) to a finite set of points $\{0=\xi^1,\ldots,\xi^{L^\prime}\}$ contained in $\overline{\Omega}$ with $L^\prime \leq L$. Fix $\rho >0$ such that $B_\rho \subset\subset \Omega^\prime$ and $B_\rho(\xi^h) \cap B_\rho = \emptyset$ for all $h =2,\ldots,L^\prime$. For $\e>0$ small enough we have that either $B_{r_\e^\ell}(x_\e^\ell) \cap B_\rho =\emptyset$ or $B_{r_\e^\ell}(x_\e^\ell) \subset\subset B_\rho$. Furthermore, by~\eqref{def:muepsp}, \eqref{eq:convergence of muepsp}, and the fact that $|\mu|(\partial B_\rho)=0$, we have that 
        \begin{align}\label{eq:measurecounting}
        \sum_{x_\e^\ell \in B_\rho} \mu_{v_\e}(B_{r_\e^\ell}(x_\e^\ell)) = d\,.
        \end{align}
        We prove that 
        \begin{align*}
        \liminf_{\e \to 0} \frac{1}{\e^2|\log \e|} E_\e(u_\e, B_\rho) \geq 2\sqrt{3}\pi|d|\,.
        \end{align*}
        Since our estimate is local, we can assume that $|\mu_{v_\e}|(\mathbb{R}^2\setminus B_\rho) =0$, which implies that $x_\e^\ell \in B_\rho$,  \ie $B_{r_\e^\ell}(x_\e^\ell) \subset B_\rho$,  for $\ell=1,\ldots,L$  and $\e$ small enough. To reduce to the setting in Lemma~\ref{lem:liminf with exponents} we follow an argument introduced, \eg in~\cite{DLGarPon} or \cite{AliBraCicDLPia}. It is aimed at separating the scales of the radii of the balls charged by $\mu_{v_\e}$. 
        
          Fix $0< p' < p'' < p$ such that $\mathcal{R}(\mathcal{B}_\e(t_{\e,p}^{k_\e})) \leq \e^{p''}$ (this is possible due to  \eqref{ineq:radboundt}). We consider the function $g_\e \colon [p',p''] \to \{1,\ldots,L\}$ such that $g_\e(q)$ gives the number of connected components of $\bigcup_{\ell=1}^L B_{\e^{q}}(x_\e^\ell)$. For each $\e >0$, the function $g_\e$ is monotonically non-decreasing so that it can have at most  $\hat{L} \leq L-1$ discontinuity points. We let $\{q^\e_1,\dots,q^\e_{\hat L}\}$ denote these discontinuity points with 
        \begin{align*}
        p' \leq q^\e_1 < \ldots < q^\e_{\hat{L}} \leq p''.
        \end{align*}
       There exists a finite set  $\mathfrak{D} =\{q_0,\ldots,q_{\tilde{L}+1}\}$ with $q_h < q_{h+1}$ such that, up to a subsequence, $(q^\e_j)_\e$ converges to some point in $\mathfrak{D}$ as $\e \to 0$, for $j=1,\ldots , \hat{L}$. We set $q_0=p'$,   $q_{\tilde{L}+1}=p''$, and thus $\tilde{L} \leq \hat{L}$. Let us fix $\lambda >0$ with $2\lambda<\min_h (q_{h+1}-q_h)$. For $\e>0$ small enough (that is, such that for $h'=1,\ldots, \hat{L}$ one has $|q^\e_{h'}- q_h| <\lambda/2$ for some $q_h \in \mathfrak{D}$) the function~$g_\e$ is constant in the interval $[q_h+\lambda/2,q_{h+1}-\lambda/2]$ with constant value $M^\e_h$,  where $M^\e_h \leq L$. Up to extracting a subsequence, we assume that $M^\e_h = M_h$. We now construct a family of annuli $\{A_\e^{h,m}\}_{m=1}^{M_h}$ where we can apply Lemma~\ref{lem:liminf with exponents}.

\begin{lemma}
        In the assumptions above,  for $\e$ sufficiently small,  for every  $h=0,\ldots,\tilde{L}$  there exists a family of  pairwise disjoint  annuli $\{A_\e^{h,m}\}_{m=1}^{M_h}$ with $A_\e^{h,m} := B_{\e^{q_h+\lambda}} (z_\e^{h,m}) \setminus \overline{B}_{\e^{q_{h+1}-\lambda}}(z_\e^{h,m})$ such that the sets in the family  $\{ \bigcup_{m=1}^{M_h} A_\e^{h,m} \}_{h=1}^{\tilde L}$ are pairwise disjoint  and 
        \begin{align}\label{incl:BepsCepsi}
            \bigcup_{\ell=1}^L B_{r_\e^\ell}(x_\e^\ell) \subset \bigcup_{m=1}^{M_h} B_{\e^{q_{h+1}-\lambda}}(z_\e^{h,m}) \, 
        \end{align}
         for $h=0,\ldots,\tilde{L}$.  Moreover, the points $z_\e^{h,m}$ can be chosen in $\L_\e\cap\bigcup_{\ell=1}^L B_\e(x_\e^\ell)$. 
    \end{lemma}
    \begin{proof}
    
   Let $h\in\{0\,,\ldots\,, \tilde{L}\}$. Since $g_\e \equiv M_{h}$ on $[q_{h} + \tfrac{\lambda}{2}, q_{h+1} - \tfrac{\lambda}{2}]$,  we find a partition $\{ \mathcal{I}_\e^{h,m} \}_{m=1}^{M_{h}}$ of $\{ 1, \dots, L\}$ such that $\{ \bigcup_{\ell \in \mathcal{I}_\e^{h,m}} B_{\e^{q}}(x_\e^{\ell}) \}_{m=1}^{M_{h}}$ are the $M_{h}$ connected components of $\bigcup_{\ell=1}^{L} B_{\e^{q}}(x_\e^{\ell})$ for $q \in [q_{h} + \tfrac{\lambda}{2}, q_{h+1} - \tfrac{\lambda}{2}]$. For $m = 1, \dots, M_{h}$ we choose arbitrarily $\ell(m) \in \mathcal{I}_\e^{h,m}$ and $z_\e^{h,m}\in\L_\e\cap B_\e(x_\e^{\ell(m)})$. For $\e$ small enough the balls in $\{ B_{\e^{q}}(z_\e^{h,m}) \}_{m=1}^{M_{h}}$ are pairwise disjoint for $q \in [q_{h} + \lambda, q_{h+1} -\lambda]$, since $B_{\e^q}(z_\e^{h,m})\subset B_{\e^q+\e}(x_\e^{\ell(m)})\subset B_{\e^{q_h+\frac{\lambda}{2}}}(x_\e^{\ell(m)})$, thus each of the balls is contained in a different connected component. Moreover, \eqref{incl:BepsCepsi} holds true by construction. Indeed, let $x\in B_{r_\e^\ell}(x_\e^\ell)\subset B_{\e^{p''}}(x_\e^\ell)$ for some $\ell\in\{1\,,\ldots\,, L\}$, and let $m_\ell\in\{1\,,\ldots\,, M_h\}$ with $\ell\in\mathcal{I}_\e^{h,m_\ell}$. Then
   \begin{equation*}
   |x-z_\e^{h,m}|\leq|x-x_\e^\ell|+|x_\e^{\ell}-z_\e^{h,m_\ell}|\leq\e^{p''}+\e +(\#\mathcal{I}_\e^{h,m_\ell}-1)\e^{q_{h+1}-\frac{\lambda}{2}}\leq\e^{p''}+\e +L\e^{q_{h+1}-\frac{\lambda}{2}}\ll\e^{q_{h+1}-\lambda},
   \end{equation*}
   for $\e$ sufficiently small (depending on $\lambda$), which gives~\eqref{incl:BepsCepsi}.
   Let us finally prove that  
        \begin{equation} \label{eq:from h to h-1}
            \bigcup_{m=1}^{M_{h}} B_{\e^{q_h+\lambda}}(z_\e^{h,m}) \subset  \bigcup_{n=1}^{M_{h-1}} B_{\e^{q_h-\lambda}}(z_\e^{h-1,n})\quad\text{for } h = 1\,,\ldots\,,\tilde{L} \, ,
        \end{equation}
        which implies that $\bigcup_{m=1}^{M_h} A_\e^{h,m}$ and $\bigcup_{n=1}^{M_{h-1}} A_\e^{h-1,n}$ are disjoint. 
        To prove~\eqref{eq:from h to h-1}, let $m \in \{1, \dots, M_{h}\}$ and let $x \in B_{\e^{q_h+\lambda}}(z_\e^{h,m})$  with $z_\e^{h,m}\in\L_\e\cap B_\e(x_\e^{\ell(m)})$. Moreover, let $n_m\in\{1\,,\ldots\,, M_{h-1}\}$ with $\ell(m)\in\mathcal{I}_\e^{h-1,n_m}$. Then a similar argument as above shows that
        \begin{equation*}
        |x-z_\e^{h-1,n_m}|\leq|x-z_\e^{h,m}|+|z_\e^{h,m}-z_\e^{h-1,n_m}|
        \leq \e^{q_h+\lambda}+2\e+L\e^{q_h-\frac{\lambda}{2}}\ll\e^{q_h-\lambda},
        \end{equation*}
        for $\e$ sufficiently small.
    \end{proof}

        Finally, we conclude by exploiting the annuli $A_\e^{h,m}$ to prove the lower bound. Note that, for~$\e$ small enough $A_\e^{h,m} \subset\subset \Omega^\prime$ for  $h=0,\dots,\tilde L$ and $m=1,\dots,M_h$. Moreover, in view of~\eqref{def:muepsp} and~\eqref{ineq:Muepspmassbound}, we have that $|\mu_{v_\e}(B_{\e^{q_{h+1}-\lambda}}(z_\e^{h,m})) | \leq C$ for $h=0,\ldots, \tilde{L}$ and $m=1,\ldots, M_h$. Therefore, up to extracting a further subsequence, $\mu_{v_\e}(B_{\e^{q_{h+1}-\lambda}}(z_\e^{h,m}))=d_{h,m} \in \mathbb{Z} \setminus \{0\}$ with $M_h$ and $d_{h,m}$ independent of $\e$. Finally, by~\eqref{eq:measurecounting}, we have
        \begin{align}\label{eq:sum=d}
        \sum_{m=1}^{M_h} d_{h,m} = d\,.
        \end{align}
        As $u_\e (\, \cdot \, - z_\e^{h,m}) \in \Adm{\e^{q_{h+1}-\lambda}}{\e^{q_h+\lambda}}^\e(d_{h,m})$, since $\mathcal{B}_\e^{=0}=\emptyset$, by property~(4) in Lemma~\ref{lemma:ball construction}  (recalling that $\sigma=3\e$), by~\eqref{incl:BepsCepsi}, and by Lemma~\ref{lem:liminf with exponents}, for every $h$ and $m$ we get
        \begin{align*}
        \liminf_{\e \to 0} \frac{1}{\e^2|\log \e|} E_\e(u_\e, A_\e^{h,m}) \geq (q_{h+1} -q_h -2\lambda) 2\sqrt{3}\pi |d_{h,m}|^2 \geq (q_{h+1} -q_h -2\lambda) 2\sqrt{3}\pi |d_{h,m}|\,,
        \end{align*}
        which, summing over $h$ and $m$ and using \eqref{eq:sum=d}, yields
        \begin{equation*} 
        \begin{split}
        \liminf_{\e \to 0} \frac{1}{\e^2|\log \e|}E_\e(u_\e,\Omega^\prime)&\geq \sum_{h=0}^{\tilde{L}} \sum_{m=1}^{M_h}  (q_{h+1} -q_h -2\lambda) 2\sqrt{3}\pi |d_{h,m}| \geq \sum_{h=0}^{\tilde{L}}(q_{h+1} -q_h -2\lambda) 2\sqrt{3}\pi |d| \\&= (p''-p'-2(\tilde{L}+1)\lambda)2\sqrt{3}\pi |d|   = (p''-p'-2(\tilde{L}+1)\lambda)2\sqrt{3}\pi|\mu|(\Omega^\prime)  \,.
        \end{split}
        \end{equation*}
         The claim follows letting $\lambda \to 0$, $p' \to 0$, $p''\to p$, and $p \to 1$ in the previous inequality. Thanks to the arbitrariness of $\Omega'$, we have proven~\eqref{claim:liminf}. 
      
\begin{remark}
    It is possible to obtain a non-sharp lower bound on $E_\e(u,\Omega)$ in terms of another auxiliary variable -- the spin field $u^1$ obtained by restricting $u$ to the sublattice $\mathcal{L}_\e^1$. Let $\hat{T}$ be a plaquette in the sublattice $\mathcal{L}_\e^1$, namely $\hat{T}= \conv\{\e i, \e i', \e i''\}$, $\e i, \e i', \e i'' \in \mathcal{L}_\e^1$, $|\e i-\e i'|=|\e i-\e i''|=|\e i'-\e i''|=\sqrt{3}\e$. We define
\begin{equation}\label{eq:XYsqrt3}
\begin{split}
XY_{\sqrt{3}\e}(u,\hat{T}) &= \frac{3}{2}\e^2(|u(\e i)-u(\e i')|^2 + |u(\e i)-u(\e i'')|^2 + |u(\e i')-u(\e i'')|^2) \\&= (\sqrt{3}\e)^2 \sqrt{3}\int_{\hat{T}} |\nabla \hat{u}^1|^2\mathrm{d}x\,,
\end{split}
\end{equation}
where $\hat{u}^1$ is the affine interpolation in $\hat{T}$ of the spin field $u^1$.
Let $H$ be the hexagon composed of the $6$ triangles in $\mathcal{T}_\e(\R^2)$ that intersect the interior of $\hat{T}$. By convexity of $x \mapsto|x|^2$ we get
\begin{equation*}
E_\e(u,H)  \geq \frac{1}{2} \e^2(|u(\e i)-u(\e i')|^2 + |u(\e i)-u(\e i'')|^2 + |u(\e i')-u(\e i'')|^2) =\frac{1}{3} XY_{\sqrt{3}\e}(u^1,\hat{T})  \,.
\end{equation*}
Summing over all triangles $\hat{T}$ of the sublattice $\mathcal{L}_\e^1$ and noticing that the energy  of every hexagon $H$ is counted twice, we obtain  
\begin{align*}
2 E_\e(u,\Omega) \geq \frac{1}{3} XY_{\sqrt{3}\e}(u^1,\Omega') \,,
\end{align*}
for all $\Omega' \subset\subset \Omega$ such that $\mathrm{dist}(\Omega', \partial \Omega)  > \e$. We therefore obtain the following non-sharp lower bound (\cf~\cite{DL} and \eqref{eq:XYsqrt3}): 
If $u_\e \in \mathcal{SF}_\e$ satisfies $ \|\mu_{u_\e^1} -\mu\|_{{\rm flat},\Omega'} \to 0$ for all $\Omega'\subset\subset \Omega$, then
\begin{equation*}
\liminf_{\e \to 0} \frac{1}{\e^2|\log \e|} E_\e(u,\Omega) \geq \liminf_{\e \to 0}\frac{1}{2} \frac{1}{(\sqrt{3}\e)^2|\log (\sqrt{3}\e)|} XY_{\sqrt{3}\e}(u_\e^1,\Omega') \geq \sqrt{3}\pi|\mu|(\Omega')\,.
\end{equation*}
\end{remark} 

\section{Proof of Theorem~\ref{thm:main}-{\normalfont \em iii)}}

In this section we prove Theorem~\ref{thm:main}-{\em iii)}.  We start with an upper bound of the XY-energy of the prototypical function with a vortex-like singularity. 
  
\begin{lemma} \label{lemma:XY of x/|x|}
    For $2 \e \leq r \leq R$ one has that 
    \begin{equation} \label{eq:XY on annulus}
        XY_\e\Big(\Big(\frac{x}{|x|}\Big)^{\! d}, \A{r}{R} \Big) \leq 2 \sqrt{3} \pi |d|^2  \e^2 \log \Big(\frac{R}{r }\Big)  + C \e^2 \, .
    \end{equation} 
\end{lemma}
\begin{proof}
   The computation is standard, but we present it for the sake of completeness. Set $v(x):=\big(\tfrac{x}{|x|}\big)^{\! d}$ for $x \in \R^2 \sm \{0\}$ and $v_\e(x) := \big(\tfrac{x}{|x|}\big)^{\! d}$ for $x \in \L_\e \sm \{0\}$ and $v_\e(0) := e_1$. Let $\alpha \in \{1,2,3\}$ and let $\e i, \e j \in T$ with $j-i$ parallel to $\tb_\alpha$. For every $x \in T$, we have that 
\begin{equation*}
    \begin{split}
       |\nabla \hat v_\e(x)\tb_\alpha|^2 & = \frac{|v(\e j) - v(\e i)|^2}{\e^2} =   \frac{1}{\e^2} \Big|\int_0^1 \nabla v(\e i + t (\e j - \e i))  (\e j-\e i) \, \d t \Big|^2 \\
       &  \leq \int_0^1 |\nabla v(\e i+t(\e j-\e i)) \tb_\alpha |^2 \, \d t \\
       & \leq  |\nabla v(x) \tb_\alpha|^2 +  \int_0^1 \Big( |\nabla v(\e i + t (\e j - \e i)) \tb_\alpha|^2 - |\nabla v(x) \tb_\alpha|^2 \Big) \, \d t \, .
    \end{split}
\end{equation*}
We let $z := \e i + t (\e j - \e i)$ and we find $\xi$ in the segment $[x,z]$ such that 
\begin{equation*}
    \begin{split}
        |\nabla v(z) \tb_\alpha|^2  - |\nabla v(x) \tb_\alpha|^2  & \leq |\nabla v(z)   - \nabla v(x)  |(|\nabla v(z)| + |\nabla v(x)|)  \leq \e |\nabla^2 v(\xi) | (|\nabla v(z)| + |\nabla v(x)|) \\
        & \leq \e \frac{C(d)}{|\xi|^2} \Big(\frac{|d|}{|z|} + \frac{|d|}{|x|}\Big) \leq \e \frac{C}{(|x|-\e)^3} \, .
    \end{split}
\end{equation*}
where we used the fact that $|\nabla v(x)|=\frac{|d|}{|x|}$, $|\nabla^2 v(\xi)| \leq \frac{C(d)}{|\xi|^2}$,\footnote{This follows, \eg, by a computation in polar coordinates which shows that, for $h=1,2$,
\begin{equation*}
    \begin{split}    
        \nabla^2 v^h(x) &  = \frac{1}{\rho^2} \begin{pmatrix}
       2 \de_\theta v^h \sin\theta  \cos \theta +   \de^2_\theta v^h \sin^2 \theta    &    -\de_\theta v^h \cos(2 \theta)  -    \de^2_\theta v^h \sin  \theta \cos \theta \\
       -\de_\theta v^h \cos(2 \theta)  -    \de^2_\theta v^h \sin  \theta \cos \theta  &  -  2 \de_\theta v^h \sin\theta  \cos \theta +    \de^2_\theta v^h \cos^2 \theta
    \end{pmatrix} \, .
    \end{split}
    \end{equation*}
} and $\min\{|x|,|z|,|\xi|\} \geq |x| - \e$.
We conclude that for every $x \in T$ 
\begin{equation*}
    |\nabla \hat v_\e(x)|^2  \leq  |\nabla v(x)|^2 + \e \frac{C}{(|x| - \e)^3} \, .
\end{equation*}
Therefore, by Remark~\ref{rmk:XY is integral} 
\begin{equation*}
    \begin{split}
        XY_\e(v_\e, \A{r}{R}) & \leq  \sqrt{3} \e^2 \int_{\A{r}{R}} |\nabla \hat v_\e(x)|^2 \, \d x \leq \sqrt{3} \e^2 \int_{\A{r}{R}} |\nabla v(x)|^2 \, \d x + \int_{\A{r}{R}}     \frac{C\e^3}{(|x| - \e)^3} \, \d x \\
        & = \sqrt{3} \e^2 \int_{r }^{R} \int_0^{2 \pi} \frac{|d|^2}{\rho}  \, \d \theta \, \d \rho +  C \e^3 \int_{r}^{R} \int_0^{2 \pi} \frac{\rho}{(\rho-\e)^3}  \, \d \theta \, \d \rho \\
        & \leq   2 \sqrt{3} \pi |d|^2 \e^2 \log \Big( \frac{R}{r}\Big) +  C \e^2  \Big( \frac{\e}{r -\e} - \frac{\e}{R -\e} + \frac{\e^2}{ 2(r -\e)^2 } - \frac{\e^2}{ 2(R -\e)^2}\Big) \, ,
    \end{split}
\end{equation*} 
whence \eqref{eq:XY on annulus}.  
\end{proof}

Let us prove Theorem~\ref{thm:main}-{\em iii)}. Let $\mu=\sum_{h=1}^N d_h\delta_{x_h}$ with $d_h\in\Z$ and $x_h\in\Omega$. Let us prove that there exist $u_\e\in\SF_\e$ such that  $\|\mu_{v_\e}-\mu\|_{\mathrm{flat},\Omega} \to 0$,  where $v_\e$ is as in~\eqref{def:from u to v}, and
\begin{equation}\label{est:Gamma-limsup}
\limsup_{\e\to 0}\frac{1}{\e^2|\log\e|}  E_\e(u_\e,\Omega)  \leq 2\sqrt{3}\pi|\mu|(\Omega)\,.
\end{equation}
\begin{step}{1} (The case $\mu=\pm\delta_{x_1}$)
 Let $x_1\in\Omega$ and $\mu=\pm\delta_{x_1}$. It is not restrictive to assume that $x_1=0\in\Omega$ and $\mu=\delta_0$. We define $v_\e\in \SF_\e$ by setting $v_\e(x):=\tfrac{x}{|x|}$ for every $x\in\L_\e\sm\{0\}$, $v_\e(0)\defas e_1$ and we set
\begin{equation}\label{def:ue-recovery}
u_\e(\e i)\defas v_\e(\e i)\, , \quad u_\e(\e j)\defas R[\tfrac{2\pi}{3}] (v_\e(\e j)) \, , \quad u_\e(\e k) := R[-\tfrac{2\pi}{3}](v_\e(\e k))\, ,
\end{equation}
for  $\e i\in\L_\e^1$,  $\e j\in\L_\e^2$, and $\e k\in\L_\e^3$, where $R[\,\cdot\,]$ is as in~\eqref{def:rotation}.  We now estimate~$E_\e(u_\e,\Omega)$ in terms of $XY_\e(v_\e,\Omega)$, then we can conclude using~\eqref{eq:XY on annulus}. To this end, let us fix  $\lambda\in(0,1)$   and let  $\eta \in (0,1)$  be as in Lemma~\ref{lemma:bounds with XY}. 
We observe that for every $T=\conv\{\e i,\e j,\e k\}\in\T_\e(\R^2)$ with $\e i\in\L_\e^1$, $\e j\in\L_\e^2$, and $\e k\in\L_\e^3$ we have
\begin{equation}\label{est:differences}
\frac{2}{\pi} \d_{\S^1}(v_\e(\e i), v_\e(\e j)) \leq |v_\e(\e i)-v_\e(\e j)|=  \Big| \frac{i}{|i|} - \frac{j}{|j|}\Big| \leq \Big| \frac{i}{|i|} - \frac{j}{|i|}\Big| + \Big| \frac{j}{|i|} - \frac{j}{|j|}\Big| \leq  \frac{2\e}{|\e i|}\, .
\end{equation}
Since the same reasoning holds for $\d_{\S^1}(v_\e(\e i),v_\e(\e k))$, we find $K\in\N$ (depending on  $\eta$)  such that 
\begin{equation}\label{cond:geo-distance-delta}
\d_{\S^1}(v_\e(\e i),v_\e(\e j)) <  \min\big\{ \eta , \tfrac{\pi}{2} \big\}  \quad \text{and} \quad \d_{\S^1}(v_\e(\e i),v_\e(\e k))<  \min\big\{ \eta , \tfrac{\pi}{2} \big\}  \, ,
\end{equation}
whenever $T\cap(\R^2\sm B_{K\e})\neq\emptyset$. Thanks to Lemma~\ref{lemma:bounds with XY} this allows us to estimate  $E_\e(u_\e,\Omega)$  via
\begin{equation}\label{est:energysplitting-rec-seq}
   E_\e(u_\e,\Omega)  \leq E_\e(u_\e, B_{(K+2)\e})+ (1+\lambda)  XY_\e(v_\e,\Omega\sm \overline B_{K\e})\leq E_\e(u_\e, B_{(K+2)\e})+ (1+\lambda)  XY_\e(v_\e,\A{K\e}{R})\,,
\end{equation}
where $R>0$ is chosen large enough such that $\Omega\subset \subset B_R$. Moreover, we have
\begin{equation*}
E_\e(u_\e,B_{(K+2)\e})\leq 3\e^2\#\T_\e(B_{(K+2)\e})\leq C(K+2)^2\e^2 . 
\end{equation*}
Thus, from~\eqref{est:energysplitting-rec-seq} together with Remark~\ref{lemma:XY of x/|x|} we infer
\begin{equation}\label{est:energy-rec-seq}
    \frac{1}{\e^2|\log\e|}  E_\e(u_\e,\Omega) \leq \frac{1}{\e^2|\log\e|} E_\e(u_\e,B_R)  \leq\frac{ C(K+2)^2 }{|\log\e|}+  (1+\lambda)  2\sqrt{3}\pi\frac{1}{|\log\e|}\log\frac{R}{ K\e }\,,
\end{equation}
from which we deduce~\eqref{est:Gamma-limsup} by letting $\e\to 0$ and then  $\lambda\to 0$.  To conclude the proof it thus remains to show that   $\|\mu_{v_\e}-\delta_0\|_{\mathrm{flat},\Omega} \to 0$. First of all, due to Theorem~\ref{thm:main}-{\em i)}, we have that there exists $\mu = \sum_{h=1}^N d_h \delta_{x_h}$ with $d_h \in \mathbb{Z}$ and $x_h \in \Omega$ such that, up to a subsequence  $\|\mu_{v_\e}-\mu\|_{\mathrm{flat},\Omega'}\to 0$  for all $\Omega^\prime \subset\subset \Omega$. Note that, thanks to~\eqref{cond:geo-distance-delta}, we have $\mu_{v_\e}=0$ on $\R^2\sm B_{K\e}$, which in turn implies that  $\|\mu_{v_\e}- d\delta_0\|_{\mathrm{flat},\Omega}\to 0$ for some $d\in\Z$.  We claim that $d=1$. Indeed, let $\overarc v_\e$ be the interpolation defined as in Remark~\ref{rmk:arcwise interpolation}. Note that $\overarc v_\e \in W^{1,\infty}(A_{1,2};\S^1)$, since $\mu_{v_\e}=0$ on $\R^2\sm B_{K\e}$. Let~$\zeta\colon[0,3]\to\R$ be the piecewise affine function satisfying $\zeta\equiv 1$ on $[0,1]$, $\zeta\equiv 0$ on $[2,3]$, and $\zeta$ affine on $[1,2]$ and set $\psi(x)\defas \zeta(|x|)$. Then,
\begin{equation*}
    \langle d\delta_0,\psi\rangle =\lim_{\e\to 0}\langle\mu_{v_\e},\psi\rangle=\frac{1}{\pi}\lim_{\e\to 0}\langle J(\overarc{v}_\e),\psi\rangle=-\frac{1}{\pi}\lim_{\e\to 0}\int_{\A{1}{2}} \!\!\! j(\overarc{v}_\e)\cdot\nabla^\perp\psi\dx=-\frac{1}{\pi}\int_{\A{1}{2}} \!\!\! j\big(\tfrac{x}{|x|}\big)\cdot\nabla^\perp\psi\dx\,,
    \end{equation*}
    where in the last step we used that $\overarc{v}_\e\wto\frac{x}{|x|}$ weakly in $H^{1}(\A{1}{2};\R^2)$. Moreover, $\nabla^\perp\psi(x)=-\frac{\, x^\perp}{|x|}$ on $\A{1}{2}$, thus a direct computation shows that
    \begin{equation*}
    \langle d\delta_0,\psi\rangle=\frac{1}{2\pi}\int_{\A{1}{2}}\frac{1}{|x|}\dx=1\,,
    \end{equation*}
    consequently $d=1$ and the whole sequence converges. 
\end{step}

\begin{step}{2} (The case $\mu=\sum_{h=1}^N\pm\delta_{x_h}$)
We first construct a recovery sequence when $\mu=\delta_{x_1}+\delta_{x_2}$ with $x_1,x_2\in\Omega$  and $x_1 \neq x_2$.  To simplify the exposition and the notation we assume that $x_1=0$ and we set $\overline{x}\defas x_2$. Then, to define a recovery sequence $u_\e$ for $\mu=\delta_{0}+\delta_{\overline{x}}$, we choose $\overline{x}_\e\in\L_\e\cap B_{2\e}(\overline{x})$ and we set $w_{\e}(x)\defas \frac{x}{|x|}$ for $x\in\L_\e\sm\{0\}$, $\overline{w}_{\e}(x)\defas\frac{x-\overline{x}_\e}{|x-\overline{x}_\e|}$ for $x\in\L_\e\sm\{\overline{x}_\e\}$ and $w_{\e}(0)=\overline{w}_{\e}(\overline{x}_\e)\defas e_1$. Eventually, we define $v_\e\in\SF_\e$ by setting $v_\e(x)\defas w_{\e}(x)\odot \overline{w}_{\e}(x)$ for every $x\in\L_\e$, where $\odot$ denotes the complex product, and we define $u_\e$ according to~\eqref{def:ue-recovery}.
Suppose now that $T=\conv\{\e i,\e j,\e k\}$ with $\e i\in\L_\e^1$, $\e j\in\L_\e^2$, and $\e k\in\L_\e^3$. Then
\begin{equation}\label{est:complex-product}
\begin{split}
|v_\e(\e i)-v_\e(\e j)| &=\Big|\Big(\frac{\e i}{|\e i|}-\frac{\e j}{|\e j|}\Big)\odot\frac{\e i-\overline{x}_\e}{|\e i-\overline{x}_\e|}+\frac{\e j}{|\e j|}\odot\Big(\frac{\e i-\overline{x}_\e}{|\e i-\overline{x}_\e|}-\frac{\e j-\overline{x}_\e}{|\e j-\overline{x}_\e|}\Big)\Big|\\
&\leq|w_{\e}(\e i)-w_{\e}(\e j)|+|\overline{w}_{\e}(\e i)-\overline{w}_{\e}(\e j)|\, .
\end{split}
\end{equation}
Taking  and expanding  the square in~\eqref{est:complex-product} yields
\begin{align}
        & |v_\e(\e i)-v_\e(\e j)|^2  \leq |w_{\e}(\e i)-w_{\e}(\e j)|^2+|\overline{w}_{\e}(\e i)-\overline{w}_{\e}(\e j)|^2+ \frac{2}{\e^2} XY_\e(w_\e,T)^\frac{1}{2} XY_\e(\overline{w}_\e,T)^\frac{1}{2}.  \label{est:complex-product-square-2}
\end{align} 
The same estimates hold true when either $\e i$ or $\e j$ is replaced by $\e k$. Thus, in view of~\eqref{est:complex-product} and~\eqref{est:complex-product-square-2} we can estimate $E_\e(u_\e,\Omega)$ as follows:
Letting  $\lambda\in (0,1)$ and $\eta \in (0,1)$  be  as in Step~1, from~\eqref{est:complex-product} we deduce the existence of $K\in\N$ such that  $\d_{\S^1}(v_\e(\e i),v_\e(\e j))< \min\big\{ \eta , \tfrac{\pi}{2} \big\}$ and $\d_{\S^1}(v_\e(\e i),v_\e(\e k))< \min\big\{ \eta , \tfrac{\pi}{2} \big\}$,  whenever $T \cap (\R^2\sm (B_{K\e}\cup B_{K\e}(\overline{x}_\e))) \neq \emptyset$. Then, thanks to Lemma~\ref{lemma:bounds with XY} and~\eqref{est:complex-product-square-2} we get
\begin{equation*}
E_\e(u_\e,\Omega)\leq C(K+2)^2\e^2+  (1+\lambda) \Big(XY_\e\big(w_{\e},\Omega\sm \overline B_{K\e})+XY_\e\big(\overline{w}_{\e},\Omega\sm  \overline B_{K\e}(\overline{x}_\e))+  6 I_\e  \Big)\, ,
\end{equation*}
where the remainder $I_\e$ is given by
\begin{equation*}
 I_\e := \sum_{T \in \T_\e(\Omega\sm(\overline B_{K\e}\cup \overline B_{K\e}(\overline{x}_\e))} \hspace{-2em} XY_\e(w_\e,T)^\frac{1}{2} XY_\e(\overline w_\e,T)^\frac{1}{2} \,.  
\end{equation*}

To conclude as in~\eqref{est:energy-rec-seq}, it is enough to show that $I_\e \leq C \e^2$. We split the sum in the definition of~$I_\e$. We fix $r> K\e$ such that $B_{r+2\e} \cap B_{r+2\e}(\overline{x}_\e) = \emptyset$ and $B_{r+2\e} \cup B_{r+2\e}(\overline{x}_\e) \subset \subset \Omega$. We also fix $R>r+2\e$ such that $\Omega \subset \subset B_R \cap B_R(\overline{x}_\e)$. Then, by the Cauchy-Schwarz Inequality and Lemma~\ref{lemma:XY of x/|x|},
\begin{equation} \label{eq:I1}
    \begin{split}
        \sum_{T \in \T_\e(\Omega\sm(\overline B_{r}\cup \overline B_{r}(\overline{x}_\e))} \hspace{-2em} XY_\e(w_\e,T)^\frac{1}{2} XY_\e(\overline w_\e,T)^\frac{1}{2} & \leq \big(XY_\e(w_\e,\A{r}{R})\big)^{\frac{1}{2}} \big(XY_\e(\overline w_\e,\A{r}{R}(\overline x_\e))\big)^{\frac{1}{2}} \\
        & \leq 2 \sqrt{3} \pi |d|^2  \e^2 \log \Big(\frac{R}{r }\Big)  + C \e^2  \leq C\e^2 .
    \end{split}
\end{equation}
Let $T \in \T_\e(\R^2 \sm \overline B_{K\e})$. Estimate~\eqref{est:differences} implies that  $\frac{1}{\e^2}XY_\e(w_\e,T) \leq \frac{12 \e^2}{\dist(T,0)^2}$. In particular, 
\begin{equation*}
    XY_\e(w_\e,T)^\frac{1}{2} \leq \frac{\sqrt{12} \e^2 }{\dist(T,0)} = \fint_T \frac{\sqrt{12} \e^2 }{\dist(T,0)} \, \d x \leq \fint_T \frac{\sqrt{12} \e^2}{|x| - \e}  \, \d x \, .
 \end{equation*}
Moreover, if $T \subset \R^2 \sm \overline B_{r}$, then $XY_\e(w_\e,T)^\frac{1}{2} \leq 2\frac{\sqrt{12}}{r} \e^2$. Analogously, if $T \subset \R^2 \sm \overline B_{r}(\overline{x}_\e)$, then $XY_\e(\overline w_\e,T)^\frac{1}{2} \leq 2\frac{\sqrt{12}}{r}\e^2$. 
From the previous inequalities it follows that
\begin{equation*} 
    \begin{split}
        & \sum_{T \in \T_\e(B_{r+2\e}\sm \overline B_{K\e})} XY_\e(w_\e,T)^\frac{1}{2} XY_\e(\overline w_\e,T)^\frac{1}{2} \leq \sum_{T \in \T_\e(B_{r+2\e}\sm \overline B_{K\e})} \hspace{-2em} XY_\e(w_\e,T)^\frac{1}{2}  C \e^2 \\
        & \quad \quad \leq \sum_{T \in \T_\e(B_{r+2\e}\sm \overline B_{K\e})}  \fint_T \frac{C \e^4}{|x| - \e} \, \d x \leq \int_{B_{r+2\e}\sm \overline B_{K\e})} \frac{C \e^2}{|x| - \e} \, \d x   \leq C \e^2 \, .
    \end{split}
\end{equation*}
Analogously,
\begin{equation} \label{eq:I3}
    \sum_{T \in \T_\e(B_{r+2\e} \sm \overline B_{K\e}(\overline{x}_\e))} XY_\e(w_\e,T)^\frac{1}{2} XY_\e(\overline w_\e,T)^\frac{1}{2}  \leq C \e^2 \, .
\end{equation}
Summing~\eqref{eq:I1}--\eqref{eq:I3} we obtain that $I_\e \leq C \e^2$. 
 
It remains to show that $\|\mu_{v_\e}-\mu\|_{\mathrm{flat},\Omega}\to 0$. By the same reasoning as in Step 1 we first obtain that, up to a subsequence, $\|\mu_{v_\e}-(d\delta_0+\overline{d}\delta_{\overline{x}})\|_{\mathrm{flat},\Omega}\to 0$, where we have used that $\overline{x}_\e\to x$ as $\e\to 0$. We are then left to show that $d=\overline{d}=1$. This will be done by localising the argument in Step 1. Namely, letting $\zeta\colon[0,3]\to\R$ be as in Step 1, we choose $r>0$ sufficiently small such that $B_{3r}\cap B_{3r}(\overline{x})=\emptyset$ and we set $\psi(x)\defas\zeta\big(\frac{|x|}{r}\big)$, $\overline{\psi}(x)\defas\zeta\big(\frac{|x-\overline{x}|}{r}\big)$ for every $x\in\R^2$. We let $\overarc{v}_\e$ denote the interpolation of $v_\e$ as in Remark~\ref{rmk:arcwise interpolation} and we set $v(x)\defas\frac{x}{|x|}\odot\frac{x-\overline{x}}{|x-\overline{x}|}=:w(x)\odot \overline{w}(x)$ for every $x\in\R^2\sm\{0,\overline{x}\}$. Thanks to the choice of $r$ and the fact that $\overline{x}_\e\to \overline{x}$, we have that $\overarc{v}_\e\wto v$ in $H^1(A_{r,2r}\cup A_{r,2r}(\overline{x});\R^2)$. In particular, as in Step 1 we deduce that
\begin{equation*}
\langle d\delta_0,\psi\rangle =\lim_{\e\to 0}\langle\mu_{v_\e},\psi\rangle=\frac{1}{\pi}\lim_{\e\to 0}\langle J(\overarc{v}_\e),\psi\rangle=-\frac{1}{\pi}\lim_{\e\to 0}\int_{\A{r}{2r}} \hspace*{-1em} j(\overarc{v}_\e)\cdot\nabla^\perp\psi\dx=-\frac{1}{\pi}\int_{\A{r}{2r}} \hspace*{-1em} j(v)\cdot\nabla^\perp\psi\dx\,.
\end{equation*} 
Moreover, a direct computation yields $j(v)=j(w)+j(\overline{w})$ and $\nabla^\perp\psi(x)=-\frac{x^\perp}{r|x|}$, hence
\begin{equation}\label{eq:d-delta0}
\langle d\delta_0,\psi\rangle=\frac{1}{2\pi}\int_{\A{r}{2r}}\!\frac{1}{r|x|}\dx-\frac{1}{\pi}\int_{\A{r}{2r}}\hspace*{-1em}j(\overline{w})\cdot\nabla^\perp\psi\dx=1-\frac{1}{\pi}\int_{\A{r}{2r}}\hspace*{-1em}j(\overline{w})\cdot\nabla^\perp\psi\dx\,.
\end{equation}
Eventually, the choice of $r>0$ ensures that $\overline{w}\in H^1(\A{r,2r};\S^1)$ with $\deg(\overline{w},\partial B_\rho)=0$ for every $\rho\in[r,2r]$. Since in addition $\nabla^\perp\psi=-\frac{1}{r}\tau_{\partial_{B_\rho}}$ for every $\rho\in[r,2r]$, applying the coarea formula and~\eqref{eq:degree for H1} yields
\begin{equation*}
-\frac{1}{\pi}\int_{\A{r}{2r}}\hspace*{-1em}j(\overline{w})\cdot\nabla^\perp\psi\dx=\int_{r}^{2r}\frac{1}{ r\pi}\int_{\partial B_\rho}j(\overline{w})|_{\partial B_\rho}\cdot\tau_{\partial B_\rho}\dH\, \d\rho=\int_{r}^{2r}\frac{1}{r}\deg(\overline{w},\de B_\rho)\, \d\rho=0\,.
\end{equation*}
Thus, from~\eqref{eq:d-delta0} we deduce that $d=1$. By repeating the argument in~\eqref{eq:d-delta0} with $\langle d\delta_0,\psi\rangle$ replaced by $\langle \overline{d}\delta_{\overline{x}},\overline{\psi}\rangle$ and exchanging the roles of $w$ and $\overline{w}$ we obtain $\overline{d}=1$, hence $\|\mu_{v_\e}-\mu\|_{\mathrm{flat},\Omega}\to 0$, which concludes the proof of the limsup inequality.

Since the case $\mu=\pm\delta_{x_1}\pm\delta_{x_2}$ can be treated similarly, the case $\mu=\sum_{h=1}^N\pm\delta_{x_h}$ now follows by an iterative construction.
\end{step}

\begin{step}{3} (The general case)
The general case follows from Step 2 via a diagonal argument. More in detail, given $\mu=\sum_{h=1}^Nd_h\delta_{x_h}$ with $d_h\in\Z$ and $x_h\in\Omega$ we approximate $\mu$ with a sequence of measures $\mu_n$ which are admissible for Step 2 as follows: For every $n\in\N$ and every $h\in\{1,\ldots, N\}$ we choose $|d_h|$ points $x_{h,n}^1,\ldots, x_{h,n}^{|d_h|}\in B_{\frac{1}{n}}(x_h)$ and we set
\begin{equation*}
\mu_n\defas\sum_{h=1}^N\sum_{m=1}^{|d_h|}\sign(d_h)\delta_{x_{h,n}^m}\,.
\end{equation*}
By construction, $|\mu_n|(\Omega)=\sum_h|d_h|=|\mu|(\Omega)$. Thus, for every $n\in\N$ there exist $u_{\e,n}\in\SF_\e$ and corresponding spin fields $v_{\e,n}\in\SF_\e$ such that $\|\mu_{v_{\e,n}}-\mu_n\|_{\mathrm{flat},\Omega}\to 0$ as $\e\to 0$ and
\begin{equation*}
\limsup_{\e\to 0}\frac{1}{\e^2|\log\e|}E_\e(u_{\e,n},\Omega)\leq 2\sqrt{3}\pi|\mu_n|(\Omega)=2\sqrt{3}\pi|\mu|(\Omega)\,.
\end{equation*}
Thus, since $\|\mu_n-\mu\|_{\mathrm{flat},\Omega}\to 0$ as $n\to+\infty$, a diagonal argument provides us with a sequence $(u_{\e,n(\e)})$ such that $\|\mu_{v_{\e,n(\e)}}-\mu\|_{\mathrm{flat},\Omega}\to 0$ as $\e\to 0$ and~\eqref{est:Gamma-limsup} holds true. This concludes the proof in the general case.
\end{step}

\bigskip

\noindent {\bf Acknowledgments.}  The work of A.\ Bach and M.\ Cicalese was supported by the DFG Collaborative Research Center TRR 109, ``Discretization in Geometry and Dynamics''. G.\ Orlando has received funding from the European Union's Horizon 2020 research and innovation programme under the Marie Sk\l odowska-Curie grant agreement No 792583. The work of L.\ Kreutz was funded by the Deutsche Forschungsgemeinschaft (DFG, German Research Foundation) under Germany's Excellence Strategy EXC 2044 -390685587, Mathematics M\"unster: Dynamics--Geometry--Structure. 

\end{document}